\documentclass[10pt,letterpaper]{article}
\usepackage{amsmath,amssymb,graphicx,bbm}
\usepackage{amsthm,verbatim}
\usepackage{mathrsfs,mathtools}
\usepackage[dvipsnames]{xcolor}
\usepackage{physics}
\usepackage{authblk}
\usepackage{setspace}
\usepackage{subfigure}
\usepackage{bm}
\usepackage{algorithmic}
\usepackage[ruled,vlined]{algorithm2e}
\usepackage{booktabs}
\usepackage[titletoc]{appendix}
\newcommand\independent{\protect\mathpalette{\protect\independenT}{\perp}}
\def\independenT#1#2{\mathrel{\rlap{$#1#2$}\mkern2mu{#1#2}}}

\usepackage{enumitem}
\setlist[enumerate]{leftmargin=.5in}
\setlist[itemize]{leftmargin=.5in}

\usepackage[margin=2.9cm]{geometry}

\newcommand{\bs}[1]{\boldsymbol{#1}}

\usepackage{array,multirow,booktabs}

\usepackage{mathdefs}

\usepackage[backend=bibtex,date=year,giveninits=true,sorting=nyt,natbib=true,maxcitenames=2,maxbibnames=10,url=false,doi=true,backref=false]{biblatex}
\renewbibmacro{in:}{\ifentrytype{article}{}{\printtext{\bibstring{in}\intitlepunct}}}

\usepackage[colorlinks,linkcolor = BrickRed, citecolor=blue]{hyperref} 

\definecolor{brickred}{rgb}{0.42, 0.22, 0.10}

\newtheorem{Th}{Theorem}[section]
\newtheorem{Lemma}[Th]{Lemma}
\newtheorem{Cor}[Th]{Corollary}

\newtheorem{Def}[Th]{Definition}

\newtheorem{Rem}[Th]{Remark}
\newtheorem{?}[Th]{Problem}

\newtheorem{Ass}[Th]{Assumption}

\def\R{{\mathbb R}}

\def\N{{\mathbb N}}
\def\E{{\mathbb E}}
\def\R{{\mathbb R}}

\def\P{{\mathbb P}}

\def\aetc{{\text {aetcd}}}

\def\di{{\text{dir}}}

\def\e{{\varepsilon}}

\def\v{{\mathbf{v}}}

\def\iid{\text{iid}}

\def\ex{{\text{epr}}}
\def\li{{\text{Lip}}}
\def\ext{{\text{ept}}}

\def\opt{{\text{opt}}}
\def\supp{{\text{supp}}}

\DeclareMathOperator*{\argmin}{arg\,min}
\def\bt{{\widehat{\beta}}}
\def\sg{{\widehat{\sigma}}}

\begin{document}
\title{Budget-limited distribution learning in multifidelity problems}
\date{}

\author[1,2]{\small Yiming Xu\thanks{Mailing address: 155 1400 E, Salt Lake City, UT 84112}}
\author[1,2]{\small Akil Narayan}
\affil[1]{\small Department of Mathematics, University of Utah, Salt Lake City, UT, USA\\
       \texttt{yxu@math.utah.edu}}
\affil[2]{\small Scientific Computing and Imaging Institute, University of Utah, Salt Lake City, UT, USA\\
       \texttt{akil@sci.utah.edu}}
\renewcommand\Authands{and}
  
\maketitle

\begin{abstract}
Multifidelity methods are widely used for estimating quantities of interest (QoI) in computational science by employing numerical simulations of differing costs and accuracies.
Many methods approximate numerical-valued statistics that represent only limited information, e.g., scalar statistics, about the QoI. 
Further quantification of uncertainty, e.g., for risk assessment, failure probabilities, or confidence intervals, requires estimation of the full distributions.
In this paper, we generalize the ideas in [Xu et al., \emph{SIAM J. Sci. Comput.} 44.1 (2022),
A150–A175] to develop a multifidelity method that approximates the full distribution of scalar-valued QoI. The main advantage of our approach compared to alternative methods is that we require no particular relationships among the high and lower-fidelity models (e.g. model hierarchy), and we do not assume any knowledge of model statistics including correlations and other cross-model statistics before the procedure starts. 
Under suitable assumptions in the framework above, we achieve provable 1-Wasserstein metric convergence of an algorithmically constructed distributional emulator via an exploration-exploitation strategy.
We also prove that crucial policy actions taken by our algorithm are budget-asymptotically optimal.
Numerical experiments are provided to support our theoretical analysis.
\end{abstract}


\section{Introduction}

Estimation of output QoI from complex and large-scale simulations is an important task in many areas of computational science.  
A concrete example is a forward uncertainty quantification setup, where the QoI is an output of a physical system subject to modeled randomness/uncertainty, and the goal is to identify the typical behavior of the QoI by computing its expectation. 
A universal solution for this specific task is through Monte Carlo (MC) simulation \cite{Hammersley64}, which in practice requires many repeated evaluations of an accurate forward model, and can be computationally infeasible for expensive models. 

A modern collection of approaches that address this computational challenge is the suite of multifidelity methods \cite{Peherstorfer_2018_survey}. 
Instead of operating on a single (high-fidelity) model alone, multifidelity methods combine several models of different accuracies and costs to accelerate computation. Low-fidelity models are less accurate but also inexpensive simulations, for example, resulting from numerical solvers of parametric partial differential equations (PDEs) using coarse discretizations.
More generally, the lower fidelity models used in multifidelity methods arise from simplification or reduction of the high-fidelity model and thus are cheaper but less accurate. 
However, they often contain information that, if utilized properly, can contribute to characterizing QoI.

A prototypical example of a multifidelity method is the multilevel approach \cite{giles2008multilevel,Peherstorfer_2016,Gorodetsky_2020}, which approximates the expectation of a scalar-valued QoI given by the high-fidelity model. 
Leveraging a telescoping sum using hierarchical models, multilevel estimators make use of cross-model correlations to attain a smaller variance compared to a single-model MC estimator.
Recent work has introduced a more general perspective for many existing methods within the multilevel framework, and provided a way to realize the optimal variance reduction among all linear unbiased estimators \cite{Schaden_2020, schaden2020asymptotic}. 
Multilevel estimators are considered universal in the sense that they rely only on the covariance information of models, which is used as an input for the estimator construction.   
A similar approach has recently been developed in \cite{xu2021bandit} that assumes only a linear model assumption but no a priori knowledge of covariance statistics. 

Since multifidelity methods have been so successfully applied to the parametric estimation of QoI, it is natural to ask if it is possible to extend the same technique to also characterize their distributions or equivalent statistics such as characteristic functions. 
This question has been studied in recent works \cite{Giles_2015, giles2017adaptive, Lu_2016, Krumscheid_2018}.
The major application scenario under consideration is hierarchical models where special relationships between pairs of models of different fidelities are leveraged. 
For more general non-hierarchical multifidelity setups, paradigms for efficiently learning distributions of QoI are absent to our knowledge, which is a gap that this paper seeks to fill. 

Learning a full distribution is possible when independent and identically distributed (i.i.d.) samples are available.
Given enough samples, universal approaches using empirical cumulative distribution functions (CDFs) can be employed to estimate the true distribution. 
However, the general nonparametric nature of this approach is balanced by the slow convergence rates, which considerably limits the direct usage in applications where sampling is costly. In this paper, we employ such non-parametric estimators but ameliorate the cost using multifidelity strategies. Alternative distributional models in statistics are parametric, limiting the space of expressible distributions but enabling the use of classical tools such as the maximum likelihood estimation. However, realistic models are often so complex that it is difficult to identify an appropriately expressive parametric family, so we focus on a non-parametric strategy.

\subsection{Contributions of this paper}\label{ssec:contributions}

Our approach is based on the Explore-Then-Commit (ETC) algorithm in bandit learning \cite{lattimore2020bandit, bubeck2012regret}, 
which partitions computational procedures into an \textit{exploration} phase, where models are sampled to learn information about their relationships, and an \textit{exploitation} phase, where the learned information guides the development and execution of a strategy that builds a predictor. Using such exploration/exploitation ideas to design adaptive algorithms in multifidelity estimation is not completely new; see \cite{peherstorfer2019multifidelity} for a procedure that estimates the mean, and \cite{farcas2020context} for methods that build reduced models in the multifidelity context. These methods and others like them make substantially stronger assumptions than we do; in particular assumptions regarding hierarchical relationships between models and certain \textit{a priori} knowledge of correlations and behavior of model costs\footnote{For example, requiring that models with higher correlation relative to the high-fidelity model should also incur higher cost is one such model cost behavior.} is common and ubiquitous. Our approach and setup are more relevant to \cite{xu2021bandit}, where no model relationships are assumed (in particular, no hierarchy need to be provided), no knowledge of correlations is provided, and no particular behavior of model costs is required.
In the initial exploration phase, we learn about interactions between the high and lower fidelity models and construct a linear regression emulator for the high-fidelity output using a selection of inexpensive low-fidelity models. The exploration phase terminates after an adaptively-identified budget investment and is followed by the exploitation phase that expends the remaining budget to construct an empirical estimator for the regression emulator that approximates the full distribution of the high-fidelity model output. 

Our approach, therefore, leverages models of different fidelities and costs as well as various statistical procedures to produce an efficient estimator for the unknown distribution of the high-fidelity QoI, which would be difficult and costly to estimate directly.
Our procedure does not require any hierarchy or relationships between models and does not require \textit{a priori} knowledge of any model or cross-model statistics. 

From the statistical point of view, our approach could be viewed as ``semi-parametric'' (though this terminology should not be conflated with notions of semi-parametric regression analyses \cite{Robinson_1988}). 
That is, while we do not prescribe a target family of parametric distributions for the high-fidelity model $Y$ and we learn a non-parametric distribution for a regression residual discrepancy term, we do express the distribution of $Y$ parametrically with respect to the (unknown) distributions of low-fidelity models $X_i$, $i \in \{1, \ldots, n\}$.

In summary, our contributions in this article are twofold:

\begin{itemize}[itemsep=0pt]
  \item We introduce an adaptive sequential decision-making algorithm, ``AETC-d'', Algorithm \ref{alg:aETC}, that produces an estimator for the full distribution of a high-fidelity output given a prescribed budget. This estimator is built from low-fidelity models that have no prescribed hierarchy or relationships and the algorithm proceeds without initial knowledge of any statistics. (Cf. the approaches in \cite{Giles_2015, giles2017adaptive, Lu_2016, Krumscheid_2018,peherstorfer2019multifidelity} that require more restrictive assumptions.)
  \item We prove almost sure convergence of the algorithm-produced distributional estimator in the mean 1-Wasserstein metric for a large budget. (See Theorem \ref{optimality} and Corollary \ref{cort}.) We also show almost sure optimality guarantees for particular exploration decisions made by the adaptive algorithm. (Cf. the algorithm in \cite{xu2021bandit} which is less efficient in exploration.)
\end{itemize}

We provide numerical examples that establish the efficacy of our approach. The strength of our theoretical guarantees (full distributional convergence with initial ignorance about model relationships or statistics) comes at the cost of certain technical assumptions. These are concretely described in Section \ref{ssec:assumptions}, with the most stringent one being Assumption \ref{ass:eS-cond}: that the conditional expectation of the high-fidelity model output on the low-fidelity model outputs can be written as a linear function of the latter. However, even in cases when this assumption is violated, we present a concrete strategy that empirically ameliorates this model misspecification; see section \ref{sc}.

We emphasize again that our methodology does \textit{not} require any particular knowledge of the models, i.e., hierarchical/nested structure or specific coupling assumptions, to guarantee convergence. The implementation requires only the identification of a trusted high-fidelity model, the ability to query the models themselves, and the cost of sampling each model relative to the cost of sampling the high-fidelity model.
Our general framework is similar to \cite{xu2021bandit}, but the learning objectives (and hence also the ultimate algorithm and resulting theory) are quite different.

The rest of the paper is organized as follows. 
In Section \ref{ps}, we set up the budget-limited distribution learning problem, and technically describe our assumptions and main results.
In Section \ref{sec:dis}, we briefly review results concerning the convergence of empirical measures under Wasserstein metrics and provide a few technical results for later use. 
In Section \ref{LR:analysis}, we propose an exploration-exploitation strategy for distribution learning and derive an asymptotically informative upper bound for the mean $1$-Wasserstein error of the estimator. 
We then utilize this upper bound in Section \ref{sec:alg} to devise an efficient adaptive algorithm, AETC-d, and establish a trajectory-wise optimality result for it. 
In Section \ref{sec:num}, we provide a detailed numerical study of the AETC-d algorithm, investigating consistency, model misspecification, and optimality of exploration rates. 
In Section \ref{concl}, we conclude by summarizing the main results of the paper.

\section{Problem setup}\label{ps}
\subsection{Notation}
Let $Y, X_1, \ldots, X_n\in\R$ be scalar-valued random outputs associated with the high-fidelity model and $n$ low-fidelity surrogates, respectively.  
Let $c_0$ and $c_i$, $i\in [n]: = \{1, \cdots, n\}$, be the respective cost of sampling $Y$ and $X_i$.  
The costs are assumed known and deterministic. 
No additional assumptions about the accuracy or costs of $X_i$ relative to those of $X_{i+1}$ are assumed in the following discussion. In particular, the index $i$ does not represent an ordering based on cost, accuracy, or hierarchy. 

Many recent advances in multifidelity methods center around the efficient estimation of $\E[Y]$ \cite{giles2008multilevel, Peherstorfer_2016, Gorodetsky_2020, Schaden_2020, schaden2020asymptotic, xu2021bandit}. 
Under appropriate correlation and cost conditions and for a fixed budget, the estimators from these methods are significantly more accurate than the classical MC estimator using i.i.d. samples of $Y$.
For some applications, however, obtaining only precise estimates for parameter means is not sufficient.
For instance, when estimating robust statistics such as the median, when building confidence intervals to quantify estimation uncertainty, or characterizing failure events and estimation of their corresponding probabilities, the full distribution of a QoI containing the complete information of modeled randomness is often required.
Letting $F_Y(y) = \P(Y\leq y)$ be the CDF of $Y$,
we wish to find an efficient estimate for $F_Y$ instead of only certain functionals of it. 
To make use of the low-fidelity models to this end, assumptions on cross-model correlations are not enough for this purpose. 
As opposed to imposing strong hierarchical assumptions on the models,
we introduce an alternative parametric assumption on the relationship between $X_i, i\in [n]$ and $Y$ that allows efficient estimation of $F_Y(y)$ through the $X_i$'s.

\subsection{Linear regression}\label{lr}
A simple yet useful assumption relating $X_i$ and $Y$ is through linear regression. 
For any $S\subseteq  [n]$, we assume that
\begin{align}
&Y = X_S^\top \beta_S + \e_S&X_S = \left(1, (X_i)_{i\in [S]}\right)^\top ,\label{1}
\end{align}
where $X_S$ includes the intercept as the first component, and $\e_S$ is a centered random variable with variance $\sigma_S^2$ (i.e. $\e_S\sim (0, \sigma_S^2))$ that is \emph{independent} of $X_S$.
Note that a similar linear model decomposition as \eqref{1} requiring $\e_S$ only be uncorrelated with $X_S$ always holds provided $Y$ and $X_S$ have bounded second moments, but independence requires more assumptions that we will
describe further in Section \ref{ssec:assumptions}.

\color{black}

The model \eqref{1} provides a way to simulate $Y$ using only samples of $X_S$ and a noise generator for $\e_S$.
Under a fixed budget, sampling from $X_S$ and $\e_S$ may have a lower cost than sampling $Y$ alone under appropriate cost assumptions, potentially leading to an estimator for $F_Y$ with better accuracy. This is the core idea we propose and investigate.

\color{black}

\subsection{Assumptions}\label{ssec:assumptions}

This section codifies some particular assumptions that we make. None of these assumptions are required to employ our main algorithm, but our theoretical convergence guarantees do require all these assumptions. We begin by stating two assumptions that together fully characterize \eqref{1}:

\begin{Ass}\label{ass:eS-cond}
  For every $S \subseteq [n]$, the $X_S$-conditional expectation of $Y$, $\E[Y|X_S]$, is a linear function of $X_S$.
\end{Ass}

\begin{Ass}\label{ass:eS-independent}
  For every $S \subseteq [n]$, $\e_S$ is independent of $X_S$.
\end{Ass}

Although Assumption \ref{ass:eS-cond} initially appears strong, it is not overly restrictive since nonlinear functions of $X_S$ can be added as additional regressors without changing the analysis. 
For example, linear regression of $Y$ on the covariates $X_S$ could be ``extended'' to a more general regression on the covariates $X_S$ as well as quadratic covariates $\mathrm{vec}(X_S  X_S^\top)$.
Since $\E[Y|X_S]$ is a measurable function of $X_S$, it is possible to include a sufficient number of nonlinear terms of $X_S$ to achieve a good approximation of $\E[Y|X_S]$ under appropriate regularity assumptions. 
Nevertheless, identifying an optimal set of such regressors is often problem-dependent, and is not in the scope of this paper. 

Assumption \ref{ass:eS-independent} is more difficult to inspect in practice.
In the context of PDE simulations, if we take $Y$ and $X_S$ to be a collection of functions of discrete solutions associated with, e.g., spatial refinement of a mesh, then $\e_S$ is unlikely to be independent of $X_S$. 
However, the magnitude of $\e_S$ is often relatively small thus exerting little impact on the practical usage of the developed procedure in Section \ref{sec:alg}.  
In general, one can adopt a less restrictive relationship between $\e_S$ and $X_S$ (e.g. $\e_S$ and $X_S$ are uncorrelated) but that would make both simulating $Y$ and the analysis intractable since correlation gives information only about second moments and not the full distribution.
We provide a careful empirical examination of both Assumptions \ref{ass:eS-cond} and \ref{ass:eS-independent} in practice in the numerical results section. 

The remaining assumptions we make are relatively mild, holding for a large class of problems.
\begin{Ass}\label{a1}
  The uncentered second-moment matrix $\Lambda \coloneqq \E [X X^\top]$ is invertible.
\end{Ass}
Assumption \ref{a1} is equivalent to the rather mild requirements that the second moments of $X$ exist, and that there is an event set of positive probability over which the components of $X$ are linearly independent.
Pairing this with the Cauchy interlacing theorem for eigenvalues results in the actual technical statement that we utilize:
\begin{align*}
  \textrm{Assumption \ref{a1}} \Longleftrightarrow \textrm{For every } S\subseteq  [n],\; \Lambda_S: =\E[X_SX_S^\top] \textrm{ is invertible}.
\end{align*}

Our final two assumptions are more technical but are easily satisfied in practice, which we will motivate after presenting the assumptions. First, a random vector $Z\in\R^k$ is a \emph{jointly sub-exponential} random vector, if $\sup_{\|a\|_2=1}\|a^\top Z\|_{\psi_1}<\infty$, where $\|\cdot \|_{\psi_1}$ is the $1$-Orlicz norm \cite{Vershynin_2018}. Equivalently, $Z$ is jointly sub-exponential if there exists $C>0$ such that
\begin{align}
  \sup_{\|a\|_2=1}\P\left(|a^\top Z|>z\right)&\leq 2\exp(-z/C)&\forall z&\geq 0.\label{sub-exp}
\end{align}
Our fourth assumption can now be articulated.
\begin{Ass}\label{a2}
$X = (X_1, \cdots, X_n)$ is jointly sub-exponential in the sense of \eqref{sub-exp}. Moreover, denoting by $F_a(x)$ the CDF of $a^\top X$ for $a\in\R^n$, we assume that $C_\li:=\sup_{\|a\|_2=1}\|F_a\|_\li<\infty$. 
\end{Ass}
We will motivate the sub-exponential assumption later in this section. The $a$-uniform Lipschitz assumption precludes, for example, cases when $F_X$, the multivariate distribution function of $X$, has jump discontinuities in any direction $a$. If $F_X$ is Lipschitz, then this implies the weaker Lipschitz condition above. Our final assumption is as follows:
\begin{Ass}\label{a3}
There is a universal constant $C$ such that, for any $S\subseteq  [n]$, the noise satisfies $|\e_S| \leq C$ almost surely.
\end{Ass}
One specialized situation under which the sub-exponential part in Assumption \ref{a2} and Assumption \ref{a3} are satisfied is when $Y$ and all the components of $X$ are uniformly bounded by a constant almost surely. In this case, setting $\beta_S$ in \eqref{1} as any finite coefficients (e.g. those corresponding to the best linear unbiased estimator) ensures that $\e_S \coloneqq Y - X_S^\top \beta_S$ is also bounded uniformly by a constant, hence satisfying Assumption \ref{a3}. 
This boundedness condition immediately implies that $X$ is jointly sub-exponential, hence satisfying part of Assumption \ref{a2}.

The specialized situation above, i.e., the conditions that $X$ and $Y$ are bounded by a constant almost surely, is quite natural in the setting of well-posed parametric PDEs. For simplicity, we describe the argument for $Y$ only. Let $\mathcal{S} : \mathcal{K} \rightarrow V$ be the solution operator of a parametric PDE, mapping values on some compact subset $\mathcal{K}$ of a Banach space to another Banach space $V$. In a prototypical elliptic PDE setup, $\mathcal{K}$ is a set of diffusion coefficient functions (say uniformly bounded above and below), and $V$ is the Sobolev space $H^1_0(D)$ for some spatial domain $D$. We define the QoI $Y : \mathcal{K} \rightarrow \R$ as $Y(k) = \mathcal{L}(\mathcal{S}(k))$, $k \in \mathcal{K}$, where $\mathcal{L} \in V^\ast$ is an element in the dual space of $V$. The randomness in $Y$ arises from placing a probability measure on $\mathcal{K}$. A common assumption in these settings is the uniform boundedness assumption \cite[Equation (2.20)]{cohen_approximation_2015} for elliptic problems stating that there is a constant $C'$ such that $\|\mathcal{S}(k)\|_V \leq C'$ for every $k \in \mathcal{K}$. Pairing this with the fact that $\mathcal{L}$ is bounded on $V$ implies the desired almost sure uniform bound $|Y| \leq \|\mathcal{L}\|_{V^\ast} C' \eqqcolon C$. Similar arguments hold in the case when $Y$ is computed from a well-posed discretization of $\mathcal{S}$ and $\mathcal{L}$.

In summary, Assumptions \ref{ass:eS-cond} and \ref{ass:eS-independent} are a crucial necessity for our theoretical results, but we show in practice in the numerical results section how violation of either can be ameliorated. The remaining Assumptions \ref{a1}, \ref{a2}, and \ref{a3} hold in many reasonable practical settings as described above.

\subsection{Algorithm overview}\label{ssec:algorithm}
In this section, we summarize the high-level algorithmic idea of this paper. For clarity of presentation, first suppose that \eqref{1} holds and $(\beta_S, \e_S)$ are known \emph{a priori}. Then for some fixed $S\subseteq  [n]$, one may expend a given and fixed total budget $B > 0$ to sample $X_S$ and build an empirical CDF estimator for $F_Y$:
\begin{align}
\frac{1}{N}\sum_{i\in [N]}\bm{1}_{Z_i\leq y}\label{lr:oracle},
\end{align}
where $Z_i$ are i.i.d. random variables sampled according to $X^\top _S\beta_S +\e_S$. That is, one first samples $X_S$ to compute $X^\top _S\beta_S$, and then augments it with an independent noise $\e_S$. Of course, the number of affordable samples $N$ under this model is,
\begin{align}
&N = \left\lfloor\frac{B}{c_\ext(S)}\right\rfloor& c_\ext(S) = \sum_{i\in S}c_i.\label{lr:s:oracle}
\end{align}
As opposed to the direct construction of the empirical CDF from i.i.d. samples of $Y$, the emulator \eqref{lr:oracle} will admit a much larger sampling rate under a fixed budget if $c_\ext(S)\ll c_0$, which can substantially accelerate convergence. Of course, our presentation above elides the real practical challenges of this approach: (i) oracle knowledge of $\beta_S$ is unavailable, (ii) the distribution of $\e_S$ is unknown, and (iii) the ``optimal'' model subset $S$ is unknown. The algorithm we develop overcomes these challenges and accomplishes the construction described above via two phases:

\begin{itemize}[itemsep=0pt,itemindent=1.5cm]
  \item[(Exploration)] A portion of the budget $B$ is expended to explore $X$ and $Y$ in order to learn about the regression coefficients $\beta_S$ and noise $\e_S$, and also to analyze cost-benefit tradeoffs for each model $S$. The budget expended in this phase is adaptively determined.
  \item[(Exploitation)] The remaining budget is expended over a computed optimal model $S^\ast$ identified at the end of Exploration. The distributional estimator for $Y$ is computed as the empirical CDF of $X_{S^\ast}^\top \beta_{S^\ast} + \e_{S^\ast}$ built over the exploitation samples.  We simulate $\e_{S^\ast}$ by bootstrapping the empirical residuals $Y - X_{S^\ast}^\top \beta_{S^\ast}$ from exploration data; see \eqref{boots0} and \eqref{boots} for a concrete description.
\end{itemize}

In philosophy, some of the challenges described above (ignorance of $S$, $\beta_S$) are addressed by the bandit-learning approach proposed in \cite{xu2021bandit}. Indeed, our exploration-exploitation metamodel here is essentially identical to \cite{xu2021bandit}. However, one main advance of this paper is the development of a new loss function to guarantee convergence in a probability metric, whereas direct usage of the algorithm in \cite{xu2021bandit} guarantees only convergence of an estimator for the single-statistic mean.

\subsection{Main results}
We state our main results in more technical terms compared to the descriptions in Section \ref{ssec:contributions}:
\begin{itemize}[itemsep=0pt,itemindent=1.5cm]
  \item[Algorithm \ref{alg:aETC}:] We develop an exploration-exploitation algorithm based on ideas from bandit learning that produces an estimator for the full distribution of $Y$. The algorithm requires as input only the ability to query the models $(X,Y)$, knowledge of the model costs $c_i$, $i \in \{0\}\cup[n]$, and an available computational budget $B$.
  \item[Theorem \ref{optimality}:] Using Algorithm \ref{alg:aETC}, we show that almost surely for a large budget $B$, the computationally chosen model $S^\ast$ coincides with the optimal model that would be identified with access to oracle knowledge. We show similar asymptotic optimality of the resources expended during the exploration stage.
  \item[Corollary \ref{cort}:] We show that Algorithm \ref{alg:aETC} produces a computational emulator for the distribution of $Y$ that budget-asymptotically converges almost surely to the true distribution in the 1-Wasserstein metric.
\end{itemize}
Finally, our numerical results showcase how Algorithm \ref{alg:aETC} outperforms the only competitor we are aware of in this general non-hierarchical multifidelity context: the empirical CDF estimator for $Y$ by expending the full budget $B$ on samples of $Y$.

\section{Distribution metrics and related results}\label{sec:dis}

Before directly addressing the multifidelity problem, this section provides some necessary discussion regarding distribution metrics, empirical measures, and error bounds.

\subsection{Distance between distributions}
The goal of distribution learning is to approximate a probability measure on $\R$. 
In this section, we discuss metrics to measure the discrepancy between $F_Y$ and an estimated distribution. 
For two Borel probability measures $\mu, \nu$ on $\R$, numerous metrics are available to measure discrepancy \cite{Gibbs_2002}, such as the Kolmogorov distance, Wasserstein distances, Kullback--Leibler (KL) divergence, etc. 
In this article, we will focus on Wasserstein metrics, which are defined below.

\begin{Def}[$p$-th Wasserstein distance]
Let $1\leq p< \infty$. The $p$-th Wasserstein distance between two Borel probability measures $\mu, \nu$ on $\R$ is defined as 
\begin{align*}
&W_p(\mu, \nu) = \inf_{\pi}\E_\pi\left[|x-y|^p\right]^{1/p},
\end{align*}
where infimum is taken over all Borel probability measures $\pi$ on $\R^2$ with marginals satisfying $\pi_x \equiv\mu, \pi_y \equiv\nu$. 
\end{Def}

Intuitively, $W_p^p(\mu,\nu)$ corresponds to the minimal amount of work needed to transform $\mu$ to $\nu$, with the cost function given by the $p$-th power of the moving distance. Hence it is frequently labeled as the optimal transport distance. 
A comprehensive study of the subject can be found in \cite{Villani_2003}.  
In the following discussion we are mostly concerned with $p = 1$.

Wasserstein distances are hard to compute in general. But when the metric space is the real line equipped with the Borel algebra, explicit formulas exist using inverse CDFs \cite{cambanis1976inequalities}:
\begin{align*}
W^p_p(\mu, \nu) = \int_0^1\left|F_\mu^{-1}(t)-F_\nu^{-1}(t)\right|^pdt,
\end{align*}  
where $F_\mu^{-1}(t) := \inf\{x\in\R: F_\mu(x)\geq t\}$ is the inverse CDF of $\mu$. 
When $p=1$, $W_1(\mu, \nu)$ is the $L^1(\R)$-norm of $F_\mu(t)-F_\nu(t)$:
\begin{align}
W_1(\mu, \nu) = \int_0^1\left|F_\mu^{-1}(t)-F_\nu^{-1}(t)\right|dt = \int_\R\left|F_\mu(t)-F_\nu(t)\right|dt.\label{W1-rep}
\end{align}  
 
A classical result in optimal transport is the Kantorovich--Rubinstein duality \cite{Villani_2003}, which provides an alternative characterization for $W_1(\mu,\nu)$ using Lipschitz test functions: 
\begin{align}
&W_1(\mu,\nu) = \sup_{\|f\|_\li\leq 1}\left |\int f d\mu - \int f d\nu\right|,\label{dual}
\end{align}
where $\|\cdot \|_\li$ is the Lipschitz constant. 
As a consequence, $W_1(\mu, \nu)$ provides an upper bound for the difference between all linear functionals of $\mu$ and $\nu$ with uniformly bounded Lipschitz coefficients. 

Moreover, assuming $F_\mu$ has a uniformly bounded density $f_\mu$, one can bound the Kolmogorov distance between $\mu$ and $\nu$ by $W_1(\mu, \nu)$ using \eqref{dual} \cite{chatterjee}:
\begin{align}
d_{K}(\mu, \nu) : = \sup_{x\in\R}\left |F_\mu(x)-F_\nu(x)\right |\leq 2\sqrt{\|f_\mu\|_\infty W_1(\mu, \nu)}.\label{dis-Kol}
\end{align}
In the situations of this article, at least one of the measures under comparison has a uniformly bounded density under Assumption \ref{a2}. 
Thus, in what follows we will work with the Wasserstein metrics, i.e., the $W_1$-metric.

\subsection{Convergence of empirical measures}

Our analysis in Section \ref{LR:analysis} relies on sharp convergence rates of one-dimensional empirical measures under the mean $W_p$ metric. 
Given i.i.d. samples of $\mu: Z_1, \ldots, Z_N$, the associated empirical measure is
\begin{align}
\mu_N = \frac{1}{N}\sum_{i\in [N]}\delta_{Z_i},\label{emp}
\end{align}
where $\delta_x$ denotes the Dirac mass at $x$. 
The question is to quantify the average convergence rate for $\E[W^p_p(\mu_N,\mu)]^{1/p}$ at fixed $N$. 
A comprehensive analysis quantifying $\E[W^p_p(\mu_N,\mu)]^{1/p}$ for fixed $N$ can be found in \cite{Bobkov_2019}. 
Here we only collect relevant results to be used later. 

\begin{Lemma}[\cite{Bobkov_2019}]\label{BM}
Suppose $\mu$ is a Borel probability measure on $\R$ with finite $(2+\delta)$-th moment for some $\delta>0$, i.e., $\int |x|^{2+\delta}d\mu(x)<\infty$. Let $\mu_N$ be the empirical measure defined in \eqref{emp}. Then for every $N\geq 1$,
\begin{align}
\frac{J_0(\mu)}{\sqrt{2}}\frac{1}{\sqrt{N}}\leq\E[W_1(\mu_N,\mu)]\leq\frac{J_1(\mu)}{\sqrt{N}},\label{upc}
\end{align} 
where 
\begin{align}
&J_0(\mu): = \int_\R F_{\mu}(x)(1-F_{\mu}(x)) dx&J_1(\mu): = \int_\R \sqrt{F_{\mu}(x)(1-F_{\mu}(x))}dx.\label{J}
\end{align}
\end{Lemma}
The moment assumption on $\mu$ ensures that $J_1(\mu)<\infty$ so that the upper bound in \eqref{upc} is nonvacuous. 
We prove next that \eqref{upc} is tight when $\mu$ is fast-decaying and has a bounded density, i.e., $J_0(\mu)$ and $J_1(\mu)$ are of a similar order. 

\begin{Lemma}\label{mylemo}
If there exist $z\in\R$ and constants $C>1$ and $r\geq 4$ such that
\begin{align}
\min\{F_\mu(x+z), 1-F_\mu(x+z)\}&\leq C|x|^{-r}& x\neq 0\label{241}\\
 |F_\mu'(x)|&\leq \sqrt{C} \label{242}
\end{align}
then $J_1(\mu)\leq 21CJ_0(\mu)$. 
\end{Lemma}

\begin{proof}
Since $J_0$ and $J_1$ are invariant under shifts, without loss of generality we assume $z=0$.
We first upper bound $J_1(F_\mu)$ as follows:
\begin{align*}
J_1(F_\mu) &= \int_\R \sqrt{F_\mu(x)(1-F_\mu(x))}dx\stackrel{\eqref{241}}{\leq} \int_\R \sqrt{C}\min\left\{\frac{1}{2},|x|^{-r/2}\right\} dx\leq 2\sqrt{C}\left(\frac{1}{2}+\frac{2}{r-2}\right)\leq 2\sqrt{C}.
\end{align*}
On the other hand, for any $0<\delta<1$, 
\begin{align*}
J_0(F_\mu) = \int_\R F_\mu(x)(1-F_\mu(x))dx&\geq\int_{\delta\leq F_\mu(x)\leq 1-\delta} F_\mu(x)(1-F_\mu(x))dx\\
&\geq\int_{\delta\leq F_\mu(x)\leq 1-\delta} \delta(1-\delta)dx\stackrel{\eqref{242}}{\geq}\frac{\delta(1-\delta)(1-2\delta)}{\sqrt{C}}.
\end{align*}
Taking $\delta = \frac{1}{2}(1-\frac{1}{\sqrt{3}})$ yields $J_0(F_\mu)>0.096/\sqrt{C}$.
This combined with the upper bound on $J_1(F_\mu)$ proves the desired result. 
\end{proof}
Similar nonasymptotic results for $\E[W^p_p(\mu_N,\mu)]^{1/p}$ when $p>1$ can be found in \cite[Corollary 5.10]{Bobkov_2019}. 
In those cases, the sufficient and necessary condition to achieve the optimal convergence rate has a more subtle dependence on the moments of a distribution  \cite[Corollary 6.14]{Bobkov_2019}.  

In later sections, we need to compare $J_1$ of sums of independent random variables.
The following notion is useful for the analysis:
\begin{Def}\label{maxfun}
Let $\mathcal F(\R^d)$ be the space of uniformly bounded functions on $\R^d$. 
For every $r\geq 0$, the $r$-local maximum operator $M_r: \mathcal F(\R^d)\to \mathcal F(\R^d)$ is defined as follows: For any $f\in \mathcal F(\R^d)$, 
\begin{align*}
&M_r[f](x) = \sup_{z\in B(z, r)}f(z)&B(x,r) = \{z: \|x-z\|_2\leq r\}.
\end{align*}
\end{Def}

\begin{Lemma}\label{cute}
Let $R_1 = Z+U$ and $R_2 = Z + V$, where $Z$ and $U/V$ are independent random variables with bounded $(2+\delta)$-th moments for some $\delta>0$.
Suppose $\supp(Z)\subseteq [-r, r]\subset\R$ for some $r>0$. 
For a random variable $X$, denote by $F_X(x)$ the CDF of $X$. 
Then,
\begin{align*}
&|J_1(F_{R_1})-J_1(F_{R_2})|\leq\int_\R M_r[|F_U-F_V|]^{1/2}(x) dx, 
\end{align*}
where $M_r$ is the $r$-local maximum operator in Definition \ref{maxfun}. 
\end{Lemma}

\begin{proof}
By definition,  
\begin{align}
\left|J_1(F_{R_1})-J_1(F_{R_2})\right|&\leq\int_\R\left|\sqrt{F_{R_1}(x)(1-F_{R_1}(x))}-\sqrt{F_{R_2}(x)(1-F_{R_2}(x))}\right| dx\nonumber\\
&\stackrel{}{\leq}\int_\R\sqrt{\left |F_{R_1}(x)(1-F_{R_1}(x))-F_{R_2}(x)(1-F_{R_2}(x))\right|} dx\nonumber\\
& =  \int_\R\sqrt{|F_{R_1}(x)-F_{R_2}(x)||1-(F_{R_1}(x)+F_{R_2}(x))|} dx\nonumber\\
&\stackrel{}{\leq}\int_\R\sqrt{|F_{R_1}(x)-F_{R_2}(x)|} dx.\label{dongge}
\end{align}
Since $Z$ and $U/V$ are independent, 
\begin{align}
F_{R_1}(x) &= \int_\R F_{U}(x-z)dF_X(z) = \int_\R F_{X}(x-z)dF_U(z)\label{convo1}\\
F_{R_2}(x) &= \int_\R F_{V}(x-z)dF_X(z) = \int_\R F_{X}(x-z)dF_V(z)\label{convo2}.
\end{align}
Substituting \eqref{convo1} and \eqref{convo2} into \eqref{dongge} yields the desired result:
\begin{align*}
\int_\R\sqrt{|F_{R_1}(x)-F_{R_2}(x)|} dx&=\ \int_\R\sqrt{\int_\R |F_{U}(x-z)-F_{V}(x-z)| dF_X(z)} dx\\
&\leq\int_\R M_r[|F_{U}-F_{V}|]^{1/2}(x) dx.
\end{align*}

\end{proof}

\subsection{A CDF estimate for marginals of a random vector}

We have seen in \eqref{W1-rep} that if two one-dimensional random variables have ``similar'' CDFs, then they are close in the $W_1$-metric. 
In this section, we prove a type of converse of this statement.
In particular, we show that two close marginals of a jointly sub-exponential vector have similar CDFs. 
  
We recall that $X \in \R^k$ is a jointly sub-exponential random variable if it satisfies \eqref{sub-exp}.
A useful property of sub-exponential random variables is that their maximum grows less than any polynomial rate:

\begin{Lemma}\label{expmax}
Let $\{U_n\}_{n\in\N}$ be a sequence of sub-exponential random variables with uniformly bounded sub-exponential norm, i.e., there exists a constant $C>0$ such that $\P(|U_n|>x)\leq 2\exp(-x/C)$ for all $n$. 
Denote by $U_n^* = \max_{i\in [n]}|U_i|$. 
Then, $U_n^* \lesssim\mathcal O(\log n)$ a.s., where the implicit constant is realization-dependent. 
\end{Lemma}
\begin{proof}
For every $n$ and $t_n>0$, it follows from a union bound that 
\begin{align*}
\P\left(U_n^*\geq t_n\right)\leq\sum_{i\in [n]}\P\left(|U_i|\geq t\right)\leq 2n\exp\left(-\frac{t_n}{C}\right). 
\end{align*}
Setting $t_n = 3C\log n$ and applying the Borel--Cantelli lemma finishes the proof. 
\end{proof}

Now consider two marginals $Z_1, Z_2$ of $X$: 
\begin{align}
&Z_1 = a_1^\top X&Z_2 = a^\top _2X,\label{Z12}
\end{align}
where $a_1, a_2\in\R^k$.  
Denote by $F_1(x)$ and $F_2(x)$ the CDFs of $Z_1$ and $Z_2$, respectively. 
If $\|a_1-a_2\|_2$ is small (which implies that $W_1(F_1, F_2)$ is small by Cauchy-Schwarz inequality), 
then $F_1(x)\approx F_2(x)$ for every fixed $x\in\R$. 
Since \eqref{sub-exp} ensures that the CDFs of both $Z_1$ and $Z_2$ quickly decay to zero as $|x|\to\infty$, the distance between $F_1(x)$ and $F_2(x)$ can be bounded in a uniform sense.

\begin{Lemma}\label{kouniao}
Let $X\in\R^k$ be a sub-exponential random vector satisfying \eqref{sub-exp} for some constant $C=C_1>0$, and $Z_1, Z_2$ be the marginals defined in \eqref{Z12}. 
Denote by
\begin{align*}
&\delta: = \|a_1-a_2\|_2 & C_2 :=\sqrt{\|a_1\|_2^2+\|a_2\|_2^2}<\infty. 
\end{align*}
Suppose the CDFs of $Z_1$ and $Z_2$, $F_1(x)$ and $F_2(x)$, are globally Lipschitz continuous; i.e., 
\begin{align*}
\max\{\|F_1\|_\li, \|F_2\|_\li\}\leq C_3
\end{align*}
for some $C_3<\infty$.   
Then, for any $p\in (0,\infty)$, there exists a constant $K>0$ that depends on $C_1, C_2, C_3$ and $p$ such that the following bounds hold true:
\begin{align}
\|M_r[|F_1-F_2|]\|^p_{L^p(\R)}\leq \begin{cases} 
      K\left(1+2r\right)\left(\delta + \delta^{5p/6}\right)\log\left(\frac{1}{\delta}\right) & 0\leq \delta\leq 1/2 \\
      \\
      8\left(r + \frac{2^{p}C_1C_2}{p}\right)\delta & \delta\geq1/2 
         \end{cases} \label{use}
\end{align}
where $M_r$ is the $r$-local maximum operator in Definition \ref{maxfun}. 
\end{Lemma}

\begin{proof}
Without loss of generality, we assume $F_1(x)\geq F_2(x)$ for all $x\in\R$, otherwise, one can divide into two regimes and discuss them separately.   

We first give a bound for $M_r[|F_1-F_2|](x)$ when $|x|\geq 2r$:
\begin{align}
M_r[|F_1-F_2|](x) = \sup_{z\in B(x, r)}(F_1(z)-F_2(z))&\leq \sup_{z\in B(x, r)}\min\{F_1(z), 1-F_2(z)\}\nonumber\\
&\leq 2\sup_{z\in B(x, r)}\exp\left(-\frac{|z|}{C_1C_2}\right)\nonumber\\
&\leq 2\exp\left(-\frac{|x|}{2C_1C_2}\right).\label{reg1}
\end{align}

When $\delta\geq 1/2$, since $M_r[|F_1-F_2|](x)\leq 1$, we have 
\begin{align*}
\|M_r[|F_1-F_2|]\|^p_{L^p(\R)}&\leq\int_{|x|<2r}1 ~dx + \int_{|x|\geq 2r}2^p\exp\left(-\frac{p|x|}{2C_1C_2}\right)dx\\
&\leq 4r + \frac{2^{p+2}C_1C_2}{p}\\
&\leq 8\left(r + \frac{2^{p}C_1C_2}{p}\right)\delta.
\end{align*}

To obtain the bound when $\delta\leq 1/2$, note that for $t>0$, $\{Z_1\leq x\}\subseteq\{Z_2\leq x+t\}\cup\{Z_2-Z_1>t\}$,
which implies
\begin{align}
M_r[|F_1-F_2|](x)&\leq \sup_{z\in B(x, r)}\left(F_2(z+t)-F_2(z)+\P\left(Z_2-Z_1>t\right)\right)\leq C_3 t + 2\exp\left(-\frac{t}{C_1\delta}\right),\label{reg2}
\end{align}
where the first inequality follows from applying \eqref{sub-exp} to the normalized sub-exponential random variable $(Z_1-Z_2)/\delta$.
Setting $t = C_1\delta\log\left(1/\delta\right)$ yields
\begin{align*}
M_r[|F_1-F_2|](x)\leq C_1C_3\delta\log\left(\frac{1}{\delta}\right) + 2\delta\leq (2+C_1C_3)\delta^{5/6},
\end{align*}
where we used the facts $\log (1/x)\leq x^{-1/6}$ and $x<x^{5/6}$ when $x\leq 1/2$. 
Combining estimates \eqref{reg1} and \eqref{reg2}, we can bound $\|M_r[|F_1-F_2|]\|^p_{L^p(\R)}$ as follows:
Fixing
\begin{align*}
T = \max\left\{2r, \frac{2C_1C_2\log(1/\delta)}{p}\right\},
\end{align*}
we compute 
\begin{align}
\|M_r[|F_1-F_2|]\|^p_{L^p(\R)}&=\ \int_{|x|<T} (2+C_1C_3)^p\delta^{5p/6}dx + \int_{|x|\geq T}2^p\exp\left(-\frac{p|x|}{2C_1C_2}\right)dx\nonumber \\
&\leq 2(2+C_1C_3)^p\delta^{5p/6}\left(2r + \frac{2C_1C_2\log(1/\delta)}{p}\right) + \frac{2^{p+2}C_1C_2\delta}{p}\nonumber\\
&\lesssim \left(1+2r\right)\left(\delta + \delta^{5p/6}\right)\log\left(\frac{1}{\delta}\right).\label{sun1}
\end{align} 
\end{proof}

\section{Empirical CDF estimator under linear regression}\label{LR:analysis}

In this section, we develop a rigorous framework for estimation of $F_Y(y)$ under \eqref{1}. Our framework follows the overall idea presented in Section \ref{ssec:algorithm}.
 
\subsection{Exploration and exploitation}\label{ex-ex}
In this section, we introduce an exploration-exploitation strategy following the ideas in \cite{xu2021bandit}. 
In \eqref{1}, since neither the best parametric model $S$ nor the corresponding coefficients $\beta_S$ is known, it is necessary to expend some effort (budget) to decide on $S$ before committing to \eqref{lr:oracle}.  
We split a given total budget $B > 0$ into two parts, one for \emph{exploration} and the other for \emph{exploitation}. 
In the exploration stage, we collect $m$ independent joint samples of all models $(Y,X)$ to estimate $\beta_S$ and $F_{\e_S}$ (i.e. the CDF of $\e_S$) for every $S\subseteq  [n]$, based on which we decide the best model for exploitation. 
In the exploitation stage, we follow the decision made in the exploration phase and use \eqref{lr:oracle} to construct an estimator for $F_Y$ by plugging in the estimated coefficients. 
Denote 
\begin{align*}
&\textsf{(exploration samples)}\ \ \ \ \ X_{\ex,\ell} := (1, X_{1,\ell}, \cdots, X_{n,\ell}, Y_\ell)^\top &\ell\in [m]\\
&\textsf{(design matrix for $S$)}\ \ \ \ \ Z_S := \begin{bmatrix}
1,(X_{i,1})_{i\in S}\\
\vdots\\
1,(X_{i,m})_{i\in S}
\end{bmatrix}\in\R^{m\times (s+1)}\\
&\textsf{(exploration responses)}\ \ \ \ \ Y_{\ex}  := (Y_1, \cdots, Y_m)^\top ,
\end{align*}
where $\ell$ is the sampling index, and $N_S = \left\lfloor\right (B-c_\ex m)/c_\ext(S)\rfloor$ is the number of affordable samples for exploiting $S$, where $c_\ex = \sum_{i=0}^nc_i$ is the cost for an exploration sample.  
For $m>|S|+1$, $\beta_S$ can be estimated using standard least-squares.  
Assuming that $Z_S$ has full column rank, the estimator for $\beta_S$ is given by
\begin{align}
\bt_S = Z_S^\dagger Y_\ex = \left(Z_S^\top Z_S\right)^{-1}Z_S^\top Y_\ex,\label{dadad}
\end{align}  
The full distribution of $\e_S$ can be estimated using the following plug-in empirical CDF estimator based on exploration residuals:
\begin{align}
\widehat{F}_{\e_S}(y) = \frac{1}{m}\sum_{\ell\in [m]}\bm{1}_{Y_{\ell}-X^\top_{S,\ell}\bt_S\leq y}.\label{boots0}
\end{align}
The resulting empirical estimator from model $S$ is
\begin{align}
\widehat{F}_{Y, S}(y): = \frac{1}{N_S}\sum_{j\in [N_S]}\bm{1}_{A_j\leq y},\label{lr:est}
\end{align}
where $A_j$ are i.i.d. random variables sampled according to
\begin{align}
Y': = X_{S}^\top \bt_S + \widehat{\e}_S,\label{1263}
\end{align}
that is, one first samples $X_S$ to compute $ X_{S}^\top \bt_S$ then adds it with an independently sampled noise from 
\begin{align}
\widehat{\e}_S\sim \widehat{F}_{\e_S}(y).\label{boots}
\end{align}
Note that samples of $A_j$ use estimated statistics and hence do not rely on oracle information.

To measure the average quality of \eqref{lr:est} as an estimator for $F_Y$, we use the following mean $1$-Wasserstein distance as the \emph{loss function}:
\begin{align}
L_S(m): = \E\left[W_1(\widehat{F}_{Y, S}, F_Y)|Z_S\right],\label{W_11}
\end{align}
where the expectation averages out the \emph{randomness in exploitation} as well as the \emph{noise in exploration}. Note that the $L_S(m)$ defined in \eqref{W_11} is still random due to the remaining randomness in $Z_S$, and one could alternatively define it by averaging all the randomness, but this would lead to the term $\E[(Z^\top _SZ_S)^{-1}]$, which is difficult to analyze. 
Thus we will pursue \eqref{W_11} in this article and consider the case when $B\to\infty$.

Obtaining exact asymptotically equivalent expressions for \eqref{W_11} is difficult.
To find computably informative substitutes, we compute sharp estimates for an upper bound of \eqref{W_11} in the next section.

\subsection{Upper bounds}\label{upde}
We cannot directly compute $L_S(m)$ in \eqref{W_11} in an algorithmic setting; our goal in this section is thus to derive a computable upper bound. We start by writing \eqref{W_11} in two parts using the triangle inequality:
\begin{align}
L_S(m)\leq \E\left[W_1(F_{Y}, F_{Y'})|Z_S\right] + \E\left[W_1(\widehat{F}_{Y, S}, F_{Y'})|Z_S\right],\label{W_1tri}
\end{align}
where $F_{Y'}$ is the CDF of $Y'$; see \eqref{1263} for the definition of $Y'$. 

We provide some intuition for the bound \eqref{W_1tri}. 
The first term measures the mean $W_1$ distance between $X_S^\top \beta_S+\e_S$ and $X_S^\top \bt_S+\widehat{\e}_S$, which depends on the accuracy of $\bt_S$ and $\widehat{\e}_S$, hence on the exploration rate $m$.  The second term measures the mean convergence rate of empirical measures, which depends on the exploitation rate $N_S$.  
A good exploration-exploitation strategy (i.e. determination of $m$) will balance these quantities.
We will now produce a computable asymptotic upper bound for \eqref{W_1tri} that can be used to find such an $m$.  

\begin{Lemma}\label{lemma1}
Under Assumptions \ref{ass:eS-cond}-\ref{a3}, and given any $\delta>0$, it holds almost surely that for sufficiently large $m$ and every $B>c_\ex m$, 
\begin{subequations}
\begin{align}
\E\left[W_1(F_{Y}, F_{Y'})|Z_S\right]&\leq \frac{(2+\delta)\sqrt{s+1}\sigma_S+J_1(F_{\e_S})}{\sqrt{m}},\label{bdd1}\\
\E\left[W_1(\widehat{F}_{Y, S}, F_{Y'})|Z_S\right] & \leq (1+\delta)\frac{J_1(F_Y)}{\sqrt{N_S}}.\label{bdd2}
\end{align}
\end{subequations}
\end{Lemma}

The proof of Lemma \ref{lemma1} is quite technical and can be found in Appendix \ref{4.1l}. 
An immediate consequence of Lemma \ref{lemma1} is that $L_S$ can be estimated by a more computable quantity that serves as an asymptotic upper bound. This fact will be used for algorithm design in Section \ref{sec:alg}.
\begin{Th}\label{Thm:bdd}
Assume $B>c_\ex m$. 
Under Assumptions \ref{ass:eS-cond}-\ref{a3}, with probability $1$,  
\begin{align}\label{eq:GS-def}
&\limsup_{B,m\uparrow\infty}\frac{L_S(m)}{G_S(m)}\leq 1&G_S(m) := \sqrt{\frac{k_1(S)}{m}}+\sqrt{\frac{k_2(S)}{B-c_\ex m}},
\end{align}
where 
\begin{align*}
&k_1(S) = \left(2\sqrt{s+1}\sigma_S+J_1(F_{\e_S})\right)^2& k_2(S) = c_{\ext}(S)J^2_1(F_Y).
\end{align*}
\end{Th}
\begin{proof}[Proof of Theorem \ref{Thm:bdd}]
Combining \eqref{W_1tri} and Lemma \ref{lemma1} yields that for any $\delta>0$,
\begin{align*}
&\limsup_{B,m\uparrow\infty}\frac{L_S(m)}{G_S(m)}\leq 1+\delta&a.s.
\end{align*}
The proof is finished by sending $\delta\downarrow 0$.  
\end{proof}

\begin{Rem}
Theorem \ref{Thm:bdd} does not require $B$ and $m$ diverge at a certain rate as long as $B>c_\ex m$ is satisfied to ensure $G_S(m)$ is well-defined. 
This will be made clear from the proof of Lemma \ref{lemma1}, where we shall see that consistency of the parameters only depends on the exploration stage. 
\end{Rem}

\section{Algorithm}\label{sec:alg}
In this section, we first use the asymptotic upper bound $G_S$ from \eqref{eq:GS-def} to analyze the optimal exploration rate for each exploitation choice $S\subseteq  [n]$ in an exploration-exploitation policy, which allows us to find a deterministic strategy that explores and exploits optimally. 
Then we propose an adaptive procedure that, along each trajectory, resembles the best exploration-exploitation policy. 
Integer rounding effects defining $N_S$ are subsequently ignored to simplify analysis. 

\subsection{Optimal exploration based on $G_S$}\label{sec:alg1}
The quantity $G_S(m)$ defined in Theorem \ref{Thm:bdd} provides a computable criterion to evaluate the model $S$ as a simulator for $Y$.  
For fixed $S$, $G_S(m)$ is a strictly convex function of $m$ in the domain, attaining its minimum at 
\begin{align}
m^* (S) = \frac{B}{c_\ex+\left(\frac{c^2_\ex k_2(S)}{k_1(S)}\right)^{1/3}},\label{optm}
\end{align}
with optimum value
\begin{align}
G^*_S: = G_S(m^* (S)) = \frac{\left[(c_\ex k_1(S))^{1/3}+k_2(S)^{1/3}\right]^{3/2}}{\sqrt{B}}\propto \left[(c_\ex k_1(S))^{1/3}+k_2(S)^{1/3}\right]^{3/2}.\label{optval}
\end{align}
\eqref{optval} is the (asymptotic) minimum loss of exploiting model $S$ with optimal exploration rate \eqref{optm}.  
The model with the smallest value $G_S^*$ is considered the optimal model under the ``upper bound criterion'', which is our terminology for using $G_S$ as a criterion.
We assume the optimal model is unique and denoted by $S_\opt$, i.e., 
\begin{align}
S_\opt = \argmin_{S\subseteq  [n]}\left[(c_\ex k_1(S))^{1/3}+k_2(S)^{1/3}\right]^{3/2}.\label{optmod}
\end{align} 
We call the policy that spends $m^* (S_\opt)$ rounds on exploration and then selects model $S_\opt$ for exploitation a \emph{perfect exploration-exploitation policy}. This is similar to the perfect uniform exploration policies in \cite{xu2021bandit} where exact asymptotics of the loss function is used for estimation of first-order statistics of QoI.  

To investigate the performance of a perfect exploration-exploitation policy, we compare it to the empirical CDF estimator for $Y$ using samples of $Y$ only. 
Let $\widehat{F}_{Y, \di}(B)$ denote the empirical CDF of $Y$ with the whole budget devoted to sampling $Y$. 
The number of admissible samples for $\widehat{F}_{Y, \di}(B)$ is $B/c_0$, which, combined with the lower bound in \eqref{upc}, implies the following lower bound for the mean $W_1$ distance between $\widehat{F}_{Y, \di}(B)$ and $F_Y$: For every $B>c_0$, 
\begin{align}
\E\left[W_1\left(\widehat{F}_{Y, \di}(B), F_{Y}\right)\right]\geq \sqrt{\frac{c_0J^2_0(F_{Y})}{2B}}.\label{lbb}
\end{align}
Under the same budget, the average $W_1$ distance between $F_Y$ and the estimator given by the perfect exploration-exploitation policy is asymptotically bounded by $G^*_{S_\opt}$, as guaranteed by Theorem \ref{Thm:bdd}.
This motivates us to introduce the ratio between the lower bound in \eqref{lbb} and $G^*_{S_\opt}$ as a measure for the efficiency of perfect exploration-exploitation policies relative to $\widehat{F}_{Y, \di}(B)$: 
\begin{align}
  \frac{\sqrt{\frac{c_0J^2_0(F_{Y})}{2B}}}{G^*_{S_\opt}} &= \sqrt{\frac{c_0J^2_0(F_{Y})}{2\left[(c_\ex k_1(S_\opt))^{1/3}+k_2(S_\opt)^{1/3}\right]^{3}}}\stackrel{\text{(Jensen)}}{\geq}  \sqrt{\frac{c_0J_0^2(F_Y)}{8(c_\ex k_1(S_\opt)+k_2(S_\opt))}}\label{ratio9}\\
& \geq \frac{1}{\sqrt{8\left(\kappa_0^2 + \kappa_1^2\right)}} \geq \frac{1}{4\max\{\kappa_0, \kappa_1\}}\nonumber,
\end{align}
where
\begin{align*}
&\kappa_0 =\frac{2\sigma_{S_\opt}\sqrt{n+1}+J_1(F_{\e_{S_\opt}})}{J_0(F_Y)}&\kappa_1 =\sqrt{\frac{c_\ext(S_\opt)}{c_0}}\frac{J_1(F_Y)}{J_0(F_Y)}.
\end{align*}
The efficiency ratio \eqref{ratio9} is large if $\max\{\kappa_0, \kappa_1\}$ is small. 
Lemma \ref{mylemo} tells us that, under appropriate assumptions on $Y$, $J_1(F_Y)/J_0(F_Y)$ is of constant order. 
In this case, $\max\{\kappa_0, \kappa_1\}\ll 1$ if $\sigma^2_\opt\ll 1$ and $\sqrt{c_\ext(S_\opt)/c_0}\ll 1$, which corresponds to the scenario where the exploration is efficient and exploitation sampling rate is large.
Under such circumstances, the perfect exploration-exploitation policy is expected to demonstrate a superior performance over the empirical CDF estimator based only on the samples of $Y$.   

\subsection{An adaptive algorithm}\label{sec:alg}
Finding perfect exploration-exploitation policies requires evaluation of \eqref{optm} and \eqref{optval}, which uses oracle information of model statistics such as $\sigma_S^2$, $J_1(F_{\e_S})$ and $J_1(F_Y)$. 
These statistics are not available in practice but can be approximated in an online fashion using exploration data.
At each step $t\geq n+2$ in exploration, we define
\begin{align}
\sg_S^2(t) &= \frac{1}{t-|S|-1}\left\|Y_\ex(t) - Z_S(t)\bt_S(t)\right\|^2_2&\widehat{J}_{Y, 1}(t) = J_1(F_{Y_\ex(t)}(y))\nonumber\\
\widehat{J}_{\e_S, 1}(t) &= J_1(F_{\widehat{\e}_S(t)}(y)),\label{estp}
\end{align}
where the parameter $t$ indicates that estimates/data are based on the first $t$ rounds of exploration, and $F_{Y_\ex(t)}(y)$ is the empirical CDF of $Y$ based on the exploration samples $Y_\ex(t)$.  
Under Assumptions \ref{ass:eS-cond}-\ref{a3}, the consistency of $\sg_S^2(t)$ is well known in the linear model literature; see \cite{Lai_1982} for instance. 
The following lemma shows that both $\widehat{J}_{Y, 1}(t)$ and $\widehat{J}_{\e_S, 1}(t)$ are also consistent. 
\begin{Lemma}\label{mylemo1}
Under Assumptions \ref{ass:eS-cond}-\ref{a3}, $\widehat{J}_{Y, 1}(t)\to J_1(F_Y)$ and $\widehat{J}_{\e_S, 1}(t)\to J_1(F_{\e_S})$ as $t\to\infty$.
\end{Lemma}
\begin{proof}
We only provide the proof for $\widehat{J}_{Y, 1}(t)\to J_1(F_Y)$; the consistency of $\widehat{J}_{\e_S, 1}(t)$ can be shown similarly together with the second part of the proof of Lemma \ref{lemma1}.  

Let $I_t$ be the support of the $t$ exploration samples of $Y$ that are used to produce $F_{Y_\ex(t)}(y)$. 
Assumptions \ref{a2}-\ref{a3} ensure that $Y$ is sub-exponential. 
By Lemma \ref{expmax}, 
\begin{align}
&|I_t|\lesssim\log t &a.s., \label{maxexp}
\end{align}
where the implicit constant is realization-dependent.  

On the other hand, according to the Dvoretzky-Kiefer-Wolfowitz inequality \cite{massart1990tight}, for every $t\in\N$, with probability at least $1-2t^{-2}$, $d_K(F_Y(y), F_{Y_\ex(t)})\leq\sqrt{\log t/t}$, where $d_K$ is the Kolmogorov distance defined in \eqref{dis-Kol}.
This combined with a Borel--Cantelli argument yields that
\begin{align}
&d_K(F_Y(y), F_{Y_\ex(t)})\lesssim\sqrt{\frac{\log t}{t}} &a.s. \label{DKW}
\end{align}

We now consider a trajectory along which both \eqref{maxexp} and \eqref{DKW} hold. 
Define $I_t' = I_t\cup [-\log t, \log t]$. 
In this case, for all large $t$, 
\begin{align}
\left|\widehat{J}_{Y,1}(t)- J_1(F_Y)\right|&\stackrel{\eqref{dongge}}{\leq}\int_{\R}\sqrt{|F_Y(y)-F_{Y_\ex(t)}(y)|}dy\nonumber\\
&\leq\int_{I_t'}\sqrt{|F_Y(y)-F_{Y_\ex(t)}(y)|}dy + \int_{(I'_t)^\complement}\sqrt{|F_Y(y)-F_{Y_\ex(t)}(y)|}dy\nonumber\\
&\lesssim \log t\cdot d_K(F_Y(y), F_{Y_\ex(t)}) +  \int_{I^\complement_t}\sqrt{|F_Y(y)-F_{Y_\ex(t)}(y)|}dy.\label{tb}
\end{align}
As $t\to\infty$, the first term in \eqref{tb} diminishes to $0$ because of \eqref{maxexp} and \eqref{DKW}, and the second term diminishes to $0$ as a result of dominated convergence. 
\end{proof}

Plugging \eqref{estp} into $G_S$ and \eqref{optm} allows us to estimate the optimal loss for each model $S$ at exploration step $t$:
\begin{align}
\rho_S = \widehat{G}_S(\widehat{m}^* (S)\vee t; t),\label{rho}
\end{align}
where 
\begin{align}
\widehat{k}_1(S; t) &= \left(2\sqrt{s+2}\sg_S(t)+\widehat{J}_{\e_S, 1}(t)\right)^2, &\widehat{k}_2(S; t) &= c_{\ext}(S)\widehat{J}^2_{Y,1}(t)\label{kk1}\\
\widehat{G}_S(m; t) &= \sqrt{\frac{\widehat{k}_1(S; t)}{m}}+\sqrt{\frac{\widehat{k}_2(S; t)}{B-c_\ex m}}&\widehat{m}^* (S; t) &= \frac{B}{c_\ex+\left(\frac{c^2_\ex \widehat{k}_2(S; t)}{\widehat{k}_1(S; t)}\right)^{1/3}}\label{kk2}.
\end{align}
The second argument in $\widehat{G}_S(\cdot\ ; \ \cdot)$ denotes the number of samples used to estimate the parameters in \eqref{kk1} and \eqref{kk2}, and the first argument denotes the function variable. 
\eqref{rho} can be used to decide which model is optimal to exploit at time $t$, and the corresponding estimated optimal stopping time $\widehat{m}^* $ will indicate if more exploration is needed.  
The details are given in Algorithm \ref{alg:aETC}, the adaptive Explore-Then-Commit algorithm for multifidelity distribution learning (AETC-d).

\begin{algorithm}[H]
\hspace*{\algorithmicindent} \textbf{Input}: samplers for high-fidelity model $Y$ and $n$ low-fidelity models $X_i$; $B$: total budget; $c_i$: cost parameters (i.e. the exploration cost $c_\ex = \sum_{i=0}^nc_i$ and the exploitation cost $c_S = \sum_{i\in S}c_i$ for $S\subseteq [n]$)\\
    \hspace*{\algorithmicindent} \textbf{Output}: An estimate for $F_Y$
 \begin{algorithmic}[1]
\STATE compute the maximum exploration round $M = \lfloor B/c_{\ex}\rfloor$
\IF{$M\leq n+2$}{
\STATE report ``budget is too small"
}
\ELSE{
\STATE collect $t = (n+2)$ samples for exploration
 \WHILE{$n+2\leq  t\leq M$}
 \FOR{$S\subseteq  [n]$}{
 \STATE{compute $\widehat{k}_1(S; t)$ and $\widehat{k}_2(S; t)$ as \eqref{kk1} using $t$ exploration samples}
 \STATE{compute $\widehat{m}^* (S; t)$ and $\rho_S$ as in \eqref{kk2} and \eqref{rho}}
 }
 \ENDFOR
\STATE {find the estimated optimal model $S^* = \argmin_{S\subseteq  [n]}\rho_S$}
\IF{$\widehat{m}^*(S^*; t)>2t$}
        \STATE take $t$ new exploration samples, and set $t \gets 2t$
    \ELSIF{$t<\widehat{m}^*(S^*; t)\leq 2t$}
    \STATE take $\lceil(\widehat{m}^*(S^*; t)-t)/2\rceil$ new exploration samples, and set $t \gets \lceil\frac{t + \widehat{m}^*(S^*; t)}{2}\rceil$
    \ELSE
        \STATE {exhaust the remaining budget $(B - c_\ex t)$ to draw $N_{S^*}:=\lfloor (B-c_\ex t)/c_{S^*}\rfloor$ i.i.d. joint samples of $X_{S^*}$, $\{X_{\ext, S^*, j}\}_{j\in [N_{S^*}]}$}
        \STATE {transform $X_{\ext, S^*, j}$ to the surrogate emulator samples $A_j:= X_{\ext, S^*, j}^\top\bt_{S^*} +\widehat{\e}_{S^*, j}$, where $X_{\ext, S^*, j}$ includes the intercept term $1$ in its first component, $\bt_{S^*}$ is the estimated least-squares coefficients based on $t$ exploration samples as in \eqref{dadad}, and $\widehat{\e}_{S^*, j}$ are i.i.d. random variables (independent of $X_{\ext, S^*, j}$) generated from the empirical noise simulator $\widehat{F}_{\e_{S^*}}$ that is learned using $t$ exploration samples as in \eqref{boots0}}
        \STATE {compute the exploitation estimator for $F_Y$ as in \eqref{lr:est}: $\widehat{F}_{Y, S^*}(y) = \frac{1}{N_{S^*}}\sum_{j\in [N_{S^*}]}\mathbf 1_{A_j\leq y}$, and set $t \gets M+1$}
     \ENDIF
 \ENDWHILE
 }
 \ENDIF
 \end{algorithmic}
\caption{AETC algorithm for multifidelity distribution learning (AETC-d)} 
\label{alg:aETC}
\end{algorithm}

When $M = \lfloor B/c_{\ex}\rfloor$ is too small, Algorithm \ref{alg:aETC} terminates immediately due to a deficient budget for initial sampling in exploration. This regime is of lesser interest since the full correlation information between the high- and low-fidelity models cannot even be estimated.
Note that compared to the AETC algorithm in \cite{xu2021bandit}, Algorithm \ref{alg:aETC} does not use regularization and takes a more aggressive exploration to accelerate computation.
(The regularization parameters could be added to stabilize the algorithm when the budget is small, but this would result in choosing additional input parameters.)
Our main theoretical results, asymptotic analysis of Algorithm \ref{alg:aETC}, are given in the next section.

\subsection{Asymptotic performance of the AETC-d: optimality and consistency}

\begin{Th}\label{optimality}
Let $m(B)$ and $S(B)$ denote the exploration rate and selected model in Algorithm \ref{alg:aETC}, respectively. 
  Under Assumptions \ref{ass:eS-cond}-\ref{a3}, with probability $1$,  
\begin{align}
\lim_{B\uparrow\infty}\frac{m(B)}{m^*(S_\opt)} &= 1 &\lim_{B\uparrow\infty}S(B) = S_\opt\label{main1},
\end{align}
where $S_\opt$ is defined in \eqref{optmod}. 
\end{Th}

\begin{proof}
Without loss of generality we assume $M = \lfloor B/c_{\ex}\rfloor > n+2$. 
The exploration rate $t$ grows nonlinearly with respect to the index that counts the iterations of the WHILE loop in step 6 of Algorithm \ref{alg:aETC}.
For convenience, we let $q$ denote the loop iteration index, and $t_q$ the corresponding exploration rate, i.e. $t_1 = n+2$. 
Let $q(B)$ be the total iteration steps in Algorithm \ref{alg:aETC}, which is random. 
It follows from the definition that $t_{q(B)} = m(B)$ and 
\begin{align}
&n+2\leq t_q\leq t_{q+1}\leq 2t_q& 1\leq q< q(B).\label{myq}
\end{align} 

Since the desired results in Theorem \ref{optimality} only concern decisions made during exploration, the probability space we will be working on is the product space of all exploration samples. 
Under Assumptions \ref{a1}-\ref{a3} and using Lemma \ref{mylemo1}, $\widehat{k}_1(S; t)>0$ a.s. for all $t\geq n+2$ ($\widehat{k}_1(S; t)=0$ implies that the residuals in the exploration are all the same), and both estimators $\widehat{k}_1(S)$ and $\widehat{k}_2(S)$ are consistent for $k_1(S; t)$ and $k_2(S; t)$, respectively, for all $S\subseteq  [n]$:  
\begin{subequations}\label{eq:myheart}
\begin{align}
\widehat{k}_1(S; t) &= \left(2\sqrt{s+2}\sg_S(t)+\widehat{J}_{\e_S, 1}(t)\right)^2\xrightarrow{t\to\infty} k_1(S) = \left(2\sqrt{s+2}\sigma_S+J_1(F_{\e_S})\right)^2\label{myheart1}\\
\widehat{k}_2(S; t) & = c_{\ext}(S)\widehat{J}^2_{Y,1}(t)\xrightarrow{t\to\infty} k_2(S) = c_{\ext}(S)J^2_1(F_Y).\label{myheart2}
\end{align}
\end{subequations}
Thus, 
\begin{align}
&\underline{\widehat{k}}_1(S):=\inf_{t\geq n+2}\widehat{k}_1(S; t)>0&\forall S\subseteq [n]\ \ \ \ \ a.s.,\label{huaern}
\end{align}
where $\underline{\widehat{k}}_1(S)$ is realization-dependent. 

We claim that, with probability $1$, the adaptive exploration rate $m(B)$ diverges, i.e. $m(B)\to\infty$ as $B\to\infty$. 
The proof of this fact is similar to \cite[Theorem 5.2]{xu2021bandit}, but we give the details for convenience.
If $m(B)$ did not diverge, there would exist a stochastic realization $\omega$ along which \eqref{myheart1} and \eqref{myheart2} hold but $m(B, \omega)<\infty$ as $B\to\infty$, i.e., there exists a $M(\omega)\in\N$ such that $m(B, \omega)\leq M(\omega)<\infty$ for all $B>0$.
Since \eqref{myheart2} holds for $\omega$, there exists $C(\omega)<\infty$ such that
\begin{align}
\sup_t\max_{S\subseteq  [n]}k_2(S; t,\omega)\leq C(\omega). \label{myheart3}
\end{align}
The stopping rule in Algorithm \ref{alg:aETC} implies that
\begin{align}
M(\omega)\geq m(B; \omega)\geq \frac{B}{c_\ex+\left(\frac{c^2_\ex \widehat{k}_2(S^*; t,\omega)}{\widehat{k}_1(S^*; t,\omega)}\right)^{1/3}}\stackrel{\eqref{huaern}, \eqref{myheart3}}{\geq} \frac{B}{c_\ex+\left(\frac{c^2_\ex C(\omega)}{\underline{\widehat{k}}_1(S; \omega)}\right)^{1/3}}\xrightarrow{B\uparrow\infty} \infty,\label{lkj}
\end{align}
where $S^*$ in \eqref{lkj} is the estimated optimal model (which is random) in the $q$-th loop iteration. 
This is a contradiction, which occurs for all $\omega$ such that \eqref{eq:myheart} holds, i.e., the contradiction occurs with probability 1. 
Thus, $m(B; \omega)$ diverges, proving that $m(B)\uparrow\infty$ with probability $1$.  

We now consider any trajectory $\omega$ along which \eqref{myheart1} and \eqref{myheart2} hold and $m(B; \omega)$ diverges. 
We will prove that both results in \eqref{main1} hold for such an $\omega$. 
Fix $\delta<1/2$ sufficiently small. 
Since $S_\opt$ is assumed unique, \eqref{myheart1} and \eqref{myheart2} together with a continuity argument implies that, there exists a sufficiently large $T(\omega)$, such that for all $t\geq T(\omega)$, 
\begin{align}
\max_{(1-\delta)m^* (S_\opt)\leq m\leq (1+\delta)m^* (S_\opt)}{\widehat{G}}_{S_\opt}(m; t)&< \min_{S\subseteq  [n], S\neq S_\opt}\min_{0<m<B/c_\ex}\widehat{G}_S(m; t).\label{112}\\
1-\delta&\leq \frac{\widehat{m}^* (S; t)}{m^* (S; t)}\leq 1+\delta&\forall S\subseteq  [n]\label{113}.
\end{align}
To see why \eqref{112} can be achieved for sufficiently small $\delta$, one may without loss of generality assume $B=1$ (which appears as a multiplicative scaling factor). 
Since $S_\opt$ is unique, the maximum of $G_{S_\opt}(m)$ in a small neighborhood of its minimizer, $m^*(S_\opt)$, is strictly smaller than $\min_{S\subseteq  [n], S\neq S_\opt}\min_{0<m<B/c_\ex}G_S(m)$.
On the other hand, ${\widehat{G}}_{S_\opt}(m; t)$ converges uniformly to $G_{S_\opt}(m)$ in any sufficiently small compact neighborhood around $m^*(S_\opt)$, and $\min_{S\subseteq  [n], S\neq S_\opt}\min_{0<m<B/c_\ex}\widehat{G}_S(m; t)\to\min_{S\subseteq  [n], S\neq S_\opt}\min_{0<m<B/c_\ex}G_S(m)$ as a result of strong consistency. 
This justifies the existence of a small $\delta$ for which \eqref{112} is achievable.

Since $m^*(S)$ scales linearly in $B$ and $m(B;\omega)$ diverges as $B\uparrow\infty$, there exists a sufficiently large $B(\delta;\omega)$ such that for $B>B(\delta;\omega)$, 
\begin{align}
\min_{S\subseteq  [n]}m^*(S)&> 4T(\omega)\label{2345}\\
t_{q(B)} = m(B;\omega) &> 4T(\omega).\label{3456}
\end{align}
  Consider $q'< q(B)$ that satisfies $t_{q'-1}< T(\omega) \leq t_{q'}$. Such a $q'$ always exists due to \eqref{3456}, and satisfies 
\begin{align*}
t_{q'}\stackrel{\eqref{myq}}{\leq} 2t_{q'-1}<2T(\omega)\stackrel{\eqref{2345}}{\leq}\frac{1}{2}\min_{S\subseteq  [n]}m^*(S)\stackrel{\eqref{113}, \delta<1/2}{\leq}\widehat{m}^* (S; t_{q'}).
\end{align*} 
This inequality tells us that in the $q'$-th loop iteration, for all $S\subseteq  [n]$, the corresponding estimated optimal exploration rate is larger than the current exploration rate. 
In this case, 
\begin{align*}
&\rho_S = \widehat{G}_S(\widehat{m}^* (S; t_{q'})\vee t_{q'}; t_{q'}) = \widehat{G}_S(\widehat{m}^* (S; t_{q'}); t_{q'})&S\subseteq  [n].
\end{align*}
This, along with \eqref{112} and \eqref{113}, tells us that $S_\opt$ is the estimated optimal model in the current step, and more exploration is needed.  

To see what $t_{q'+1}$ should be, we consider two separate cases.
If $2t_{q'}\leq \widehat{m}^* (S_\opt; t_{q'})$, then
\begin{align*}
T(\omega)<t_{q'+1} = 2t_{q'}\leq \widehat{m}^* (S_\opt; t_{q'})\leq (1+\delta)m^*(S_\opt),
\end{align*}
which implies 
\begin{align}
(1-\delta)m^*(S_\opt)\stackrel{\eqref{113}}{\leq} t_{q'+1}\vee\widehat{m}^* (S_\opt; t_{q'+1})\leq (1+\delta)m^*(S_\opt).\label{myku}
\end{align}
If $t_{q'}\leq \widehat{m}^* (S_\opt; t_{q'})<2t_{q'}$, then
\begin{align*}
t_{q'+1} = \left\lceil\frac{t_{q'} + \widehat{m}^* (S_\opt; t_{q'})}{2}\right\rceil\leq \widehat{m}^* (S_\opt; t_{q'})\leq (1+\delta)m^*(S_\opt),
\end{align*}
which also implies \eqref{myku}. 
But \eqref{myku} combined with \eqref{112} and \eqref{113} implies that $S_\opt$ is again the estimated optimal model in the $(q'+1)$-th loop iteration. 
Applying the above argument inductively proves $S(B) = S_\opt$, i.e., the second statement in \eqref{main1}.
Note \eqref{myku} holds true until the algorithm terminates, which combined with the termination criteria $t_{q(B)}\geq\widehat{m}^*(S_\opt; t_{q(B)})\geq (1-\delta)m^*(S_\opt)$ implies 
\begin{align*}
1-\delta\leq\frac{m(B)}{m^*(S_\opt)} = \frac{t_{q(B)}}{m^*(S_\opt)}\leq 1+\delta.
\end{align*}
The first part of \eqref{main1} now follows by noting that $\delta$ can be set arbitrarily small. 
\end{proof}

\begin{Rem}
As the budget goes to infinity, when looking at the exploration rate as well as the selected model for exploitation, one finds that almost every policy realization of Algorithm \ref{alg:aETC} is an asymptotically perfect exploration-exploitation policy. This establishes a trajectory-wise optimality result for the AETC-d under the upper bound criterion. Nevertheless, since Algorithm \ref{alg:aETC} determines the exploration rate $m(B)$ in an adaptive fashion, we cannot directly conclude the optimality for the average loss as defined in \eqref{W_11}.
\end{Rem}

\begin{Rem}
By the design of Algorithm \ref{alg:aETC}, for large $B$, the total number of loop iteration steps is $\mathcal O(\log B)$ rather than $\mathcal O(B)$ when increasing $t$ step by step, which considerably increases computational efficiency when dealing with a large budget. The same idea can also be used in the AETC algorithm in \cite{xu2021bandit}.  

\end{Rem}

\color{black}

As a consequence, we have the following consistency result for the CDF estimator produced by the AETC-d algorithm:
\begin{Cor}\label{cort}
Denote the CDF estimator produced by Algorithm \ref{alg:aETC} as $\widehat{F}_{Y,\aetc}(B)$. 
  Under Assumptions \ref{ass:eS-cond}-\ref{a3}, then with probability $1$, 
\begin{align*}
\lim_{B\uparrow\infty}W_1\left(\widehat{F}_{Y,\aetc}(B), F_Y\right) = 0.
\end{align*}
\end{Cor}
\begin{proof}
For fixed $B$, let $N(B)$ be the exploitation sampling rate in Algorithm \ref{alg:aETC}, and denote the parametric model used for exploitation as $Y'(B)$, i.e.,  
\begin{align}
Y'(B) = X_{S(B)}^\top \bt_{S(B)} + \widehat{\e}_{S(B)},\label{345}
\end{align}
where the two terms on the right-hand side of \eqref{345} are independent. 
By Theorem \ref{optimality}, with probability $1$, $S(B) = S_\opt$ for all sufficiently large $B$, and $m(B)/m^*(S_\opt)\to 1$, i.e., $m(B)\to\infty$ and $N(B)\to\infty$ as $B\to\infty$.
Under Assumption \ref{a2}, we apply \eqref{use} with $p=1$ and $r=0$ together with the strong consistency of $\bt_{S_\opt}$ \cite[Theorem 1]{Lai_1982} to conclude that $W_1(X_{S(B)}^\top \bt_{S(B)}, X_{S_\opt}^\top\beta_{S_\opt})\to 0$ as $B\to\infty$ almost surely. 
Adjusting the argument in \eqref{f22}, one can also show that $W_1(\widehat{\e}_{S(B)}, \e_{S_\opt})\to 0$.
By the additivity of the $W_1$ metric, we conclude $W_1(Y'(B), Y)\to 0$. 
(The randomness here only depends on exploration.)  

Now fix a trajectory along which $N(B)\to\infty$ and $W_1(Y'(B), Y)\to 0$ as $B\to\infty$.
Let $\{B_k\}$ be an arbitrary sequence such that $B_k\uparrow\infty$ as $k\uparrow\infty$. 
Since convergence in $W_1$ implies convergence in distribution, $\{F_{Y'(B_k)}\}$ is $\delta$-tight\footnote{A sequence of probability measures $\{P_k\}$ defined on a metric space is called $\delta$-tight if for every $\e>0$, there exist a compact measurable set $K$ and a sequence $\delta_k\downarrow 0$ such that $P_k(K^{\delta_k})>1-\e$ for every $k$, where $K^{\delta_k}: = \{x: \text{dist}(x, K)<\delta_k\}$.} \cite[Section 17.5]{williams1991probability}. 
This observation, combined with the fact that $\widehat{F}_{Y,\aetc}(B_k)$ is an empirical measure of $Y'(B_k)$ consisting of $N(B_k)$ samples, implies that $\widehat{F}_{Y,\aetc}(B_k)$ converges to $F_Y$ in distribution almost surely \cite[Theorem 1]{beran1987convergence}. 
To lift the convergence to $W_1$, it only remains to show $\int |x| d\widehat{F}_{Y,\aetc}(B_k)\to\int |x| dF_Y$ as $k\to\infty$, which can be verified using \eqref{345} and the strong law of large numbers. 
The proof is finished by noting that $\{B_k\}$ is arbitrary.  
\end{proof}

\section{Numerical experiments}\label{sec:num}

In this section, we demonstrate the performance of the AETC-d (Algorithm \ref{alg:aETC}) for multifidelity estimation of univariate distributions.
We will focus on consistency, optimality of exploration rates, and model misspecification. 
Four methods will be considered for estimating $F_Y(y)$: 
\begin{itemize}
 \item (ECDF-Y): Empirical CDF estimator for $F_Y$ based on the samples of $Y$ only.
 \item (AETC-d): Algorithm \ref{alg:aETC}.
 \item (AETC-d-no): A modification of Algorithm \ref{alg:aETC}, where the noise in \eqref{lr:est} is omitted, i.e., we set $\widehat{\e}_S$ to 0.
 \item (AETC-d-q): A modification of Algorithm \ref{alg:aETC} using quantile regression; see Appendix \ref{app5} for a detailed description of the algorithm.   
 \end{itemize}
 The modifications AETC-d-no and AETC-d-q are introduced to empirically investigate plausible alternatives that one may consider (which do not enjoy our theoretical guarantees). 

To evaluate results, we compute and report an empirical mean $W_1$ distance (error) between $F_Y$ and the estimated CDF given by the algorithms over $200$ samples. Since AETC-d produces random estimators due to the exploration phase, the experiment is repeated $100$ times with the $5\%$-$50\%$-$95\%$-quantiles recorded to measure this extra uncertainty. 

We emphasize that direct comparison with most alternative multilevel/multifidelity procedures is not possible here, as our theory and the associated empirical errors reported here are in terms of the $W_1$ error on the full distribution of the output. 

\subsection{Ishigami function}\label{ishigami}
In this example, we investigate the performance of the proposed AETC-d algorithm (Algorithm \ref{alg:aETC}) on a multifidelity algebraic system consisting of Ishigami functions \cite{Ishigami}, which have been widely used as a test model for uncertainty quantification.   
We consider an extended version of the setup in \cite{Qian_2018}. 
The high-fidelity model output corresponds to the following random variable:
\begin{align}
&Y =  \sin Z_1 + a\sin^2 Z_2 + bZ_3^4\sin Z_1 + c\sin^3 Z_4 + d\sin^4 Z_5,
\end{align}
where $a, b, c, d$ are deterministic constants, and $Z_{i}, i\in [5]$ are independent random variables,
\begin{align*}
Z_{1,2, 3, 4, 5}\stackrel{\iid}{\sim}\text{Unif}(-\pi, \pi).
\end{align*}
In the following experiments, we set $a = 5$ and $b = 0.1$.   
For the low-fidelity models, we will consider two different scenarios with different model assumptions. 

\subsubsection{Perfect model assumptions}\label{sssec:pm}
Let $c = 1, d = 0.1$. 
We first consider a synthetic dataset consisting of two low-fidelity models:
\begin{align}
X_1 &= \sin Z_1 + a\sin^2 Z_2 + bZ_3^4\sin Z_1 + c\sin^3 Z_4\label{canada}\\
X_2 &= \sin Z_1 + a\sin^2 Z_2 + bZ_3^4\sin Z_1\nonumber.
\end{align} 
In this case, the linear model assumption \eqref{1} is satisfied, i.e., $\E[Y|X_2] = X_2+c\E[\sin^3 Z_4]+d\E[\sin^4 Z_5]$, $\E[Y|X_1] = \E[Y|X_1, X_2]=X_1+d\E[\sin^4 Z_5]$.  
The correlation between $Y$ and $X_1, X_2$ are approximately $0.999$ and $0.986$, respectively. 
The cost of sampling $Y, X_1$ and $X_2$ are assumed hierarchical, assigned as $(c_0, c_1, c_2) = (1, 0.05, 0.001)$. 

We first compare the performance of AETC-d with ECDF-Y. 
The total budget $B$ in the experiment ranges from $10$ to $10^5$.
The ground truth is taken as an empirical CDF of $Y$ computed using $10^7$ independent samples.
Accuracy results for ECDF-Y and AETC-d are reported in Figure \ref{fig1}(a). 
In this example, AETC-d consistently outperforms ECDF-Y by a factor of around $4$. 
The average estimation error decays to $0$ as the budget goes to infinity, verifying the consistency result in Corollary \ref{cort}.

To inspect the optimality of exploration rates selected by AETC-d, we fix the total budget at $B=10^3$ and consider a deterministic version of Algorithm \ref{alg:aETC} with fixed exploration rate $m$. 
The new algorithm collects $m$ exploration samples to estimate the parametric coefficients and then selects $S_\opt$ for exploitation.
In this example, $S_\opt=\{1\}$, and the maximum exploration rate is $M = \lfloor B/1.051\rfloor = 951$.
We apply the deterministic algorithm and compute the average $W_1$ distance between the resulting estimator and $F_Y$ over $100$ independent experiments for $m=(10, 30, 50, 100, 150, 200, 300, 400, 500, 600)$. 
According to Theorem \ref{Thm:bdd}, for fixed $m$, the average error of the estimator given by the deterministic algorithm is asymptotically bounded by a function of the form 
\begin{align}
&f(m; \alpha_1, \alpha_2) = \frac{\alpha_1}{\sqrt{m}} + \frac{\alpha_2}{\sqrt{\frac{B}{c_\ex}-m}}& 0<m<\frac{B}{c_\ex},\label{fit}
\end{align}
for some constants $\alpha_1, \alpha_2>0$. 
Assuming the analysis in Theorem \ref{Thm:bdd} is tight, we would expect the mean error of the deterministic algorithm, as a function of $m$, to approximately fit a curve of the form \eqref{fit}. 
(Other parametric forms may be used if one obtains a different upper bound for the loss.)
We use a least-squares procedure to obtain such a curve by optimizing over $\alpha_1, \alpha_2$.
We also run AETC-d $100$ times and record the $5\%$-$50\%$-$95\%$ quantiles of the adaptively determined exploration rates.
The results are illustrated in Figure \ref{fig1}(b).  

By comparing the mean errors committed by the deterministic algorithm at various fixed exploration rates, we find that the $5\%$-$50\%$-$95\%$ quantiles of the exploration rate chosen by the AETC-d algorithm, $178$-$191$-$203$, are larger than the optimal trade-off point, which is near $m = 41$. 
Although the difference between the AETC-d and optimal trade-offs is large, it is modest relative to the total budget $B$. 
In particular, the relative over-exploration rate by AETC-d in this example is $(191-41)/951\approx 15.8\%$.
According to \eqref{eq:myheart} in Theorem \ref{optimality}, the over-exploration rate will tend to stabilize when the budget becomes large. 
Such a discrepancy is not unexpected since AETC-d is constructed and optimized with respect to $G_S$ rather than the true loss.
Nevertheless, the loss \textit{value} of the AETC-d algorithm is close to the minimum. This empirically suggests that the upper bound criterion is useful in terms of balancing the exploration and exploitation errors.

\begin{figure}[htbp]
  \centering 
  \subfigure[]{\includegraphics[width=0.25\linewidth]{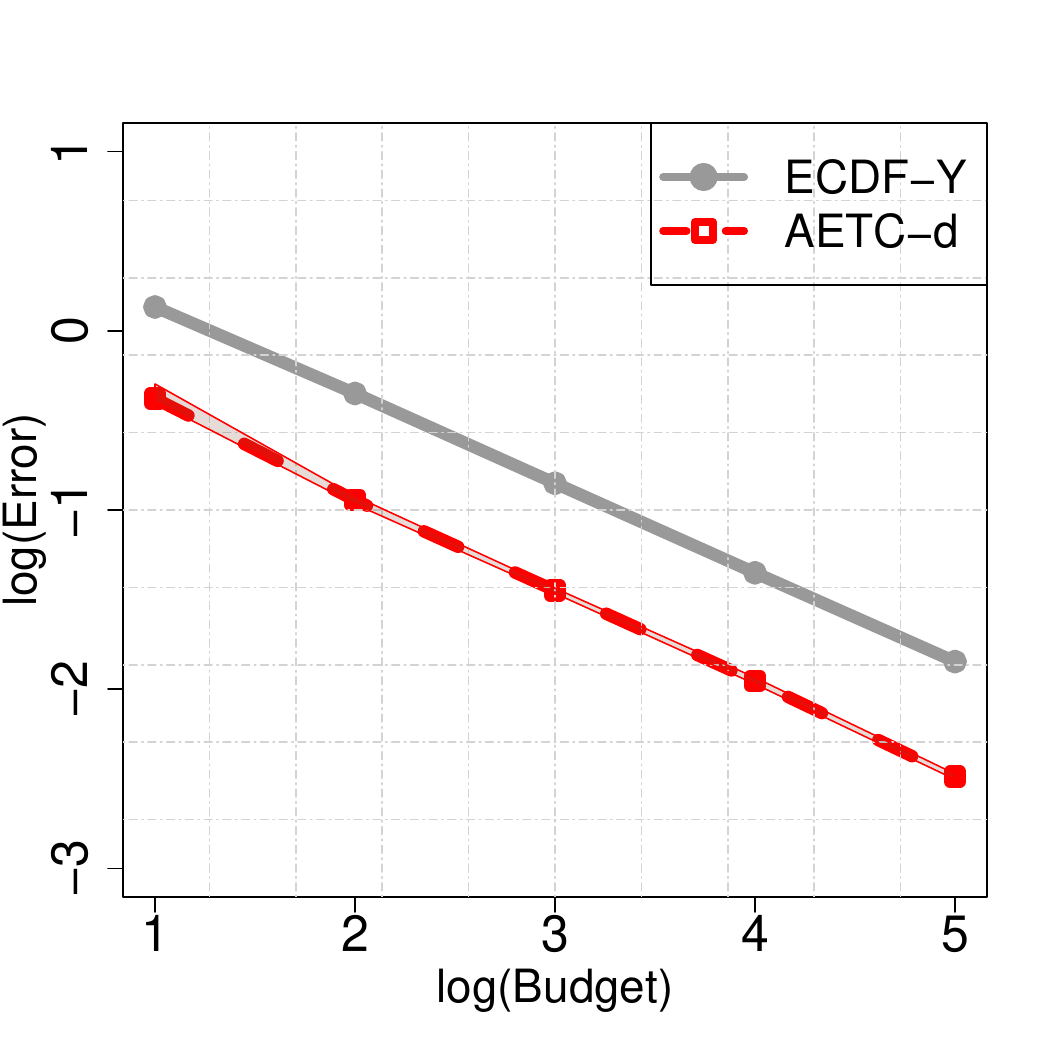}}\hfill
  \subfigure[]{\includegraphics[width=0.25\linewidth]{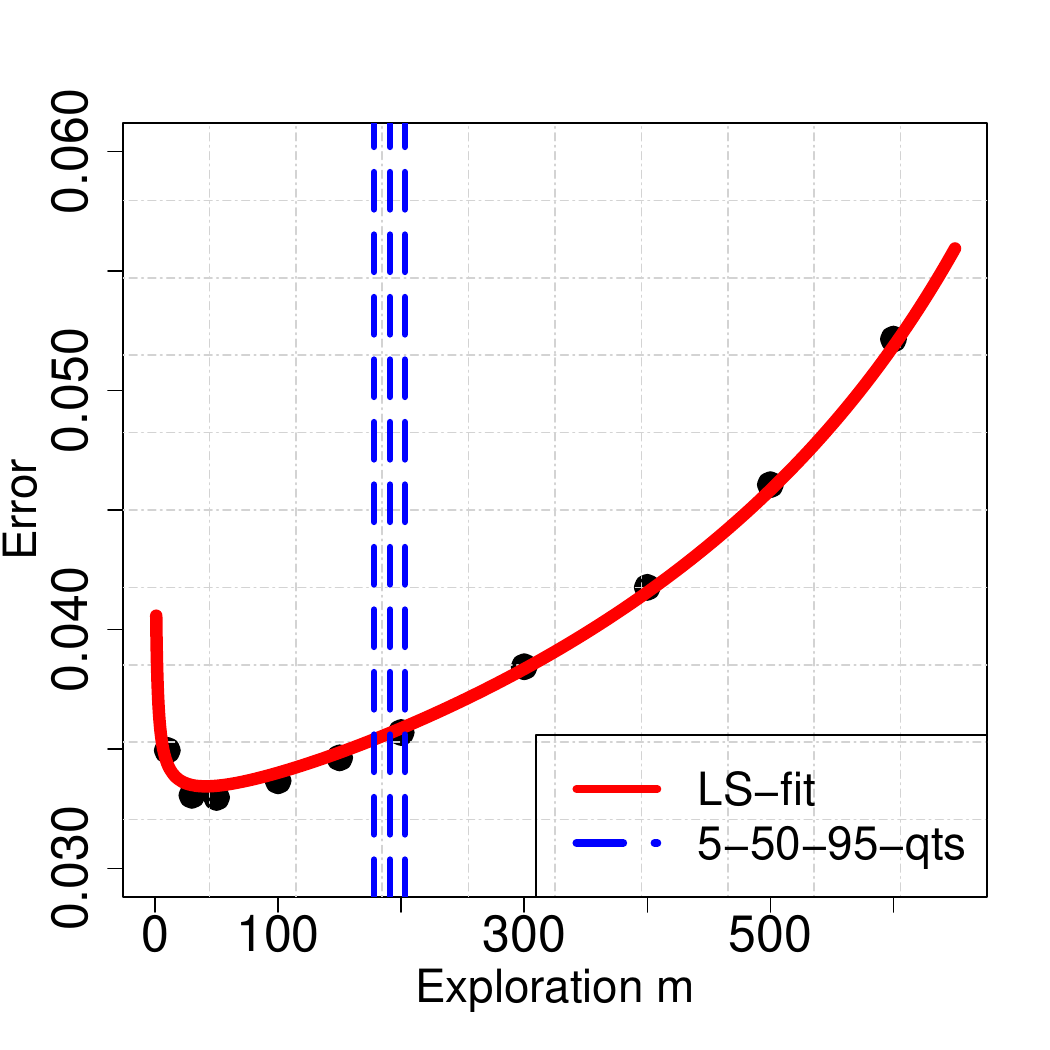}}\hfill
  \subfigure[]{\includegraphics[width=0.25\linewidth]{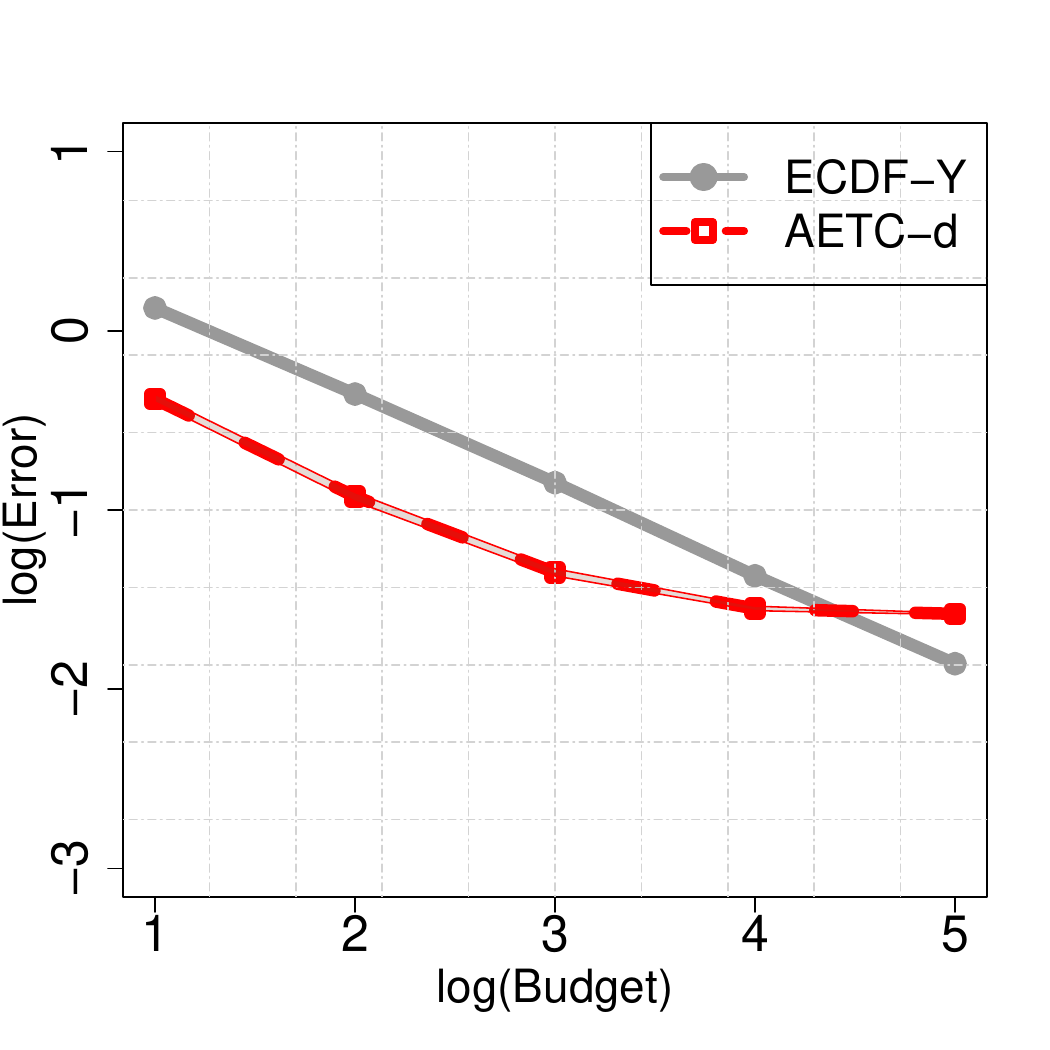}}\hfill
  \subfigure[]{\includegraphics[width=0.25\linewidth]{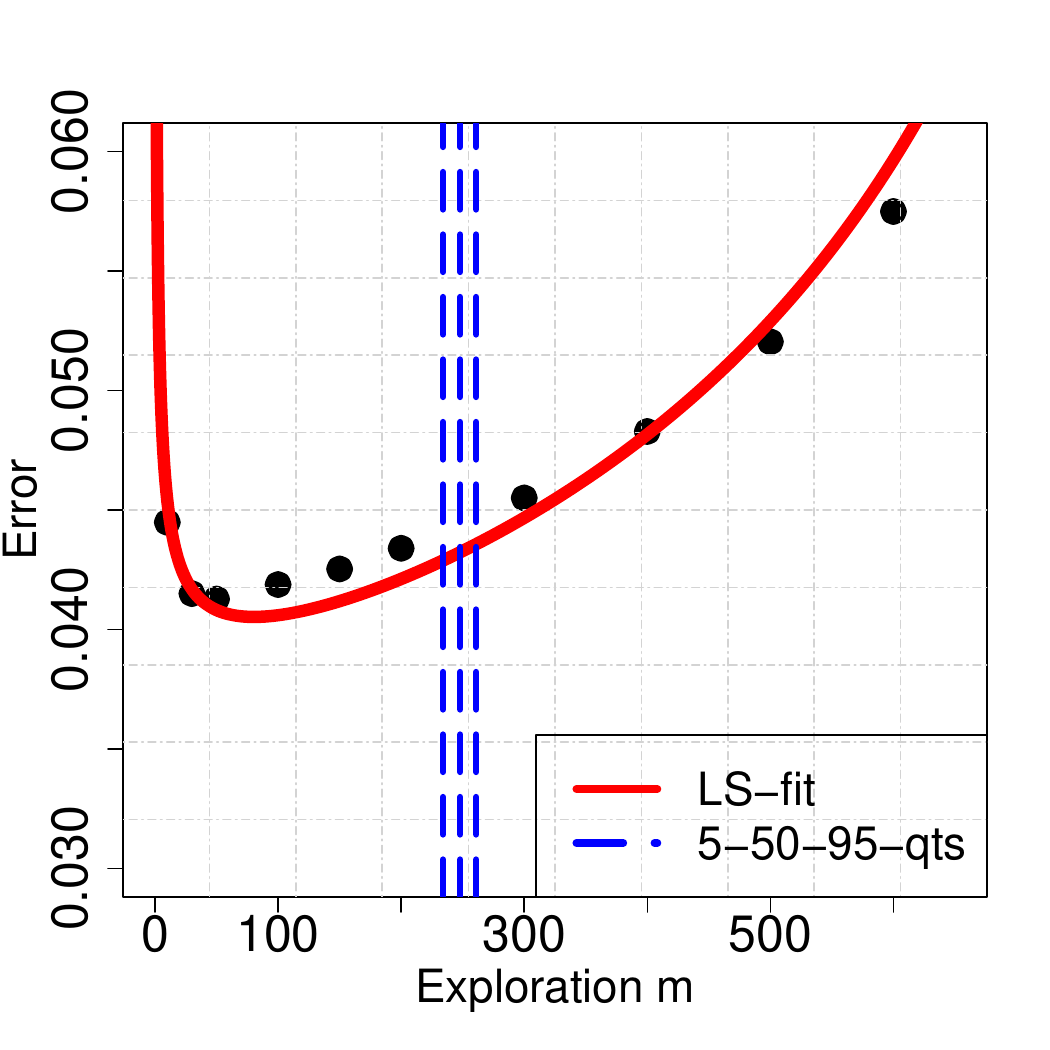}}\hfill
  \caption{($\log_{10}$) mean $W_1$ distance between $F_Y$ and the estimated CDFs given by ECDF-Y and AETC-d in model \eqref{canada} (a) and model \eqref{canada1} (c), with the $5\%$-$50\%$-$95\%$ quantiles plotted for AETC-d to measure its uncertainty in the exploration phase.
 Mean $W_1$ error of the deterministic algorithm at different exploration rates $m$, and the quantiles of the exploration rate from AETC-d in model \eqref{canada1} (b) and model \eqref{canada1} (d).} 
  \label{fig1}
\end{figure}

\subsubsection{Approximate linear assumptions}\label{als}
Let $c=d=0$. 
We now consider the same low-fidelity models which were used in \cite{Qian_2018} for sensitivity analysis: 
\begin{align}
X_1 &= \sin Z_1 + 0.95a\sin^2 Z_2 + bZ_3^4\sin Z_1\label{canada1}\\
X_2 &= \sin Z_1 + 0.6a\sin^2 Z_2 + 9bZ_3^2\sin Z_1\nonumber. 
\end{align} 
The correlations between $Y$ and $X_1, X_2$ are approximately $0.999$ and $0.950$, respectively. 
Under the current setup, the linear model assumption \eqref{1} is not satisfied. 
Nevertheless, the correlations between $Y$ and $X_1, X_2$ suggest that the relationship between the high-fidelity and low-fidelity models is approximately linear. 
The cost of sampling $Y, X_1$ and $X_2$ and the budget range is the same as in the previous case. 
We repeat the same experiments (both the accuracy and the optimality of exploration rates) and report the corresponding results in Figure \ref{fig1}(c,d). 

Despite model misspecification, AETC-d still demonstrates reasonable performance until the total budget exceeds $10^3$. 
When the budget is sufficiently large, both exploration and exploitation errors are so small that model misspecification errors start to dominate. 
A further investigation of the model misspecification effects in this example will be carried out in Section \ref{sc}. 
The optimal model selected by the AETC-d is $S_\opt = \{1,2\}$. 
By comparing the two plots in the right panel in Figure \ref{sc}, we see that AETC-d again overexplores and has a close-to-optimal loss value.

\subsubsection{Model misspecification}\label{sc}

In this section, we identify the misspecified assumptions that prevent the experiment in Section \ref{als} from converging. 
For every $S\subseteq  [n]$, the following decomposition holds:
\begin{align*}
Y = \E[Y|X_S]+ (Y-\E[Y|X_S]),
\end{align*}
where $\E[Y|X_S]$ is a measurable function of $X_S$, and $Y-\E[Y|X_S]$ is a centered random variable that is uncorrelated with $\E[Y|X_S]$.  
The linear model assumption \eqref{1} (or equivalently, Assumptions \ref{ass:eS-cond}-\ref{ass:eS-independent}) assumes the following: 
\begin{itemize}
\item $\E[Y|X_S]$ a linear function of $X_S$;
\item $Y-\E[Y|X_S]$ is independent of $\E[Y|X_S]$. 
\end{itemize}
Violation of either may cause the inconsistent behavior of AETC-d in this example.

To inspect the linearity assumption, we expand the feature space by incorporating nonlinear terms of the existing regressors to fit a larger linear model. 
To test the noise independence condition, 
we can either drop the noise in \eqref{lr:est} to reduce the erroneous noisy effect\footnote{When the variance ratio between $\e_S$ and $Y$ is small, $Y\approx X_S^\top \beta_S$, so that adding the noise emulator has little impact on the accuracy of the resulting estimator. When the ratio is moderate, adding $\e_S$ as an independent component will degrade the quality of the estimator if the independence assumption is violated.}, or use a different method such as quantile regression to model the potential heteroscedasticity of the noise. 
A brief description of using quantile regression combined with AETC-d in our setup is given in Appendix \ref{app5}. 

In the following experiment, we expand the original model by including higher-order terms of the existing features.   
We include additional regressors $\{X_1^2, X_1^3, X_2^2, X_2^3, X_1X_2, X_1^2X_2, X_1X_2^2\}$, and call this enlarged model L.
Note that in L, an exploitation model is no longer defined by the subset of regressors, but instead by the admissible $\sigma$-field generated by the regressors. 
In fact, since both $X_1$ and $X_1^2$ generate the same $\sigma$-field, $\E[Y|X_1] = \E[Y|X^2_1]$. As a result, the linear model assumption \eqref{1} cannot hold simultaneously for both regressors as two distinct models unless $X_1$ is a constant a.s. 
For convenience, for each $S\subseteq [n]$, we define the model associated with $S$ as the intersection of the linear space spanned by the existing features and the $\sigma$-field generated by $X_S$, with exploitation cost given by the total cost of the low-fidelity models in $S$. 
For instance, if $S=\{1\}$, then the corresponding model is $\text{span}\{1, X_1, X_1^2, X_1^3\}$, with exploitation cost $c_1$ per sample.   
In this case, any model generating the same $\sigma$-field can be viewed as a sub-model under our definition. 
Also, in the exploitation stage of AETC-d, we consider two alternative methods to estimate $F_Y$: AETC-no and AETC-q.
The performance of these modified estimators in both the original and the expanded models is compared in Figure \ref{fig3}(a-c).
For ease of illustration, we only plot the median of errors instead of the $5\%$-$50\%$-$95\%$ quantiles region in the $100$ experiments. 

Figure \ref{fig3} shows that the overall performance of all the methods under comparison improves in the enlarged model L when the budget exceeds $10^3$.
This implies that adding nonlinear terms of the existing regressors can help reduce the model misspecification effects.
The improvement for AETC-d in model L diminishes when the budget reaches a new threshold near $10^5$, after which more regressors are needed to further increase model expressivity. 
In this example, noise has a less dominant role in affecting model convergence. 
Nevertheless, using quantile regression for reconstruction is beneficial for mitigating the noise misspecification effects.  

We end this section by providing an instance of the pointwise absolute error between $F_Y$ and the estimated CDFs given by the ECDF-Y, the AETC-d, the AETC-d (L), the AETC-d-no (L), and the AETC-d-q (L) when $B = 10^4$.
For better visualization, we plot the distribution of these errors (i.e. CDF) rather than the pointwise curves. 
The result is given in Figure \ref{fig3}(d).

\begin{figure}[htbp]
\centering 
\subfigure[]{\includegraphics[width=0.25\linewidth]{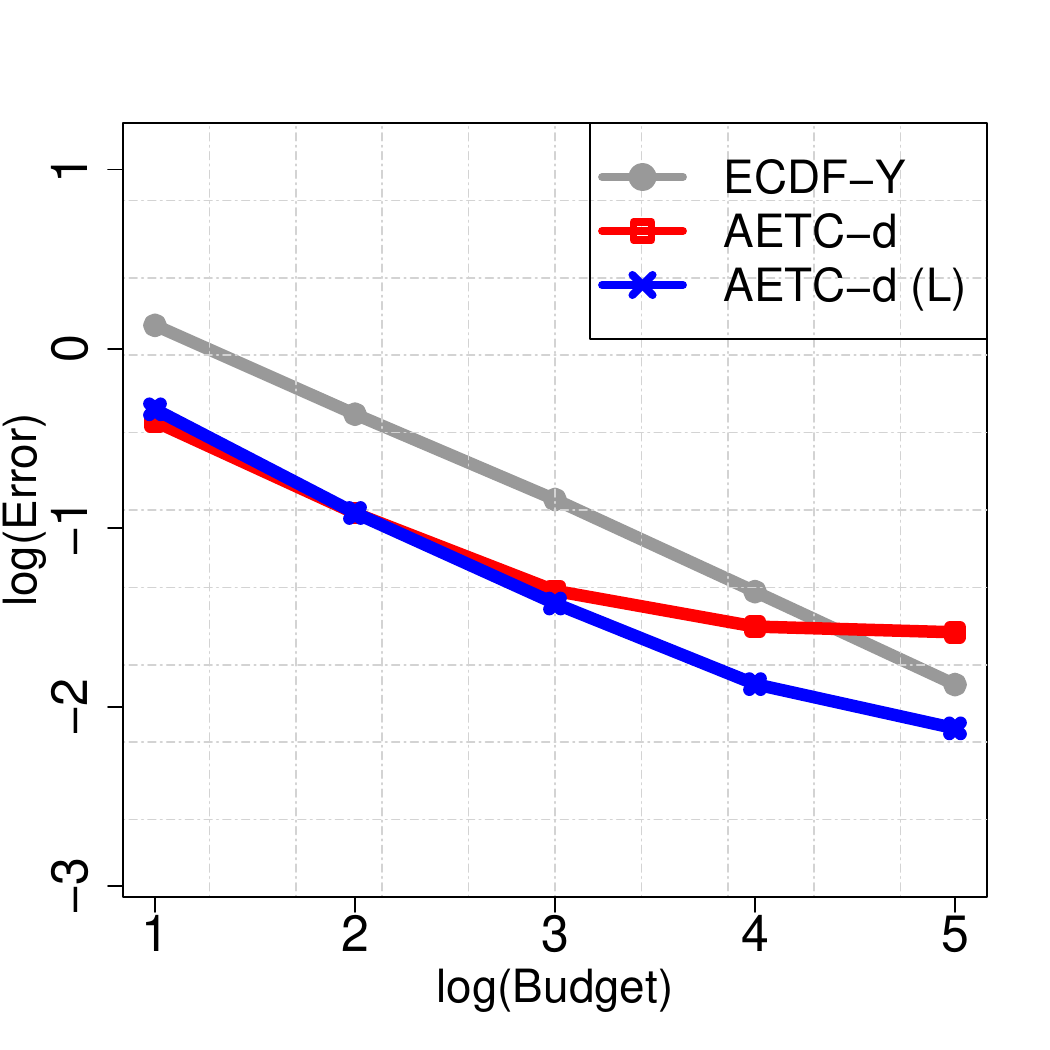}}\hfill
\subfigure[]{\includegraphics[width=0.25\linewidth]{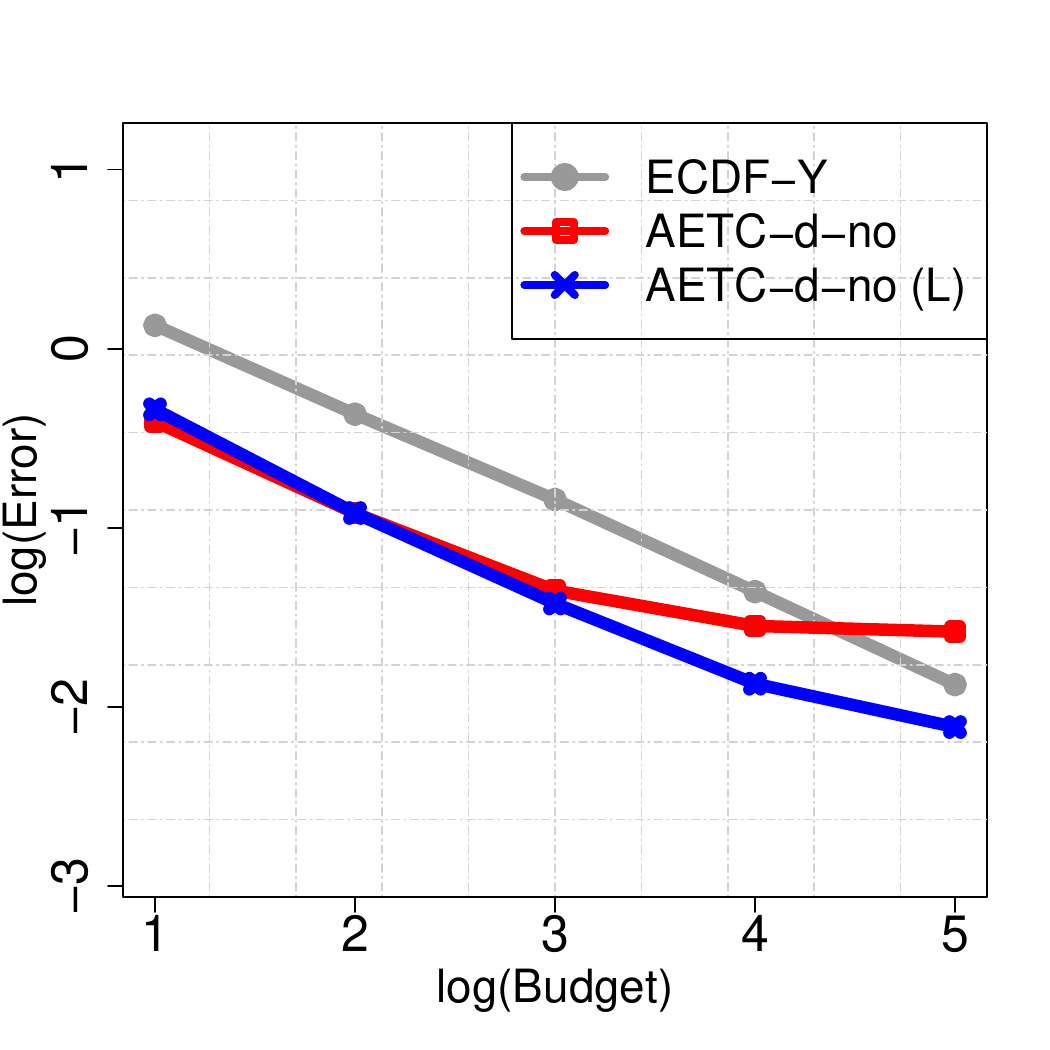}}\hfill
\subfigure[]{\includegraphics[width=0.25\linewidth]{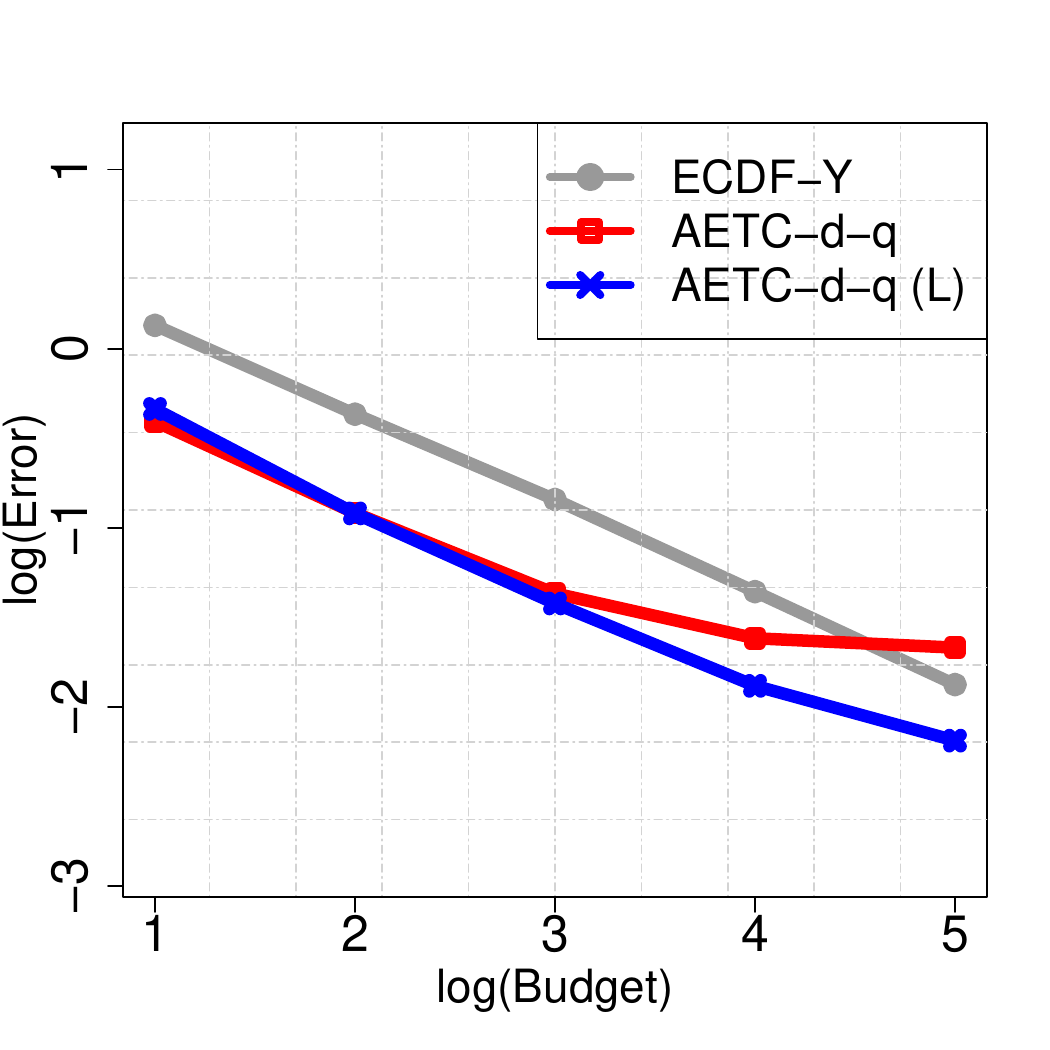}}\hfill
\subfigure[]{\includegraphics[width=0.25\linewidth]{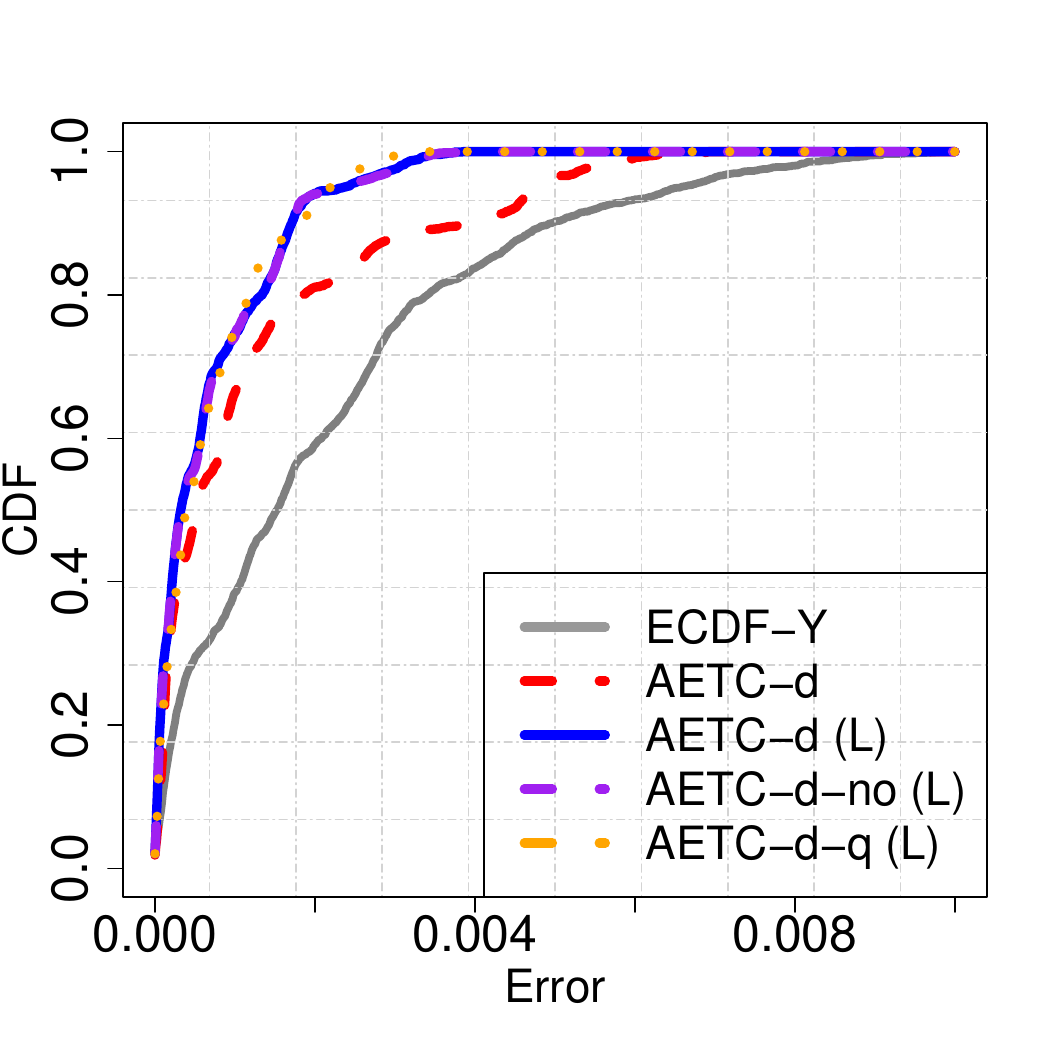}}\hfill
\caption{Median estimation errors of the estimated CDFs given by ECDF-Y and AETC-d with different exploitation strategies in both the original and the expanded model L (a)-(c). An instance of distribution of pointwise absolute errors of the CDFs given by ECDF-Y, AETC-d, AETC-d (L), AETC-d-no (L), and AETC-d-q (L) when $B = 10^4$ (d). } \label{fig3}
\end{figure}

\subsubsection{Cost assumptions}

Relative information between cross-model correlations and model costs may also impact the performance of the AETC-d algorithm.
To understand the relationship between them, we fix the models and vary the (relative) cost assigned to each model. 
In particular, we test the following four cost conditions on both model \eqref{canada} (perfect model assumptions) and model \eqref{canada1} (approximate model assumptions):
\begin{itemize}
\item $(c_0, c_1, c_2) = (1, 0.5, 0.001)$;
\item $(c_0, c_1, c_2) = (1, 0.05, 0.001)$; \hskip 10pt (original choice in sections \ref{sssec:pm} and \ref{als})
\item $(c_0, c_1, c_2) = (1, 0.005, 0.001)$;
\item $(c_0, c_1, c_2) = (1, 0.001, 0.001)$,
\end{itemize}
For model \eqref{canada1}, we also test the enlarged model L that includes the higher order terms; see Section \ref{sc}. 
The results are reported in Figure \ref{fig33}.
For ease of illustration, we only plot the error curves given by ECDF-Y and AETC-d.

\begin{figure}[htbp]
\centering
\subfigure[$(c_0, c_1, c_2) = (1, 0.5, 0.001)$]{\includegraphics[width=0.25\linewidth]{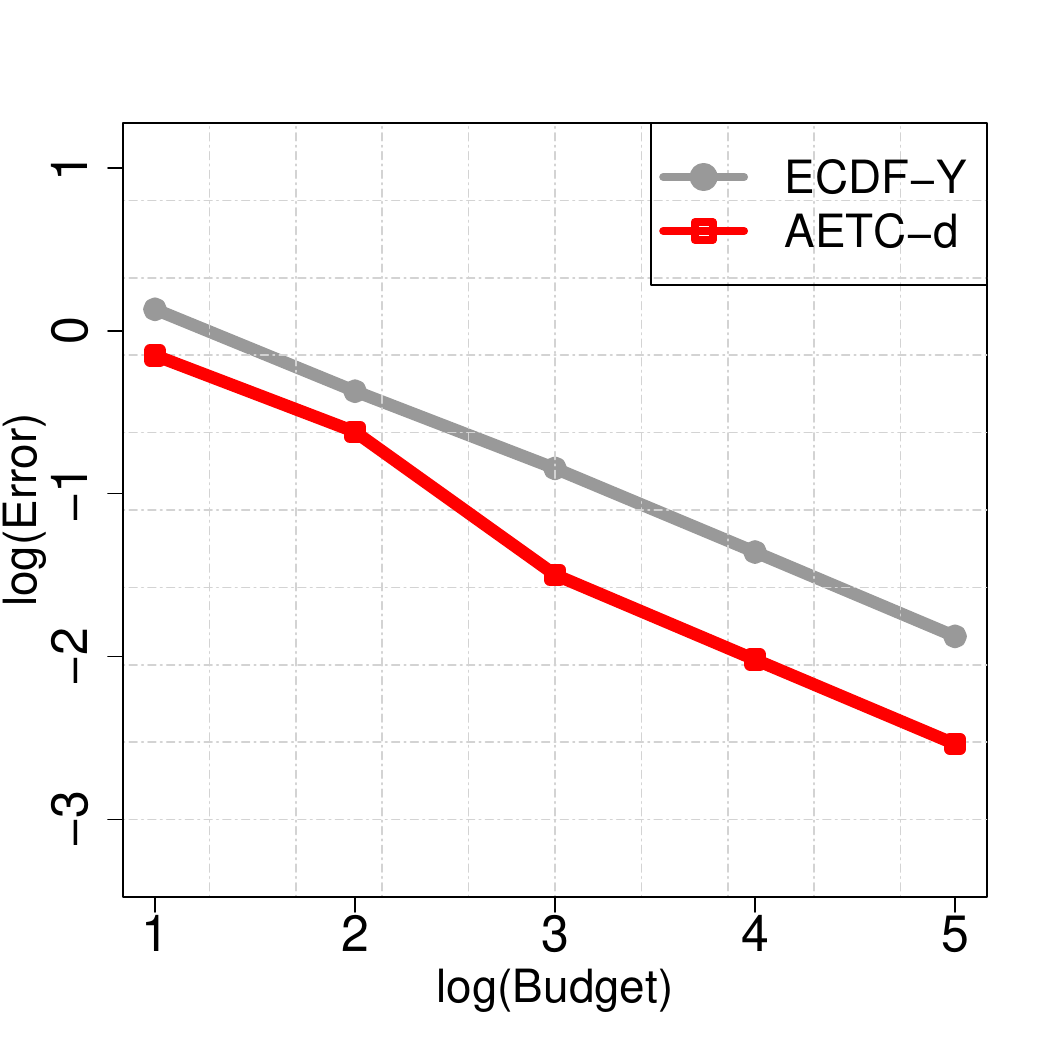}}\hfill
\subfigure[$(1, 0.05, 0.001)$]{\includegraphics[width=0.25\linewidth]{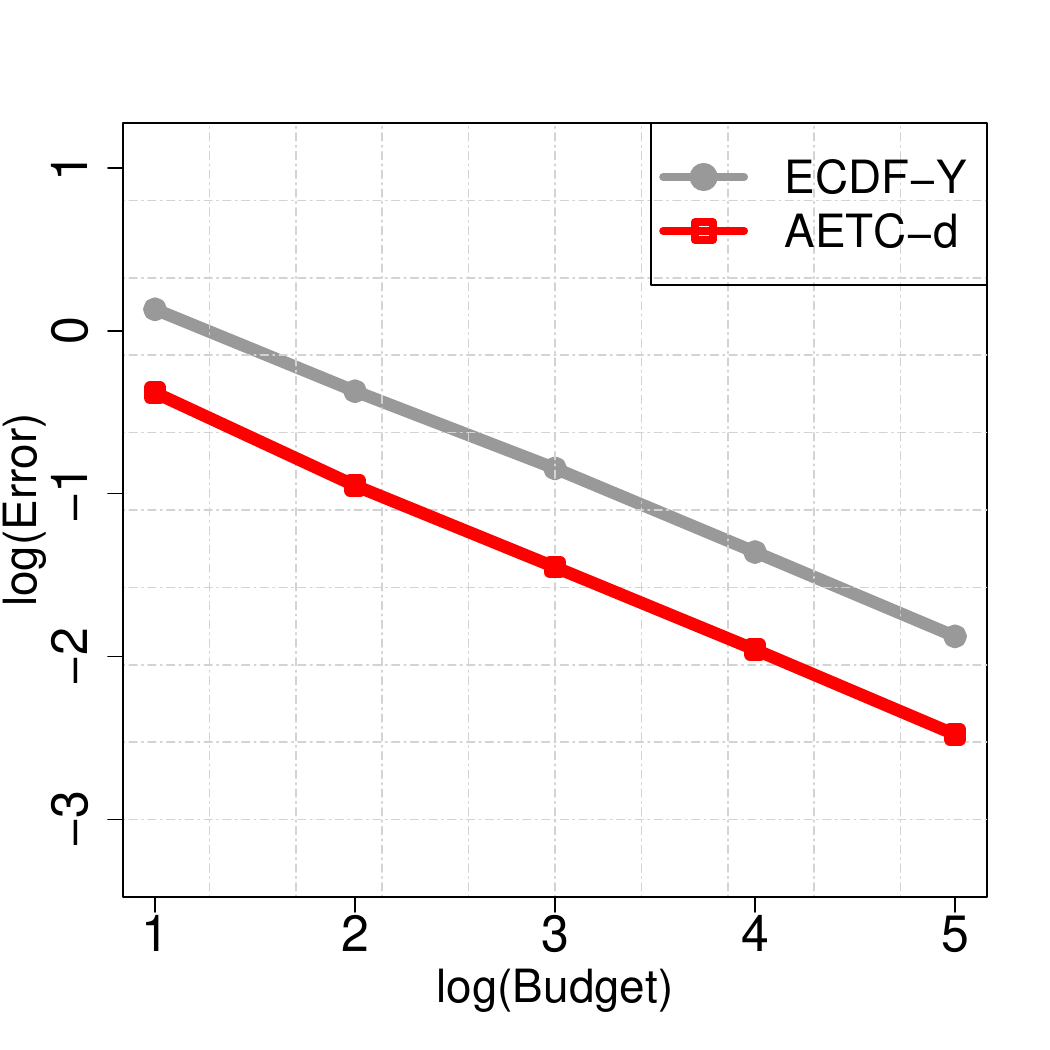}}\hfill
\subfigure[$(1, 0.005, 0.001)$]{\includegraphics[width=0.25\linewidth]{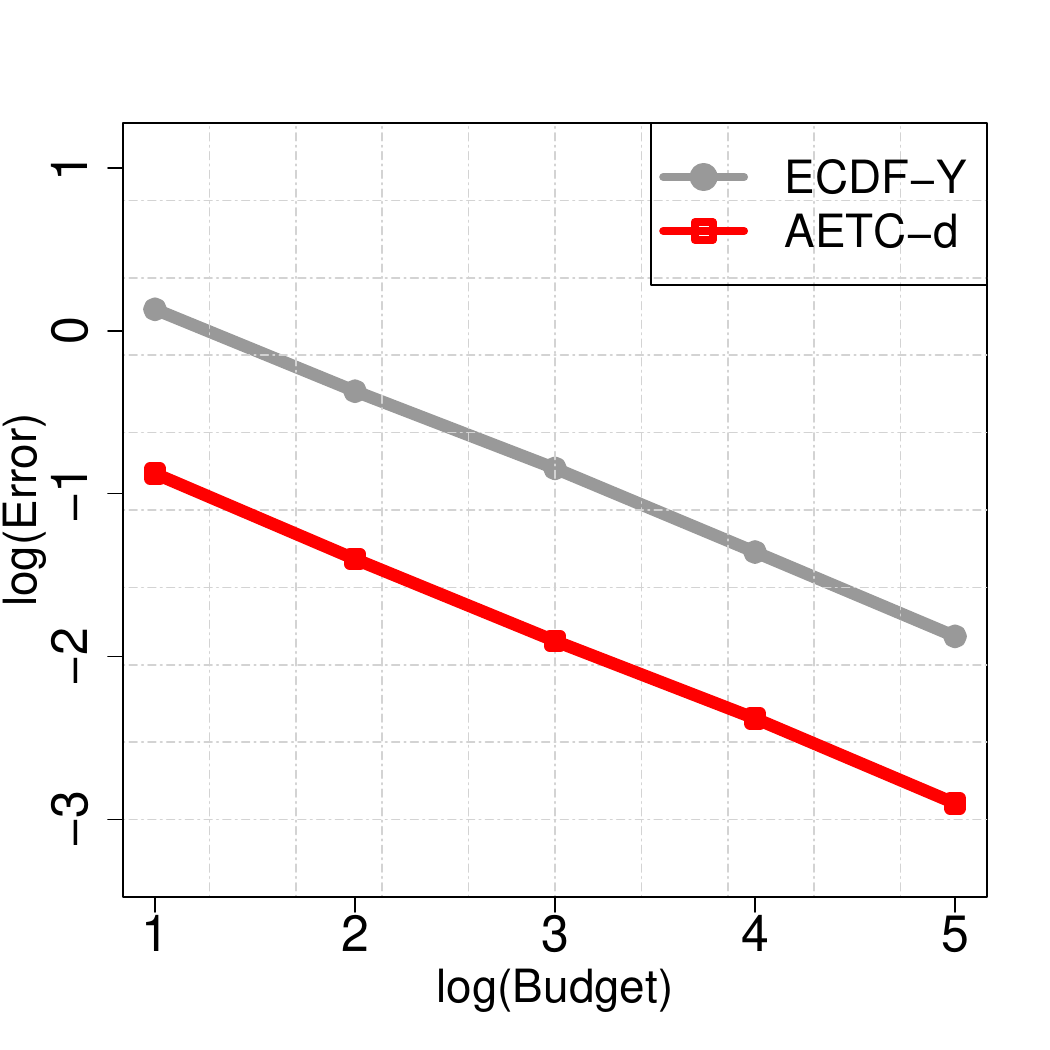}}\hfill 
\subfigure[$(1, 0.001, 0.001)$]{\includegraphics[width=0.25\linewidth]{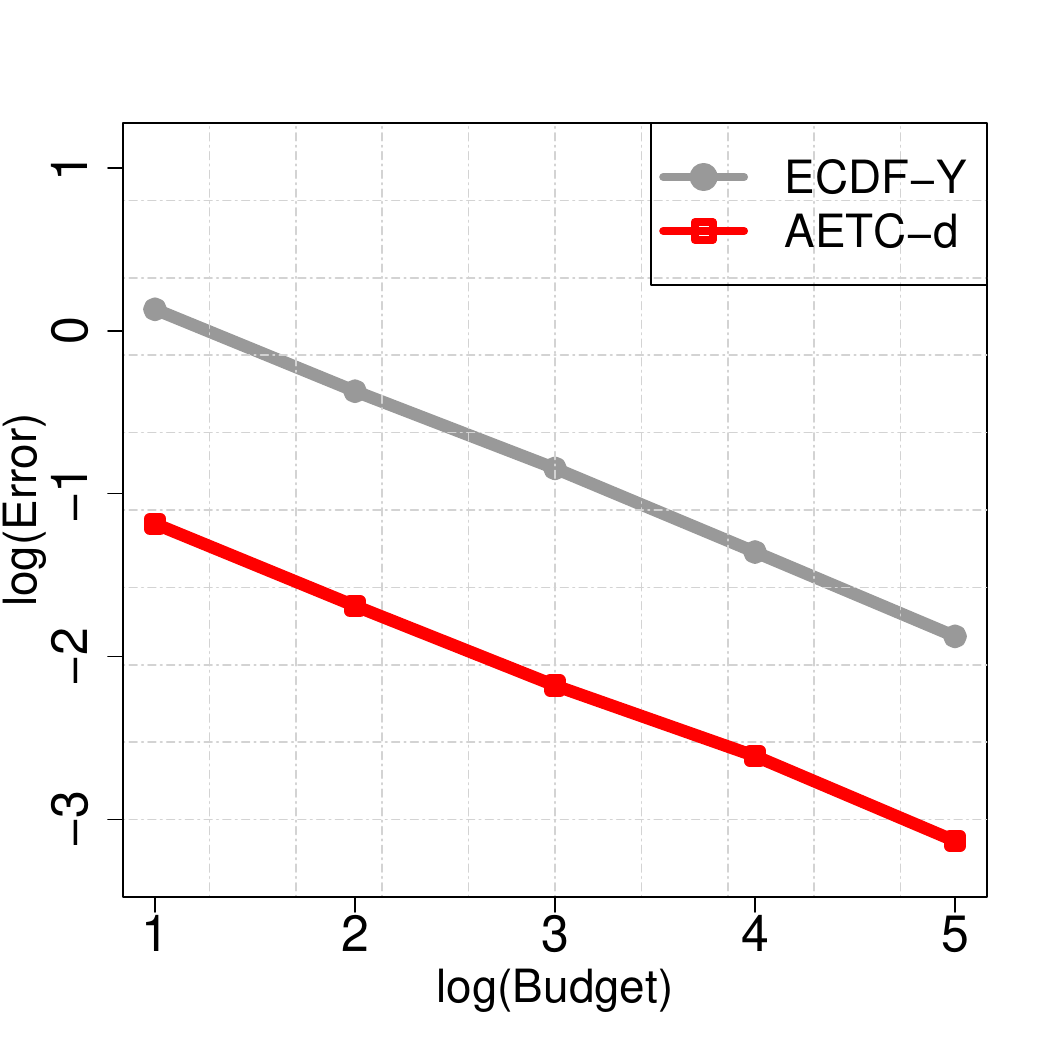}}\hfill
\subfigure[$(1, 0.5, 0.001)$]{\includegraphics[width=0.25\linewidth]{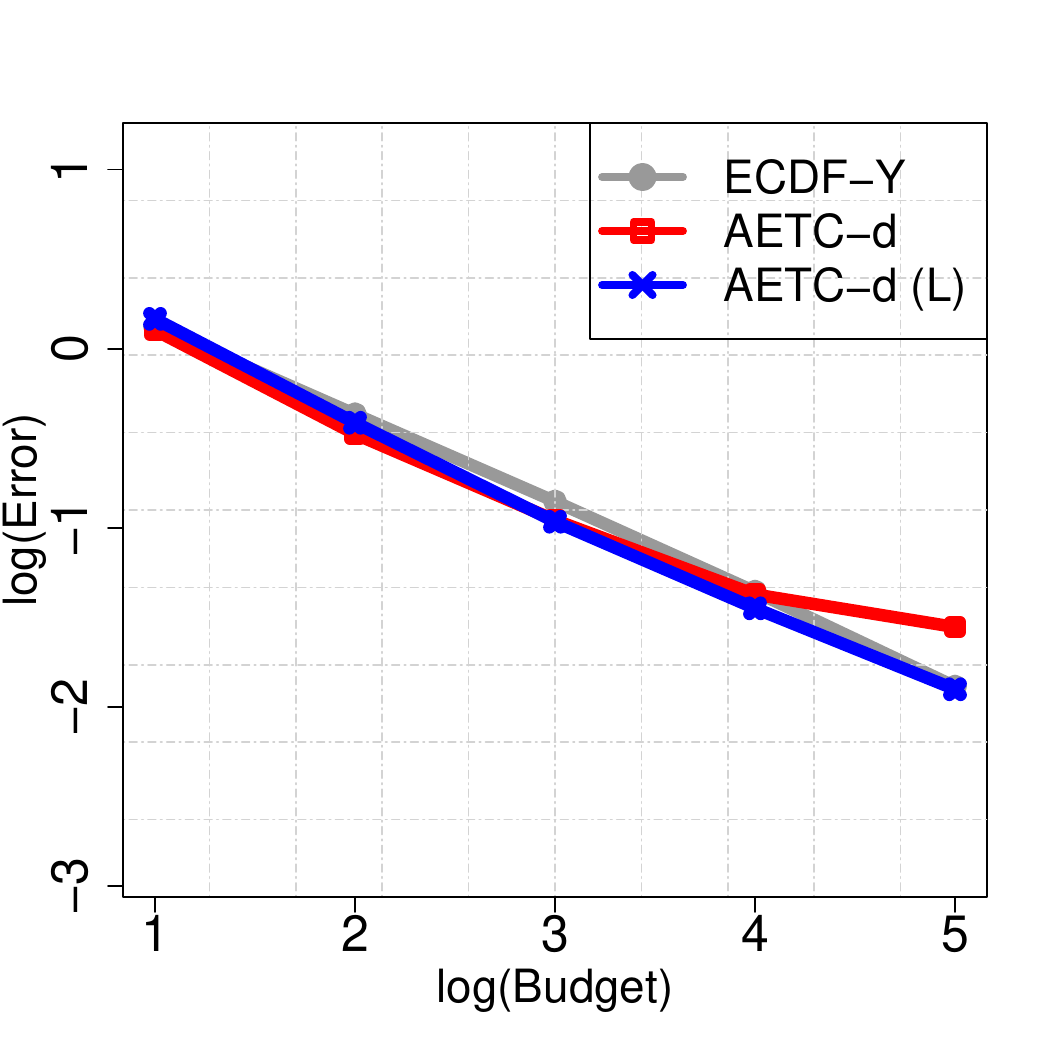}}\hfill
\subfigure[$(1, 0.05, 0.001)$]{\includegraphics[width=0.25\linewidth]{2022-01.pdf}}\hfill
\subfigure[$(1, 0.005, 0.001)$]{\includegraphics[width=0.25\linewidth]{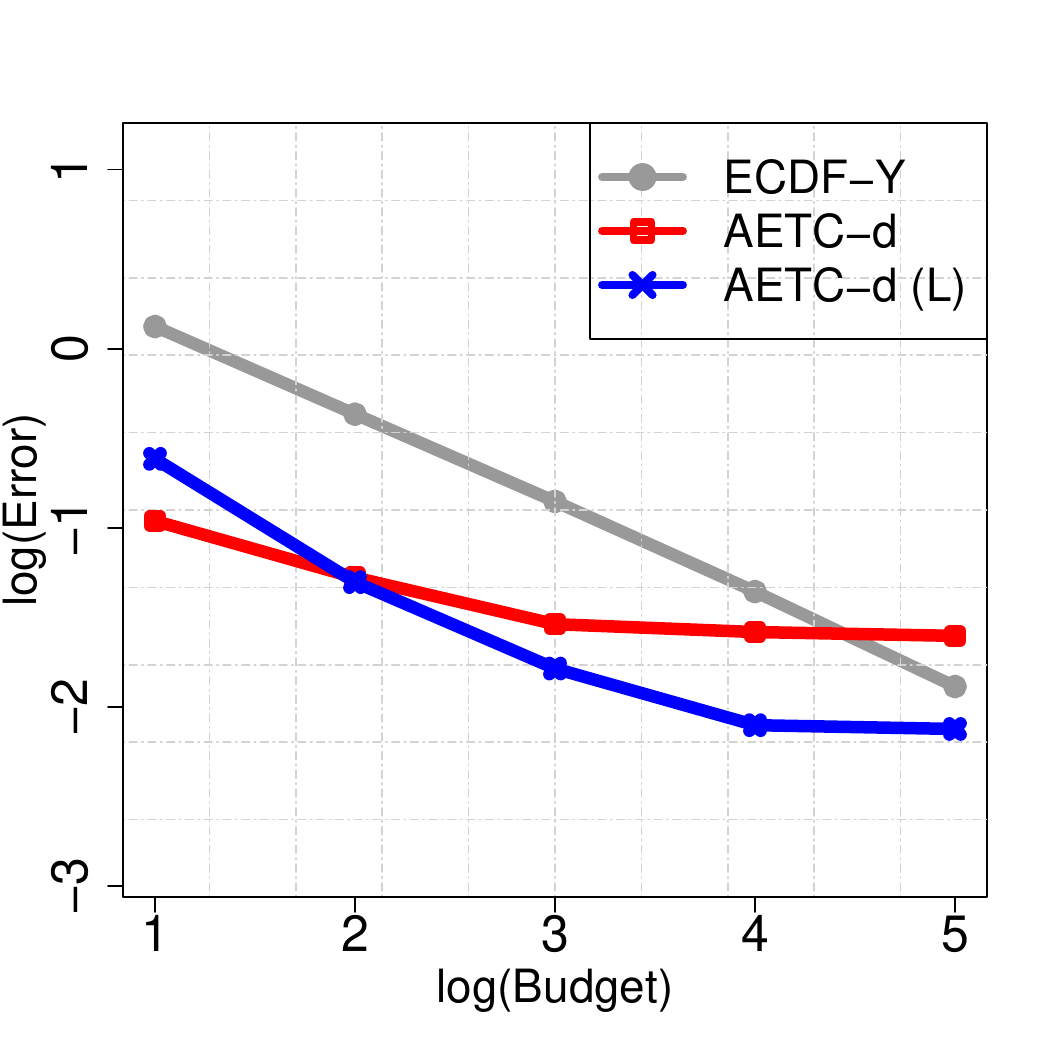}}\hfill 
\subfigure[$(1, 0.001, 0.001)$]{\includegraphics[width=0.25\linewidth]{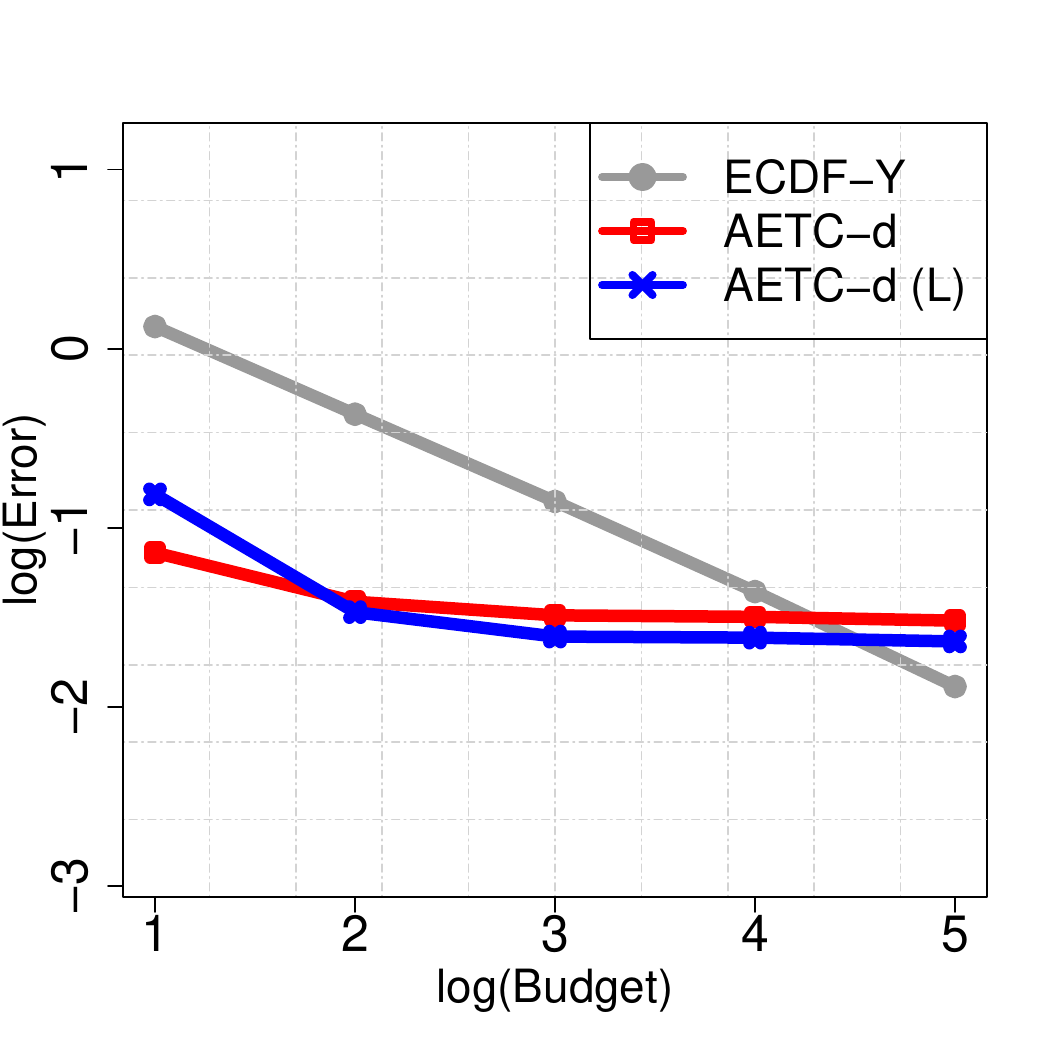}}\hfill 
\caption{Median estimation errors of the estimated CDFs given by ECDF-Y and AETC-d in both model \eqref{canada} (a)-(d) and model \eqref{canada1} (e)-(h) under different cost assumptions.} \label{fig33}
\end{figure}

In Figure \ref{fig33}, the limiting models selected by the AETC-d algorithm in scenarios (a)-(d) are $\{2\}$, $\{1\}$, $\{1\}$, and $\{1\}$, respectively. 
In particular, when model $\{1\}$ becomes almost as expensive as the high-fidelity model, the AETC-d algorithm turns to use the less expensive low-fidelity model to better leverage cost and accuracy. 
When both low-fidelity models become equally cheap, the AETC-d algorithm selects $\{1\}$ only. 
In none of the above scenarios did AETC-d choose $\{1, 2\}$ for exploitation when $B$ is large, which is consistent with the observation $\E[Y|X_1, X_2] = \E[X_1]$, i.e., adding $X_2$ to $X_1$ does not provide additional information in terms of understanding $Y$. 
This suggests that the AETC-d algorithm can filter redundant information among models in model selection. 
Finally, note that the performance of AETC-d in (a) is even slightly better than in (b) despite a larger cost for the chosen low-fidelity models. 
One would expect that model $\{2\}$ is also slightly better than model $\{1\}$ when $(c_0, c_1, c_2) = (1, 0.05, 0.001)$.
(We used the word ``slightly'' because model $\{2\}$ should have a similar efficiency in (a)-(d) since only the cost of $\{1\}$ is varied and the high-fidelity cost dominates.)
This shows that the $G_S$ criteria may select a suboptimal model when the efficiency gap between different models is small. 
AETC-d has a better chance of finding the optimal model when the gap is large (i.e. in (a) and (d)). 

For scenarios (e)-(g), the asymptotically selected models in both the original and enlarged model L are $\{1, 2\}$. 
As a result, they have the same error threshold below which the model misspecification effects start to dominate. 
The major difference lies in the relative efficiency of AETC-d over ECDF-Y. 
For instance, in (g), the cost ratio between the high- and low-fidelity models is large, and the computational gain of AETC-d is more prominent over ECDF-Y compared to (e) and (f) before the error threshold, which, as expected, arrives under a smaller budget. 
For scenario (h), the AETC-d algorithm chooses model $\{1\}$ only. 
For model $\{1\}$, model misspecification occurred at a similar threshold as in $\{1, 2\}$. 
In this case, adding only quadratic and cubic terms of $X_1$ is not sufficient to mitigate the effects (as opposed to adding seven higher order terms for $\{1,2\}$).  
Unfortunately, the AETC-d algorithm makes decisions assuming all the models are correct and finds the most efficient one among them. 

\color{black}

\subsection{Parametrized PDEs}\label{sec:ppde}
In the last experiment, we consider a multifidelity model given by a parametric elliptic equation. 
The setup is taken from \cite{xu2021bandit}.
We consider an elliptic PDE over a square spatial domain $D = [0,1]^2$ that governs displacement in linear elasticity.
The geometry and boundary conditions are shown in Figure~\ref{fig:struct}.

\begin{figure}[htbp]
\begin{center}
\includegraphics[width=0.38\textwidth]{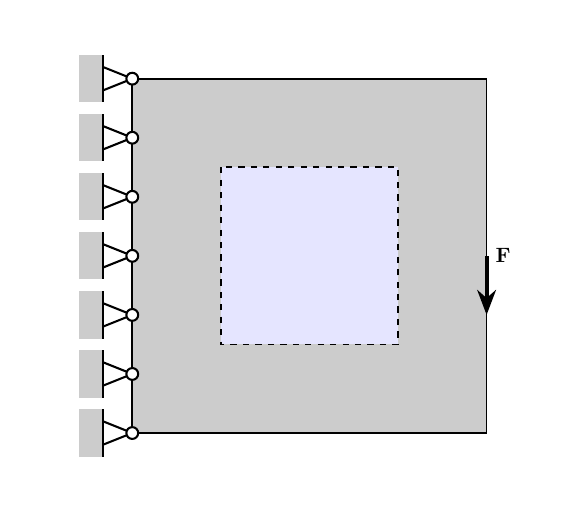}
\caption{\small Geometry and boundary conditions for the linear elastic structure with the square domain.}\label{fig:struct}
  \end{center}
\end{figure}

The parametric version of this problem equation seeks the displacement field $\bs{u} = (u, v)^\top$ that is the solution to the PDE system,
\begin{align*}
  -\nabla \cdot \left(\kappa(\bs{p}, \bs{x})\; \bs{\sigma}(\bs{x},\bs{p})\right) = \bs{F}(\bm x), \hskip 10pt \forall (\bm p,\bm x) \in \mathcal{P} \times D \hskip 20pt\\
       \bs{\sigma} = \begin{bmatrix}
        \sigma_1 & \sigma_{12} \\ 
        \sigma_{12} & \sigma_2
        \end{bmatrix}, \hskip 20pt
         \begin{bmatrix}
          \sigma_1 \\ \sigma_2 \\ \sigma_{12} 
          \end{bmatrix} = \frac{1}{1 - \nu^2} \begin{bmatrix}
           \pfpx{u}{x_1} + \pfpx{\nu}{x_2} \\ \pfpx{v}{x_2} + \nu \pfpx{u}{x_1} \\ \frac{1-\nu}{2} (\pfpx{u}{x_1} + \pfpx{v}{x_2}) 
     \end{bmatrix}
\end{align*}
where $\bm p \in \R^{4}$ is a random vector with independent components uniformly distributed on $[-1,1]$.
We have fixed-displacement boundary conditions on the left wall, with the forcing $\bs{F}$ being nonzero only on the right edge of the structure and equal to the constant 1.
We set the Poisson ratio to $\nu = 0.3$, and $\kappa(\bm p, \bm x)$ is a scalar modeled as a truncated Karhunen-Lo\`{e}ve expansion, given by 
\begin{align*}
  \kappa(\bm p, \bm x) = 1 + 0.5 \sum_{i=1}^{4} \sqrt{\lambda_i} \phi_i(\bm x) p_i,
\end{align*}
where $(\lambda_i, \phi_i)$ are ordered eigenpairs of an exponential covariance kernel on $D$, i.e., 
\begin{align*}
  \mathrm{corr}(\kappa(\bm p, \bm x), \kappa(\bm p, \bm y)) = \exp(-\| \bm x - \bm y\|_1/a),
\end{align*}
where $\|\cdot\|_1$ is the $\ell^1$ norm on vectors, and we choose $a = 0.7$. 
The displacement $\bm u$ is used to compute a scalar QoI, the structural \emph{compliance} or energy norm of the solution, which is the measure of elastic energy absorbed in the structure as a result of loading:
\begin{align}
  E \coloneqq \int_D (\bs{u} \cdot \bs{F}) d \bm x.
\end{align}
We solve the above system for each fixed $\bm p$ via the finite element method with standard bilinear square isotropic finite elements on a rectangular mesh. 

\color{black}

In this example, we form a multifidelity hierarchy through mesh coarsening through mesh parameter $h$.
The model solved with mesh size $h =2^{-7}$ is the high-fidelity model.
We create three low-fidelity models based on more economical discretizations: $h = 2^{-3}, 2^{-2}, 2^{-1}$. 
The outputs of these models are the energy QoI computed from the respective approximate solutions. 

The cost for each model is the computational time, which we take to be inversely proportional to the mesh size squared, i.e., $h^2$. (This corresponds to using a linear solver of optimal linear complexity.)
We normalize cost so that the model with the lowest fidelity has unit cost, i.e., $c_0 = 4096, c_1 = 16, c_2 = 4, c_3 = 1$. 
The correlations between the outputs of $Y$ and $X_1, X_2, X_3$ are $0.940, 0.841, -0.146$, respectively.  
The total budget $B$ ranges from $10^5$ to $10^7$. 

To mitigate potential model misspecification effects, we include second-order interactions between $X_i, i\in [3]$ as additional regressors, i.e., $X_iX_j, i, j\in [3]$.
(We do observe model misspecification effects when $B$ is near $10^7$ without the interaction terms.)
In this case, we have seven different models, namely, 
\begin{align*}
\{i\} &\sim \text{span}\{1, X_i, X_i^2\}& i\in [3]\\
\{i, j\} &\sim \text{span}\{1,X_i, X_j, X_i^2, X_j^2, X_iX_j\}& i, j\in [3], i\neq j\\
\{1,2,3\}&\sim\text{span}\{1, X_1, X_2, X_3, X_1^2, X_2^2, X_3^2, X_1X_2, X_1X_3, X_2X_3\},
\end{align*} 
with exploitation cost defined as the sum of the cost of the low-fidelity models used to build the regressors. 
The oracle $F_Y$ is taken as an empirical CDF constructed from two million independent samples of the high-fidelity model. 
Accuracy results for ECDF-Y and AETC-d are reported in the first plot in Figure \ref{fig4}.
For visualization, we also provide an instance of the estimated densities 
(\texttt{density} function in \texttt{R} \cite{R} with the default bandwidth (data-dependent), and the Gaussian kernel for smoothing) and some statistics (mean, variance, skewness, and kurtosis) given by ECDF-Y and AETC-d using the same training dataset when $B = 10^6$.

\begin{figure}[htbp]
\centering 
\subfigure[]{\includegraphics[width=0.38\linewidth]{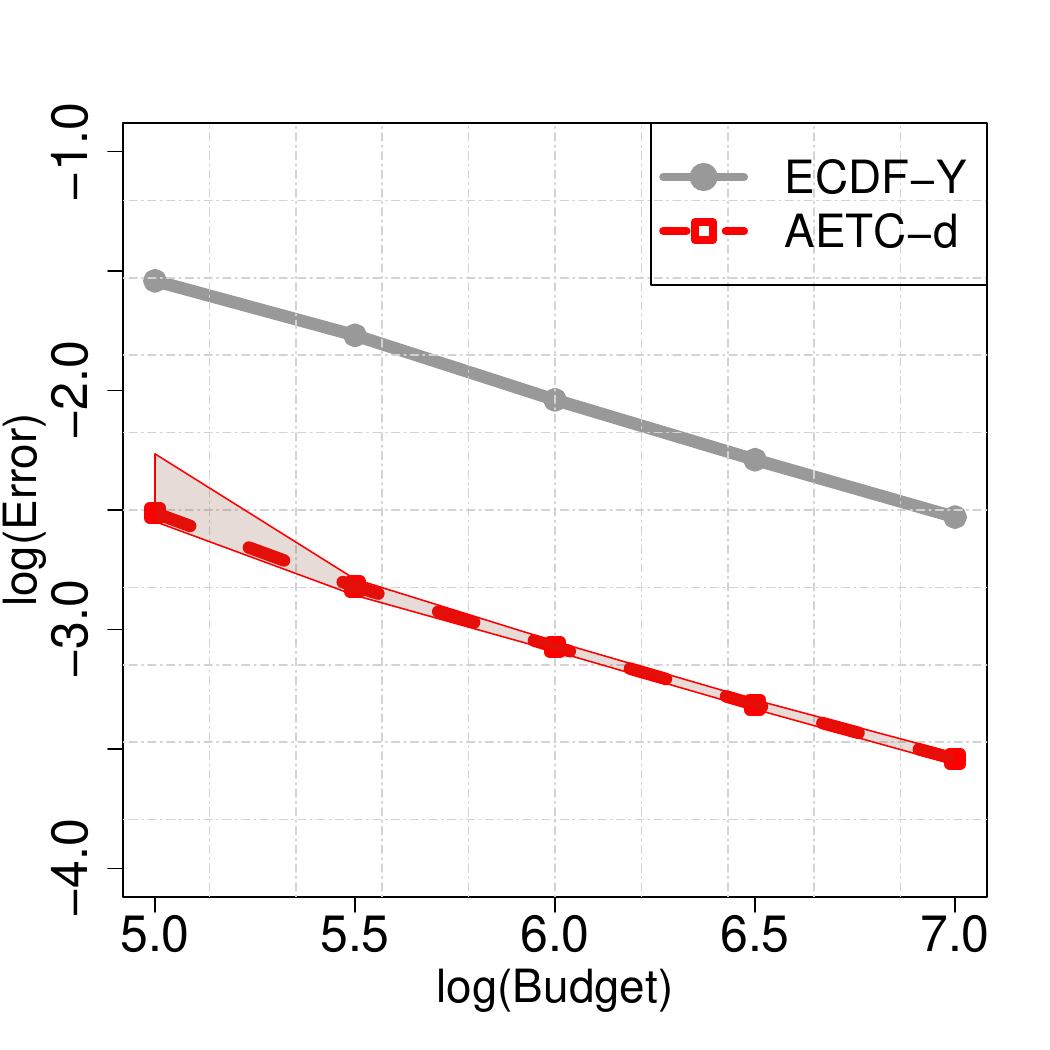}}\hfill
\subfigure[]{\includegraphics[width=0.38\linewidth]{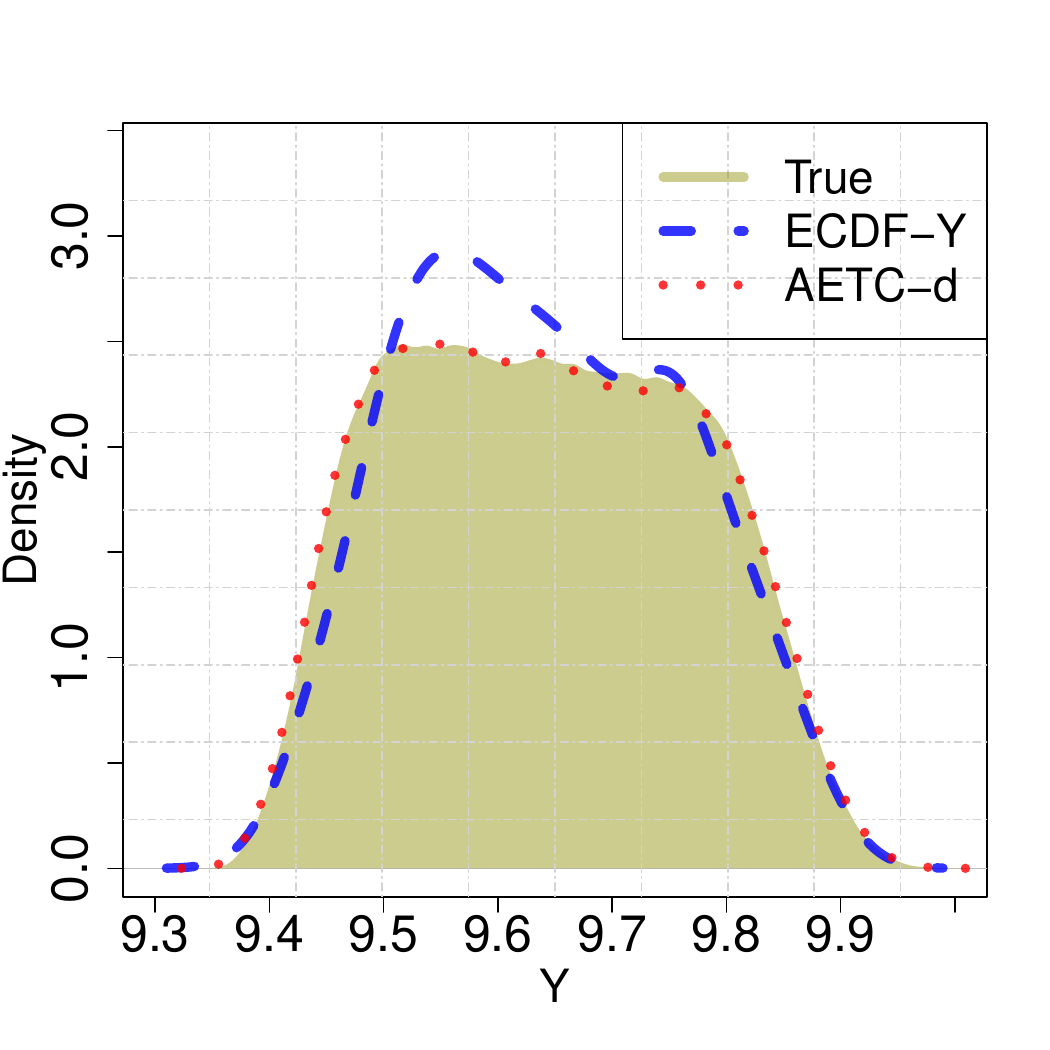}}\hfill
\subfigure[]{\includegraphics[width=0.25\linewidth]{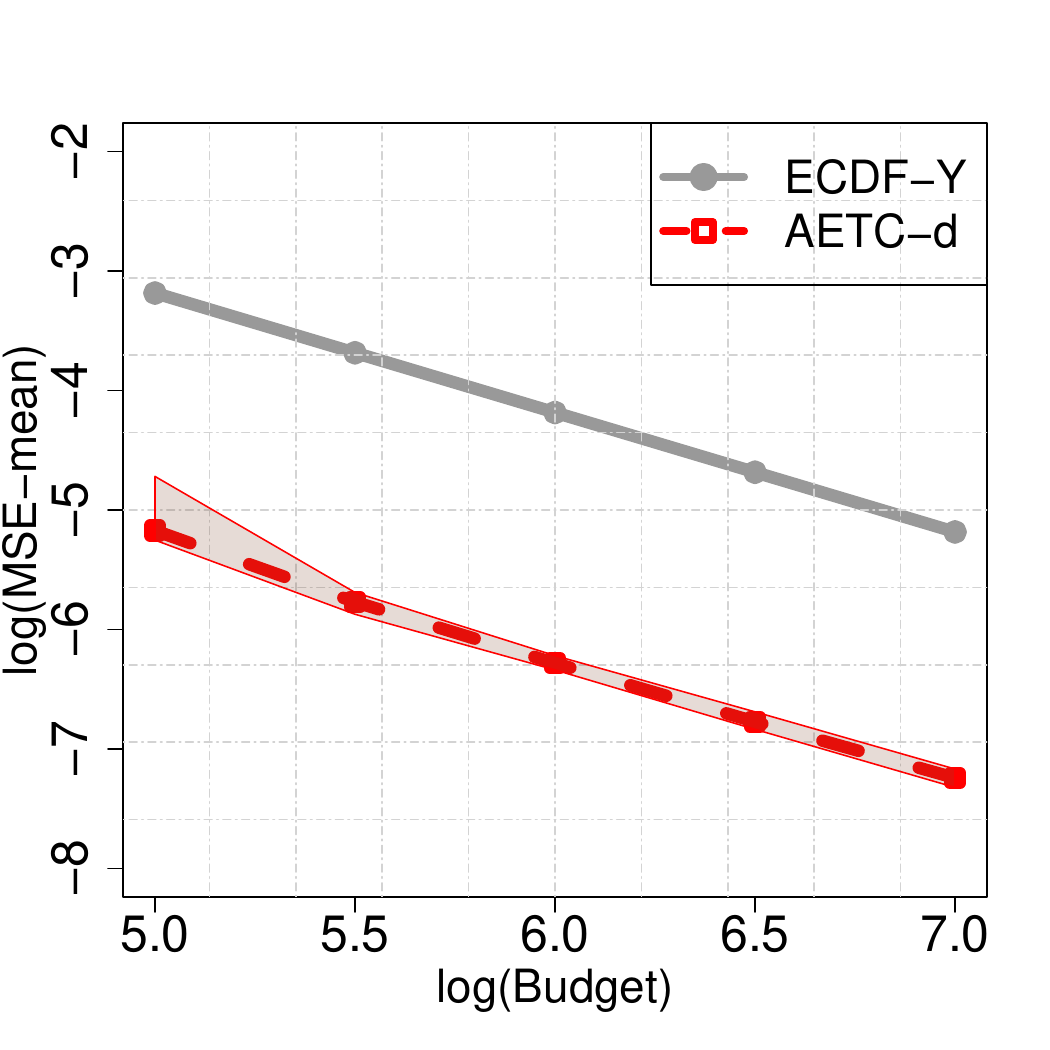}}\hfill
\subfigure[]{\includegraphics[width=0.25\linewidth]{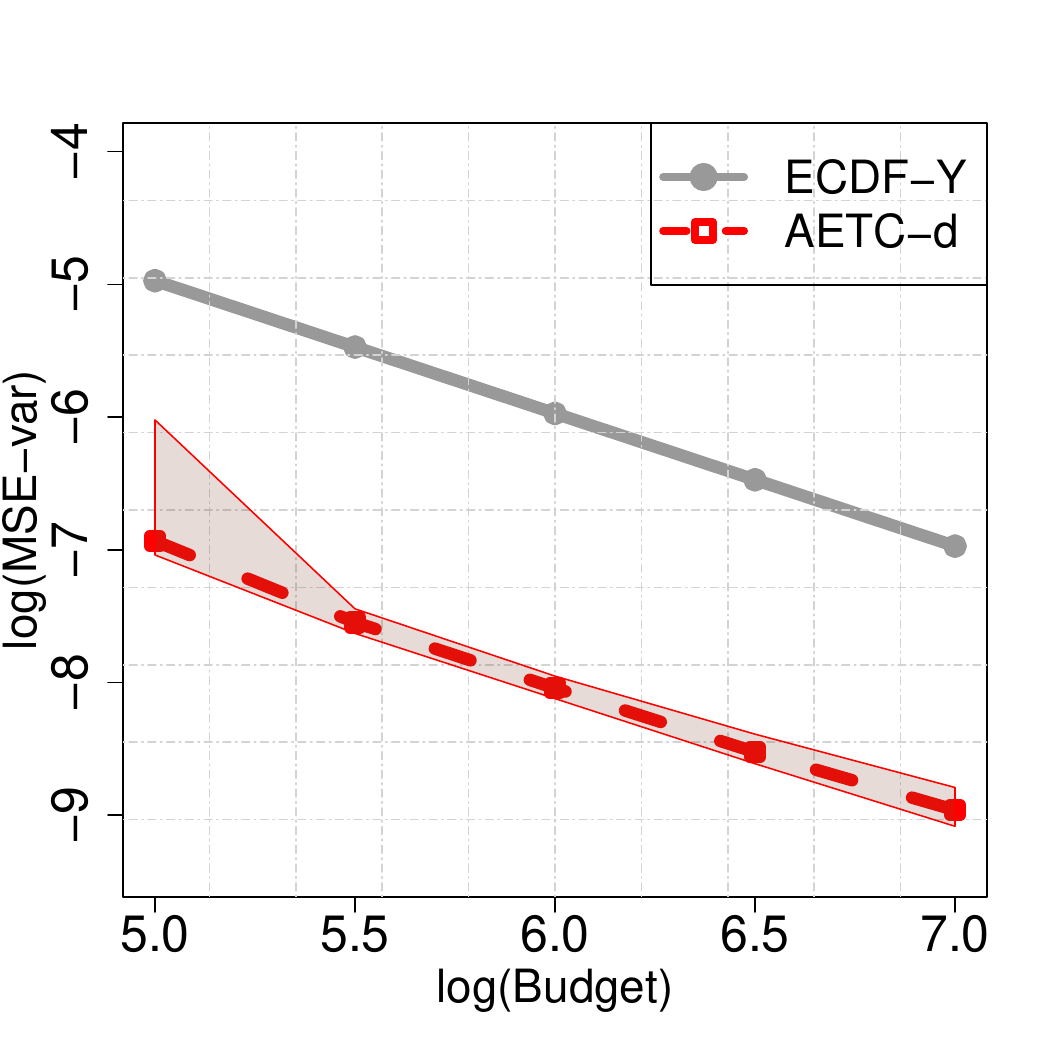}}\hfill
\subfigure[]{\includegraphics[width=0.25\linewidth]{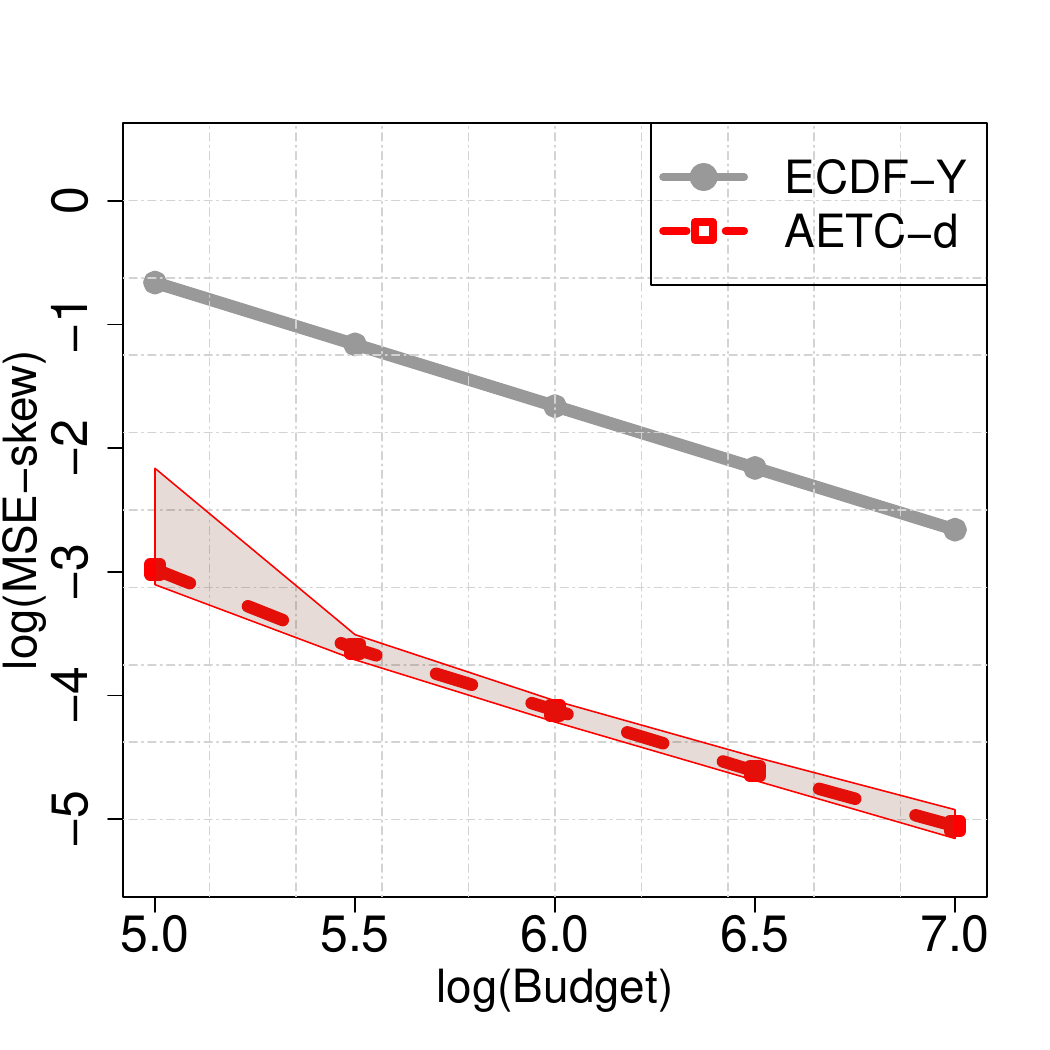}}\hfill
 \subfigure[]{\includegraphics[width=0.25\linewidth]{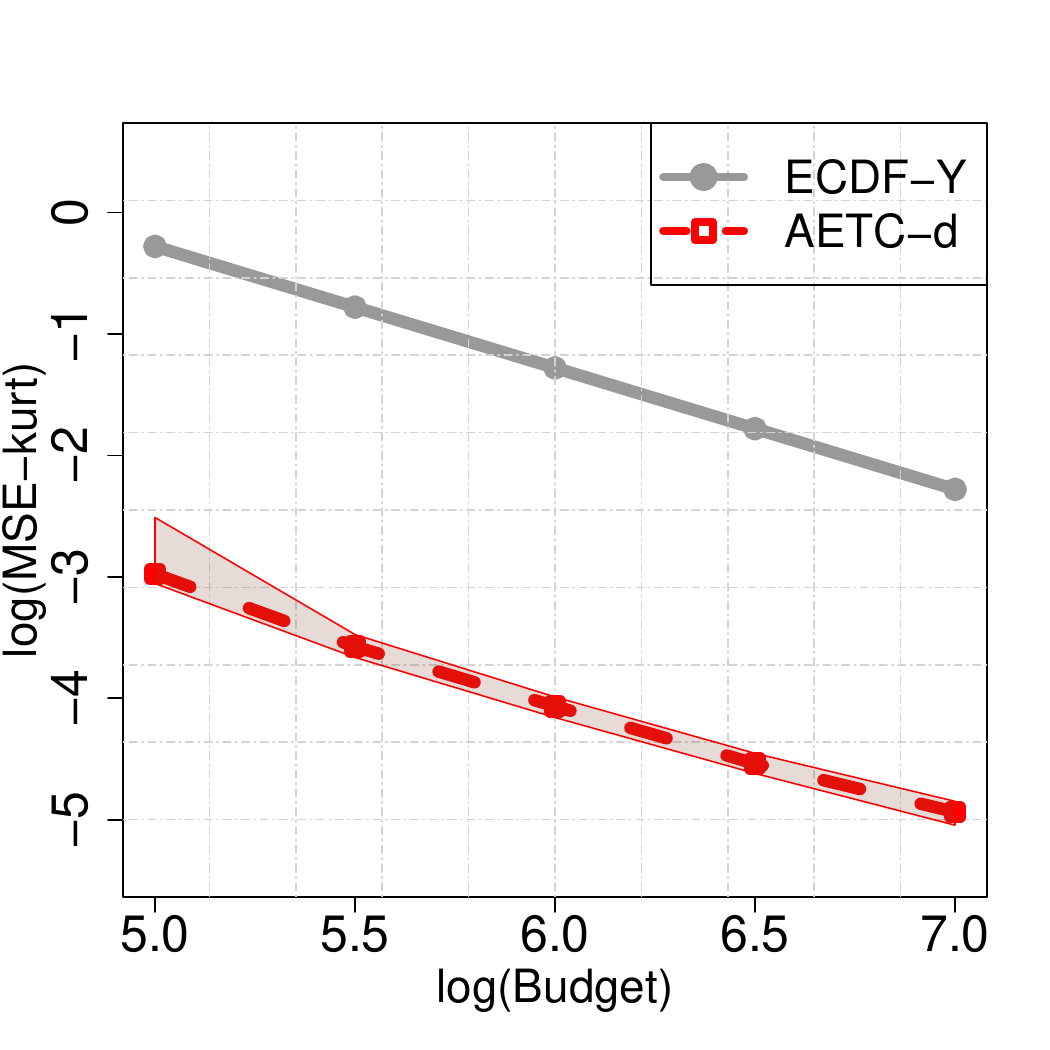}}\hfill
\caption{($\log_{10}$) mean $W_1$ distance between $F_Y$ and the estimated CDFs given by ECDF-Y and AETC-d. The $5\%$-$50\%$-$95\%$ quantiles are plotted for the AETC-d algorithm to measure its uncertainty in the exploration phase (a). An instance of the densities given by ECDF-Y and AETC-d when $B = 10^6$ (b). Comparison of the ($\log_{10}$) mean-squared errors (MSEs) of various estimated statistics by ECDF-Y and AETC-d (bottom); statistics from bottom left to right are the mean (c), variance (d), skewness (e), and kurtosis (f), respectively. } \label{fig4}
\end{figure}

Figure \ref{fig4} shows that within the budget range of this experiment, AETC-d is asymptotically consistent and outperforms ECDF-Y by a substantial margin. 
The statistics computed by AETC-d are accurate up to the fourth order and significantly more accurate than the ones given by simple MC. 
The estimated density given by AETC-d at $B = 10^6$ almost matches the density obtained from two million independent samples of $Y$, which approximately costs $10^{10}$ budget units (and required around a week to generate on our hardware).

We notice that model misspecification seems to have little impact in this experiment. 
A possible explanation is that the selected regressors (which coincide with the regressors of the optimal model $S_\opt = \{1,2,3\}$) are sufficiently expressible for $Y$, and the magnitude of the corresponding model variance is much smaller than the variance of $Y$, making the potential noise misspecification effect negligible under the budget range of this example. 
We substantiate this hypothesis by fitting a linear model using the complete training data, and we approximately compute the variance ratio between the model noise and $Y$, which is $3\times 10^{-5}\ll 1$.

\section{Conclusions}\label{concl}

In this paper, we introduce an efficient strategy for learning the distribution of a scalar-valued QoI in the multifidelity setup.
Under a linear model assumption, we propose a semi-parametric approach for approximating the distribution by leveraging samples of models of different resolutions/costs. 
The main novelty in our analysis is to provide an asymptotically informative and computationally estimable upper bound for the average $1$-Wasserstein distance between the estimator and the true distribution and use it to devise an adaptive algorithm, AETC-d, for efficient budget allocation.  
We show that, for a large budget, the AETC-d is consistent, and explores and exploits optimally under a proposed upper bound criterion. Our setup and algorithm require neither a model hierarchy nor an \textit{a priori} estimate of cross-model correlations.
We also discuss several approaches to mitigating the model misspecification impact when a linear regression assumption is violated. 

Distribution learning is a much harder problem compared to parameter/statistic estimation in general.
Our method takes an initial step towards addressing this problem under a linear model assumption and provides a potentially effective way to fully quantify the uncertainty of the QoI associated with a high-fidelity model. 
Our algorithm is robust when the model noise is relatively small. 
However, practical guidance to ameliorate model misspecification effects in practice is still lacking when the budget is sufficiently large. 
In future work, we will investigate alternative multifidelity CDF emulators that are less sensitive to the linear model assumptions and possess better robustness properties when the model is misspecified.

\section*{Acknowledgement}  
We would like to thank the referees for their time and helpful comments which significantly improved the presentation of the manuscript.
Y. Xu and A. Narayan are partially supported by National Science Foundation DMS-1848508. 
A. Narayan is partially supported by the Air Force Office of Scientific Research award FA9550-20-1-0338. 
Y. Xu would like to thank Dr. Xiaoou Pan for clarifying a uniform consistency result in quantile regression. We also thank Dr. Ruijian Han for a careful reading of an early draft, and for providing several comments that improved the presentation of the manuscript.

\section{Appendices}\label{sec:app}

\begin{appendices}

\section{Proof of Lemma \ref{lemma1}}\label{4.1l}

Let $|S| = s$. 
We first prove \eqref{bdd1}. 
Recall the definition of $Y$ and $Y'$:
\begin{align*}
&Y = X_S^\top \beta_S + \e_S & Y' = X_S^\top \widehat{\beta}_S + \widehat{\e}_S
\end{align*}
  where $\e_S, \widehat{\e}_S$ are independent of $X_S$ by Assumption \ref{ass:eS-independent}.

Define the true empirical distribution of $\e_S$ using exploration samples as 
\begin{align*}
\widetilde{\e}_S\sim \frac{1}{m}\sum_{\ell\in [m]}\delta_{Y_\ell-X_{S,\ell}^\top\beta_S}.
\end{align*}
By the additive property of $W_1$ metric under independence \cite{panaretos2019statistical} and the triangle inequality, we can upper bound the average $W_1$ distance between $F_Y$ and $F_{Y'}$ by conditioning on the exploration data: 
\begin{align}
&\E\left[W_1(F_Y, F_{Y'})|Z_S, Y_{\ex}\right]\nonumber\\
\leq&\  \E[W_1(F_{X_S^\top\beta_S}, F_{X_S^\top \widehat{\beta}_S})|Z_S, Y_{\ex}] +\E[W_1(F_{\e_S}, F_{\widehat{\e}_S})|Z_S, Y_{\ex}]\nonumber\\
\leq&\  \E[W_1(F_{X_S^\top\beta_S}, F_{X_S^\top \widehat{\beta}_S})|Z_S, Y_{\ex}] +\E[W_1(F_{\e_S}, F_{\widetilde{\e}_S})|Z_S, Y_{\ex}] +\E[W_1(F_{\widetilde{\e}_S}, F_{\widehat{\e}_S})|Z_S, Y_{\ex}]\label{0016}.
\end{align}

For the first term in \eqref{0016}, note $\bt_S-\beta_S$ satisfies
\begin{align}
\bt_S-\beta_S\sim\mathcal (0, \sigma_S^2(Z^\top _SZ_S)^{-1}).\label{hp}
\end{align}
Averaging out the randomness of exploration noise, we have that, almost surely, 
\begin{align}
\E[W_1(F_{X_S^\top\beta_S}, F_{X_S^\top \widehat{\beta}_S})|Z_S]&\leq \left(\E[W^2_2(F_{X_S^\top\beta_S}, F_{X_S^\top \widehat{\beta}_S})|Z_S]\right)^{1/2}\nonumber\\
&\leq\left(\E[|X_S^\top (\bt_S-\beta_S)|^2 |Z_S]\right)^{1/2}\nonumber\\
&=\left(\tr(\E[X_S X_S^\top] \E[(\bt_S-\beta_S)(\bt_S-\beta_S)^\top| Z_S])\right)^{1/2}\nonumber\\
& = \left(\frac{\sigma_S^2}{m}\tr(\Lambda_S(m^{-1}Z_S^\top Z_S)^{-1})\right)^{1/2}\nonumber\\
&\simeq \sqrt{\frac{s+1}{m}}\sigma_S,\label{f1}
\end{align}
where the first step uses Jensen's inequality, and the last step follows from the law of large numbers and Assumption \ref{a1}. 

For the second term in \eqref{0016}, note that $\widetilde{\e}_S$ is the empirical distribution of $\e_S$ based on $m$ exploration samples, which does not depend on $Z_S$ under Assumption \ref{ass:eS-independent}.  
Applying the nonasymptotic estimates on the convergence rate of empirical measures in Lemma \ref{BM} obtains
\begin{align}
\E[W_1\left(F_{\e_S}, F_{\widetilde{\e}_S}\right)|Z_S] = \E[W_1\left(F_{\e_S}, F_{\widetilde{\e}_S}\right)]\leq\frac{J_1(F_{\e_S})}{\sqrt{m}},\label{f2}
\end{align}
where $J_1$ is defined in \eqref{J}. 

For the third term in \eqref{0016}, consider the natural coupling between $\widehat{\e}_S$ and $\widetilde{\e}_S$: $\widetilde{\e}^{\leftarrow}_S(\widetilde{\tau}_\ell) = \widehat{\e}^{\leftarrow}_S(\widehat{\tau}_\ell)$, where $^{\leftarrow}$ denotes the preimage of a map and 
\begin{align}
&\widetilde{\tau}_\ell = Y_\ell-X_{S,\ell}^\top\beta_S&\widehat{\tau}_\ell = Y_\ell-X_{S,\ell}^\top\bt_S.\label{lllll}
\end{align}
In this case, 
\begin{align}
\E[W_1(F_{\widetilde{\e}_S}, F_{\widehat{\e}_S})|Z_S]\leq (\E[W^2_2(F_{\widetilde{\e}_S}, F_{\widehat{\e}_S})|Z_S])^{1/2}&\leq (\E[|\widetilde{\e}_S-\widehat{\e}_S|^2|Z_S])^{1/2}\nonumber\\
&=\left(\frac{1}{m}\sum_{\ell\in [m]}\E[(X^\top_{S,\ell}(\bt_S-\beta_S))^2|Z_S]\right)^{1/2}\nonumber\\
& = \sqrt{\frac{s+1}{m}}\sigma_S\label{f22}.
\end{align}
Putting \eqref{f1}, \eqref{f2}, \eqref{f22} together finishes the proof of \eqref{bdd1}.

We next prove \eqref{bdd2}. 
Conditioned on $Z_S$ and $Y_\ex$, $Y'$ is a random variable with bounded $r$-th moments for all $r>2$.
Appealing to Lemma \ref{BM} and averaging over the exploration noise, we have 
\begin{align}
\E\left[W_1\left(\widehat{F}_{Y,S}, F_{Y'}\right)|Z_S, Y_\ex\right]\leq\frac{\E [J_1(F_{Y'})|Z_S, Y_\ex]}{\sqrt{N_S}}\Longrightarrow \E\left[W_1\left(\widehat{F}_{Y,S}, F_{Y'}\right)|Z_S\right]\leq\frac{\E [J_1(F_{Y'})|Z_S]}{\sqrt{N_S}},\label{key1}
\end{align}
where $J_1$ is defined in \eqref{J}. 
The desired result would follow if we can show that $\E[J_1(F_{Y'})|Z_S]$ converges to $J_1(F_{Y})$ a.s. as $m\to\infty$. 
To this end, we introduce the following intermediate random variables:
\begin{align*}
&Y'' = X_S^\top\beta_S + \widetilde{\e}_S&Y''' = X_S^\top\beta_S + \widehat{\e}_S.
\end{align*}
We will prove the desired result by verifying the following convergence statements respectively: 
\begin{enumerate}
\item [(a).] $|\E[J_1(F_{Y})|Z_S]-\E[J_1(F_{Y''})|Z_S]|\to 0$ a.s.; 
\item [(b).] $|\E[J_1(F_{Y''})|Z_S]-\E[J_1(F_{Y'''})|Z_S]|\to 0$ a.s.;
\item [(c).] $|\E[J_1(F_{Y'''})|Z_S]-\E[J_1(F_{Y'})|Z_S]|\to 0$ a.s.
\end{enumerate} 

Without loss of generality, we assume $\text{supp}(\e_S)\subseteq [-1, 1]$ and $\|\beta_S\|_2 = 1$; the general case can be considered similarly by taking appropriate scaling involving a constant $C$ in Assumption \ref{a3}.

We introduce the following quantity for our analysis:
\begin{align}
K_m^* = \max_{\ell\in [m]}\|X_{S,\ell}\|_2.\label{K*}
\end{align} 
It is clear that $K_m^*$ depends only on $Z_S$. 
Under Assumption \ref{a2}, $X_{S,\ell}$'s are i.i.d. sub-exponential random variables with uniformly bounded sub-exponential norm. 
By Lemma \ref{expmax}, 
\begin{align}
&K_m^*\lesssim \log m&a.s.,\label{menghua}
\end{align}
where the implicit constant is realization-dependent.

To prove (a), we condition on $Z_S$ and $Y_\ex$.
Using \eqref{dongge}-\eqref{convo2} in Lemma \ref{cute}, 
\begin{align}
&|\E[J_1(F_{Y})|Z_S, Y_\ex]-\E[J_1(F_{Y''})|Z_S, Y_\ex]|\nonumber\\
\leq&\ \E\left[\int_\R \sqrt{| F_Y(y)-F_{Y''}(y)|} dy| Z_S, Y_\ex\right]\nonumber\\
=&\ \E\left[\int_\R\sqrt{\left|\int_\R F_{X_S^\top\beta_S}(y-z) dF_{\e_S}(z) - \int_\R F_{X_S^\top\beta_S}(y-z)dF_{\widetilde{\e}_S}(z)\right|} dy| Z_S, Y_\ex\right].\label{feier}
\end{align}

Under Assumption \ref{a2}, for every $y$, since $\|\beta_S\|_2 = 1$, $F_{X_S^\top\beta_S}(y-z)$ as a function of $z$, is $C_\li$-Lipschitz. 
By the Kantorovich-Rubinstein duality \eqref{dual}, 
\begin{align}
\left|\int_\R F_{X_S^\top\beta_S}(y-z) dF_{\e_S}(z) - \int_\R F_{X_S^\top\beta_S}(y-z)dF_{\widetilde{\e}_S}(z)\right|\leq C_\li W_1(\e_S, \widetilde{\e}_S).\label{610}
\end{align}
Meanwhile, note $\supp(\e_S) = \supp(\widetilde{\e}_S)\subseteq [-1,1]$. This combined with the fact that $X_S^\top \beta_S$ is sub-exponential implies that  
\begin{align}
&\left|\int_\R F_{X_S^\top\beta_S}(y-z) dF_{\e_S}(z) - \int_\R F_{X_S^\top\beta_S}(y-z)dF_{\widetilde{\e}_S}(z)\right|\nonumber\\
=&\ \left|\int_\R 1-F_{X_S^\top\beta_S}(y-z) dF_{\e_S}(z) - \int_\R 1-F_{X_S^\top\beta_S}(y-z)dF_{\widetilde{\e}_S}(z)\right|\nonumber\\
\leq&\ \frac{1}{2}\left(M_1[F_{X_S^\top\beta_S}](y) + M_1[1-F_{X_S^\top\beta_S}](y)\right)\nonumber\\
\leq&\  \exp\left(-\frac{\max\{|y-1|, |y+1|\}}{C}\right)\nonumber\\
\leq&\ \exp\left(-\frac{|y|}{2C}\right)&|y|\geq 2,\label{611}
\end{align}
where $C$ is an absolute constant depending only on the sub-exponential norm of $\|X_S\|_2$, and $M_1$ is the $1$-local maximum operator in Definition \ref{maxfun}. 
Substituting \eqref{610} and \eqref{611} into \eqref{feier} and applying a truncated estimate, 
\begin{align*}
&|\E[J_1(F_{Y})|Z_S, Y_\ex]-\E[J_1(F_{Y''})|Z_S, Y_\ex]|\nonumber\\
\leq&\  \E\left[\int_{|y|<\max\{2, 4C\log m\}}\sqrt{W_1(\e_S, \widetilde{\e}_S)} dy + \int_{|y|\geq \max\{2, 4C\log m\}}\exp\left(-\frac{|y|}{4C}\right)dy|Z_S, Y_\ex\right]\nonumber\\
\leq&\ \E\left[(4+8C\log m)\sqrt{W_1(\e_S, \widetilde{\e}_S)} + \frac{2}{m}|Z_S, Y_\ex\right].
\end{align*}
Since $\sqrt{W_1(\e_S, \widetilde{\e}_S)}$ is independent of $Z_S$, taking expectation over the exploration noise together with Jensen's inequality and Lemma \ref{BM} yields
\begin{align}
|\E[J_1(F_{Y})|Z_S]-\E[J_1(F_{Y''})|Z_S]|\leq (4+8C\log m)\E[W_1(\e_S, \widetilde{\e}_S)]^{1/2} + \frac{2}{m}\lesssim\frac{\log m}{\sqrt{m}}\to 0. 
\end{align}

To prove (b), note that conditioning on $Z_S$ and $Y_\ex$, the difference between $\widetilde{\tau}_\ell$ and $\widehat{\tau}_\ell$ is bounded as follows:
\begin{align}
&|\widetilde{\tau}_\ell-\widehat{\tau}_\ell|=|X_{S,\ell}^\top(\bt_S-\beta_S)|\leq \|X_{S,\ell}\|_2\|\bt_S-\beta_S\|_2\stackrel{\eqref{K*}}{\leq}K_m^*\delta&\delta = \|\bt_S-\beta_S\|_2,\label{lkng}
\end{align} 
where $\widehat{\tau}_\ell, \widetilde{\tau}_\ell$ are defined in \eqref{lllll}. 
Moreover, since $\text{supp}(\e_S)\subseteq [-1, 1]$, $|\widetilde{\tau}_\ell|\leq 1$.
This combined with \eqref{lkng} implies
\begin{align}
&\supp(\widetilde{\e}_S)\cup \supp(\widehat{\e}_S)\subseteq [-r, r]& r = 1 + K_m^*\delta.\label{major}
\end{align}
The rest is similar to the proof of statement (a),  
\begin{align*}
&|\E[J_1(F_{Y''})|Z_S, Y_\ex]-\E[J_1(F_{Y'''})|Z_S, Y_\ex]|\nonumber\\
\leq&\ \E\left[\int_\R\sqrt{|F_{Y''}(y)-F_{Y'''}(y)|}dy|Z_S, Y_\ex\right]\nonumber\\
\leq&\ \E\left[\int_\R\sqrt{\frac{1}{m}\sum_{\ell\in [m]}|F_{X_S^\top\beta_S}(y-\widehat{\tau}_\ell)-F_{X_S^\top\beta_S}(y-\widetilde{\tau}_\ell)|}dy|Z_S, Y_\ex\right].
\end{align*}
It is easy to verify using the Lipschitz assumption and the tail bound of $X_S^\top \beta_S$ that
\begin{align*}
\frac{1}{m}\sum_{\ell\in [m]}|F_{X_S^\top\beta_S}(y-\widehat{\tau}_\ell)-F_{X_S^\top\beta_S}(y-\widetilde{\tau}_\ell)|&\leq C_\li K_m^*\delta\\
\frac{1}{m}\sum_{\ell\in [m]}|F_{X_S^\top\beta_S}(y-\widehat{\tau}_\ell)-F_{X_S^\top\beta_S}(y-\widetilde{\tau}_\ell)|&\leq \exp\left(-\frac{|y|}{2C}\right)&|y|\geq 2r. 
\end{align*}
Thus,
\begin{align*}
&|\E[J_1(F_{Y''})|Z_S, Y_\ex]-\E[J_1(F_{Y'''})|Z_S, Y_\ex]|\\
\leq&\ \E\left[\int_{|y|<\max\{2r, 4C\log m\}}\sqrt{C_\li K_m^*\delta} dy + \int_{|y|\geq \max\{2r, 4C\log m\}}\exp\left(-\frac{|y|}{4C}\right)dy|Z_S, Y_\ex\right]\\
\leq&\ \E\left[(4r + 8C\log m)\sqrt{C_\li K_m^*\delta} +\frac{2}{m}|Z_S, Y_\ex\right].
\end{align*}
Averaging out exploration noise and applying Jensen's inequality,
\begin{align*}
|\E[J_1(F_{Y''})|Z_S]-\E[J_1(F_{Y'''})|Z_S]|&\lesssim \E\left[\left(K_m^*\delta\right)^{3/2} +\log m \sqrt{K_m^*\delta} + \frac{2}{m}|Z_S\right]\\
&\lesssim (K_m^*)^{3/2}\E[\delta^2|Z_S]^{3/4} + \log m\sqrt{K_m^*}\E[\delta^2|Z_S]^{1/4} + \frac{1}{m}\\
&\stackrel{\eqref{hp}, \eqref{menghua}}{\lesssim}\frac{(\log m)^{3/2}}{m^{1/4}}\to 0& a.s. 
\end{align*}

To prove (c), recall from \eqref{major} that conditioning on $Z_S$ and $Y_\ex$, $\supp(\widehat{\e}_S)\subseteq [-r ,r]$. 
Applying Lemma \ref{cute}, 
\begin{align}
|\E[J_1(F_{Y'''})|Z_S, Y_\ex]-\E[J_1(F_{Y'})|Z_S, Y_\ex]|\leq \E\left[\left\|M_{r}[|F_{X_S^\top\bt_S}-F_{X_S^\top\beta_S}|]\right\|^{1/2}_{L^{1/2}_\R}| Z_S, Y_\ex\right].\label{jh3}
\end{align}
If $\delta<1/2$, then $1/2<\|\bt_S\|_2< 3/2$. 
In this case, the $C_1, C_2, C_3$ in Lemma \ref{kouniao} are absolute constants. 
According to Lemma \ref{kouniao} with $p = 1/2$, 
\begin{align}
\left\|M_{r}[|F_{X_S^\top\bt_S}-F_{X_S^\top\beta_S}|]\right\|^{1/2}_{L^{1/2}_\R}\lesssim (r+ 1)\delta^{5/12}\log\left(1/\delta\right)\leq (r+1)\delta^{1/4}, \label{jh1}
\end{align}
where we used $\log (1/\delta)<\delta^{-1/6}$ when $\delta\leq 1/2$.

If $\delta\geq 1/2$, the same result in Lemma \ref{kouniao} implies
\begin{align}
\left\|M_{r}[|F_{X_S^\top\bt_S}-F_{X_S^\top\beta_S}|]\right\|^{1/2}_{L^{1/2}_\R}\lesssim (r + 1+\delta)\delta. \label{jh2}
\end{align}
Substituting \eqref{jh1} and \eqref{jh2} into \eqref{jh3} yields that 
\begin{align*}
|\E[J_1(F_{Y'''})|Z_S, Y_\ex]-\E[J_1(F_{Y'})|Z_S, Y_\ex]|\lesssim\E[(r+1)\delta^{1/4} + (r + 1+\delta)\delta | Z_S, Y_\ex].
\end{align*}
Taking expectation over the exploration noise and applying Jensen's inequality,
\begin{align*}
&|\E[J_1(F_{Y'''})|Z_S]-\E[J_1(F_{Y'})|Z_S]|\\
\lesssim&\ \E[\delta^{1/4}|Z_S] + \E[\delta|Z_S]+K_m^*\E[\delta^{5/4}|Z_S]+(K_m)^2\E[\delta^2|Z_S]\\
\lesssim&\ \E[\delta^2|Z_S]^{1/8} + \E[\delta^2|Z_S]^{1/2}+K_m^*\E[\delta^{2}|Z_S]^{5/8}+(K_m)^2\E[\delta^2|Z_S]\\
\stackrel{\eqref{hp}, \eqref{K*}}{\lesssim}&\ \frac{1}{m^{1/8}}\to 0& a.s.
\end{align*}
\eqref{bdd2} is proved by combining statements (a), (b), (c).

\section{A quantile regression framework}\label{app5}
Quantile regression offers an alternative approach to simulating $Y$ through a random coefficient interpretation \cite{Koenker}. 
For any $S\subseteq  [n]$ and $\tau\in (0,1)$, we assume the conditional $\tau$-th quantile of $Y$ on $X_S$ satisfies
\begin{align}
F^{-1}_{Y|X_S}(\tau) = X_S^\top \beta_S(\tau),\label{4}
\end{align}
where $\beta_S(\tau)$ the $\tau$-th coefficient vector. 
\eqref{4} is a standard quantile regression formulation, and can be used to model heteroscedastic noise effects. \begin{align*}
&\widehat{\beta}_S(\tau) = \argmin_{\beta\in\R^{s+1}}\frac{1}{m}\sum_{\ell\in [m]}\rho_\tau(Y_\ell - X^\top _{\ex,\ell}\beta)&\rho_\tau(x) = x(\tau - \bm{1}_{x<0}).
\end{align*}
Thus, \eqref{4} approximately equals
\begin{align}
\widehat{F}^{-1}_{Y|X_S}(\tau) = X_S^\top \widehat{\beta}_S(\tau).\label{44}
\end{align}
As opposed to \eqref{lr:est}, \eqref{44} provides a way to simulate $Y$ based on $X_S$ via inverse transform sampling:
\begin{align}
&Y \approx X_S^\top \widehat{\beta}_S(U).\label{qtl}&U\sim\text{Unif}(0,1)\independent X_S.
\end{align}
In our case, $X_{\ex, \ell}, \ell\in [m]$ are i.i.d. samples so \eqref{4} fits into a random design quantile regression framework as analyzed in \cite{Pan_2020}, where the authors established a strong consistency result for $\widehat{\beta}_S(\tau)$ under suitable conditions. 
The consistency result can be further proven to hold uniformly for all $\tau\in [\delta,1-\delta]$ for any fixed $\delta>0$, which justifies the asymptotic behavior of the procedure in \eqref{44} as $m, N_S\to\infty$.

In the quantile regression framework, obtaining the optimal choices for $m$ and $S$ is much harder than in the linear regression setup. 
The AETC-d-q algorithm in Section \ref{sec:num} implements \eqref{qtl} with $m$ set as the adaptive exploration rate given by the AETC-d, $S$ as the corresponding model output for exploitation, and $U$ approximated via $\frac{1}{K}\sum_{j\in [K]}\delta_{\frac{j}{K+1}}$ with $K=100$. 

\end{appendices}

\printbibliography

@phdthesis{farcas2020context,
  title={Context-aware model hierarchies for higher-dimensional uncertainty quantification},
  author={Farcas, Ionut-Gabriel},
  year={2020},
  school={Technische Universit{\"a}t M{\"u}nchen}
}

@article{peherstorfer2019multifidelity,
  title={Multifidelity Monte Carlo estimation with adaptive low-fidelity models},
  author={Peherstorfer, Benjamin},
  journal={SIAM/ASA Journal on Uncertainty Quantification},
  volume={7},
  number={2},
  pages={579--603},
  year={2019},
  publisher={SIAM}
}

@article{panaretos2019statistical,
  title={Statistical aspects of Wasserstein distances},
  author={Panaretos, Victor M and Zemel, Yoav},
  journal={Annual review of statistics and its application},
  volume={6},
  pages={405--431},
  year={2019},
  publisher={Annual Reviews}
}

@article{massart1990tight,
  title={The tight constant in the Dvoretzky-Kiefer-Wolfowitz inequality},
  author={Massart, Pascal},
  journal={The Annals of Probability},
  pages={1269--1283},
  year={1990},
  publisher={JSTOR}
}

@book{Vershynin_2018,
  title={High-dimensional probability: An introduction with applications in data science},
  author={Vershynin, Roman},
  volume={47},
  year={2018},
  publisher={Cambridge university press}
}

@book{lattimore2020bandit,
  title={Bandit algorithms},
  author={Lattimore, Tor and Szepesv{\'a}ri, Csaba},
  year={2020},
  publisher={Cambridge University Press}
}

@article{Peherstorfer_2016,
	doi = {10.1137/15m1046472},
	url = {https://doi.org/10.1137%2F15m1046472},
	year = 2016,
	month = {jan},
	publisher = {Society for Industrial {\&} Applied Mathematics ({SIAM})},
	volume = {38},
	number = {5},
	pages = {A3163--A3194},
	author = {Benjamin Peherstorfer and Karen Willcox and Max Gunzburger},
	title = {Optimal Model Management for Multifidelity Monte Carlo Estimation},
	journal = {{SIAM} Journal on Scientific Computing}
}

@article{Peherstorfer_2018_survey,
	year = 2018,
	publisher = {Society for Industrial {\&} Applied Mathematics ({SIAM})},
	volume = {60},
	number = {3},
	pages = {A550--A591},
	author = {Benjamin Peherstorfer and Karen Willcox and Max Gunzburger},
	title = {Survey of Multifidelity Methods in Uncertainty Propagation,
Inference, and Optimization},
	journal = {{SIAM} Review}
}

@book{Hammersley64,
  title={Monte Carlo Methods},
  author={J. M. Hammersley and D. C. Handscomb},
  year={1964},
  publisher={Methuen, London}
}

@article{bubeck2012regret,
  title={Regret Analysis of Stochastic and Nonstochastic Multi-armed Bandit Problems},
  author={Bubeck, S{\'e}bastien and Cesa-Bianchi, Nicol{\`o} and others},
  journal={Foundations and Trends{\textregistered} in Machine Learning},
  volume={5},
  number={1},
  pages={1--122},
  year={2012},
  publisher={Now Publishers, Inc.}
}

@article{Schaden_2020,
	doi = {10.1137/19m1263534},
	url = {https://doi.org/10.1137%2F19m1263534},
	year = 2020,
	month = {jan},
	publisher = {Society for Industrial {\&} Applied Mathematics ({SIAM})},
	volume = {8},
	number = {2},
	pages = {601--635},
	author = {Daniel Schaden and Elisabeth Ullmann},
	title = {On Multilevel Best Linear Unbiased Estimators},
	journal = {{SIAM}/{ASA} Journal on Uncertainty Quantification}
}

@article{schaden2020asymptotic,
  title={Asymptotic analysis of multilevel best linear unbiased estimators},
  author={Schaden, Daniel and Ullmann, Elisabeth},
  journal={SIAM/ASA Journal on Uncertainty Quantification},
  volume={9},
  number={3},
  pages={953--978},
  year={2021},
  publisher={SIAM}
}

@article{giles2008multilevel,
  title={Multilevel monte carlo path simulation},
  author={Giles, Michael B},
  journal={Operations research},
  volume={56},
  number={3},
  pages={607--617},
  year={2008},
  publisher={INFORMS}
}

@article{Gorodetsky_2020,
	doi = {10.1016/j.jcp.2020.109257},
	url = {https://doi.org/10.1016%2Fj.jcp.2020.109257},
	year = 2020,
	month = {may},
	publisher = {Elsevier {BV}},
	volume = {408},
	pages = {109257},
	author = {Alex A. Gorodetsky and Gianluca Geraci and Michael S. Eldred and John D. Jakeman},
	title = {A generalized approximate control variate framework for multifidelity uncertainty quantification},
	journal = {Journal of Computational Physics}
}

@article{chatterjee,
  title={Lecture notes on Stein's method},
  author={Chatterjee, Sourav},
  journal={Stanford lecture notes},
  year={2007}
}

@article{xu2021bandit,
  title={A bandit-learning approach to multifidelity approximation},
  author={Xu, Yiming and Keshavarzzadeh, Vahid and Kirby, Robert M and Narayan, Akil},
  journal={SIAM Journal on Scientific Computing},
  volume={44},
  number={1},
  pages={A150--A175},
  year={2022},
  publisher={SIAM}
}

@book{Bobkov_2019,
  title={One-dimensional empirical measures, order statistics, and Kantorovich transport distances},
  author={Bobkov, Sergey and Ledoux, Michel},
  volume={261},
  number={1259},
  year={2019},
  publisher={American Mathematical Society}
}

@incollection{Villani_2003,
	doi = {10.1090/gsm/058/08},
	url = {https://doi.org/10.1090%2Fgsm%2F058%2F08},
	year = 2003,
	month = {mar},
	publisher = {American Mathematical Society},
	pages = {205--235},
	author = {C{\'{e}}dric Villani},
	title = {The metric side of optimal transportation},
	booktitle = {Graduate Studies in Mathematics}
}

@article{cambanis1976inequalities,
  title={Inequalities for E k (x, y) when the marginals are fixed},
  author={Cambanis, Stamatis and Simons, Gordon and Stout, William},
  journal={Zeitschrift f{\"u}r Wahrscheinlichkeitstheorie und verwandte Gebiete},
  volume={36},
  number={4},
  pages={285--294},
  year={1976},
  publisher={Springer}
}

@article{Gibbs_2002,
	doi = {10.1111/j.1751-5823.2002.tb00178.x},
	url = {https://doi.org/10.1111%2Fj.1751-5823.2002.tb00178.x},
	year = 2002,
	month = {dec},
	publisher = {Wiley},
	volume = {70},
	number = {3},
	pages = {419--435},
	author = {Alison L. Gibbs and Francis Edward Su},
	title = {On Choosing and Bounding Probability Metrics},
	journal = {International Statistical Review}
}

@article{Lai_1982,
	doi = {10.1214/aos/1176345697},
	url = {https://doi.org/10.1214%2Faos%2F1176345697},
	year = 1982,
	month = {mar},
	publisher = {Institute of Mathematical Statistics},
	volume = {10},
	number = {1},
	pages = {154--166},
	author = {Tze Leung Lai and Ching Zong Wei},
	title = {Least Squares Estimates in Stochastic Regression Models with Applications to Identification and Control of Dynamic Systems},
	journal = {The Annals of Statistics}
}

@inproceedings{Ishigami,
	doi = {10.1109/isuma.1990.151285},
	url = {https://doi.org/10.1109%2Fisuma.1990.151285},
	year = 1990,
	publisher = {{IEEE} Comput. Soc. Press},
	author = {T. Ishigami and T. Homma},
	title = {An importance quantification technique in uncertainty analysis for computer models},
	booktitle = {[1990] Proceedings. First International Symposium on Uncertainty Modeling and Analysis}
}

@article{Qian_2018,
	doi = {10.1137/17m1151006},
	url = {https://doi.org/10.1137%2F17m1151006},
	year = 2018,
	month = {jan},
	publisher = {Society for Industrial {\&} Applied Mathematics ({SIAM})},
	volume = {6},
	number = {2},
	pages = {683--706},
	author = {E. Qian and B. Peherstorfer and D. O'Malley and V. V. Vesselinov and K. Willcox},
	title = {Multifidelity Monte Carlo Estimation of Variance and Sensitivity Indices},
	journal = {{SIAM}/{ASA} Journal on Uncertainty Quantification}
}

@article{beran1987convergence,
  title={Convergence of stochastic empirical measures},
  author={Beran, RJ and Le Cam, L and Millar, PW},
  journal={Journal of multivariate analysis},
  volume={23},
  number={1},
  pages={159--168},
  year={1987},
  publisher={Elsevier}
}

@Manual{R,
    title = {R: A Language and Environment for Statistical Computing},
    author = {{R Core Team}},
    organization = {R Foundation for Statistical Computing},
    address = {Vienna, Austria},
    year = {2020},
    url = {https://www.R-project.org/},
  }

@article{Pan_2020,
	doi = {10.1093/imaiai/iaaa006},
	url = {https://doi.org/10.1093%2Fimaiai%2Fiaaa006},
	year = 2020,
	month = {may},
	publisher = {Oxford University Press ({OUP})},
	author = {Xiaoou Pan and Wen-Xin Zhou},
	title = {Multiplier bootstrap for quantile regression: non-asymptotic theory under random design},
	journal = {Information and Inference: A Journal of the {IMA}}
}

@incollection{Koenker,
	doi = {10.1017/ccol0521845734.002},
	url = {https://doi.org/10.1017%2Fccol0521845734.002},
	publisher = {Cambridge University Press},
	pages = {26--67},
	author = {Roger Koenker},
	title = {Fundamentals of Quantile Regression},
	booktitle = {Quantile Regression}
}

@article{Giles_2015,
	doi = {10.1137/140960086},
	url = {https://doi.org/10.1137%2F140960086},
	year = 2015,
	month = {jan},
	publisher = {Society for Industrial {\&} Applied Mathematics ({SIAM})},
	volume = {3},
	number = {1},
	pages = {267--295},
	author = {Michael B. Giles and Tigran Nagapetyan and Klaus Ritter},
	title = {Multilevel Monte Carlo Approximation of Distribution Functions and Densities},
	journal = {{SIAM}/{ASA} Journal on Uncertainty Quantification}
}

@article{giles2017adaptive,
  title={Adaptive multilevel Monte Carlo approximation of distribution functions},
  author={Giles, Mike B and Nagapetyan, Tigran and Ritter, Klaus},
  journal={arXiv preprint arXiv:1706.06869},
  year={2017}
}

@article{Krumscheid_2018,
	doi = {10.1137/17m1135566},
	url = {https://doi.org/10.1137%2F17m1135566},
	year = 2018,
	month = {jan},
	publisher = {Society for Industrial {\&} Applied Mathematics ({SIAM})},
	volume = {6},
	number = {3},
	pages = {1256--1293},
	author = {S. Krumscheid and F. Nobile},
	title = {Multilevel Monte Carlo Approximation of Functions},
	journal = {{SIAM}/{ASA} Journal on Uncertainty Quantification}
}

@article{Lu_2016,
	doi = {10.1002/2016wr019475},
	url = {https://doi.org/10.1002%2F2016wr019475},
	year = 2016,
	month = {dec},
	publisher = {American Geophysical Union ({AGU})},
	volume = {52},
	number = {12},
	pages = {9642--9660},
	author = {Dan Lu and Guannan Zhang and Clayton Webster and Charlotte Barbier},
	title = {An improved multilevel Monte Carlo method for estimating probability distribution functions in stochastic oil reservoir simulations},
	journal = {Water Resources Research}
}

@book{williams1991probability,
  title={Probability with martingales},
  author={Williams, David},
  year={1991},
  publisher={Cambridge university press}
}

@article{Robinson_1988,
	doi = {10.2307/1912705},
	url = {https://doi.org/10.2307%2F1912705},
	year = 1988,
	month = {jul},
	publisher = {{JSTOR}},
	volume = {56},
	number = {4},
	pages = {931},
	author = {P. M. Robinson},
	title = {Root-N-Consistent Semiparametric Regression},
	journal = {Econometrica}
}

@article{cohen_approximation_2015,
	title = {Approximation of high-dimensional parametric {PDEs}},
	volume = {24},
	issn = {1474-0508},
	doi = {10.1017/S0962492915000033},
	journal = {Acta Numerica},
	author = {Cohen, Albert and DeVore, Ronald},
	year = {2015},
	pages = {1--159},
}
  
\end{document}